\newif\ifsiamvar

\siamvarfalse

\ifsiamvar
	\documentclass[review]{siamonline1116}
	
\else
	\documentclass[twoside,british,english,11pt]{article}

\fi

\newif\ifarxivvar
\arxivvartrue

\usepackage{etex}

\usepackage{marvosym}
\usepackage[T1]{fontenc}
\usepackage[latin9]{inputenc}
\usepackage[a4paper]{geometry}
\ifsiamvar
\else
	\usepackage{amsmath}
	\usepackage{amsthm}
\fi
\usepackage{verbatim}
\usepackage{float}
\usepackage{booktabs}
\usepackage{mathtools}
\usepackage{amsfonts}
\usepackage{amssymb}
\usepackage{stmaryrd}
\usepackage{graphicx}
\usepackage{esint}

\makeatletter

\providecommand{\tabularnewline}{\\}
\floatstyle{ruled}
\newfloat{algorithm}{tbp}{loa}
\providecommand{\algorithmname}{Algorithm}
\floatname{algorithm}{\protect\algorithmname}

\numberwithin{equation}{section}
\numberwithin{figure}{section}

\ifsiamvar
\else
	\theoremstyle{plain}
	\newtheorem{theorem}{Theorem}
	  \theoremstyle{remark}
	  \newtheorem*{remark}{Remark}
	  \theoremstyle{plain}
	  \newtheorem{lemma}[theorem]{Lemma}  
	  \theoremstyle{plain}
	  \newtheorem{proposition}[theorem]{Proposition}
	  \theoremstyle{plain}
	  \newtheorem{corollary}[theorem]{Corollary}
	  \theoremstyle{theorem}
	  \newtheorem{assumption}[theorem]{Assumption}
\fi
\usepackage{algolyx}

\makeatother

\usepackage{eurofont}
\usepackage{graphicx}
\usepackage{epstopdf}
\usepackage{hyperref}
\usepackage[T1]{fontenc}
\usepackage{ae,aecompl,aeguill}
\usepackage{newfloat}
\usepackage{caption}
\usepackage{pdflscape}
\usepackage{mathtools}
\usepackage{tikz}
\usepackage{pgfplots}
\usepackage{stmaryrd}
\usepackage{microtype}
\usepackage[justification=justified,singlelinecheck=false,font=small,labelfont=bf]{caption}
\usepackage{subcaption}
\usepackage{booktabs} 
\usepackage{tabularx}
\newcolumntype{Y}{>{\centering\arraybackslash}X} 
\allowdisplaybreaks[4]
\usepackage{dsfont}
\usepackage{bold-extra}
\usepackage{cases}
\usepackage{enumerate}
\usepackage{verbatim}
\usepackage{xcolor}
\usepackage{mathtools}

\allowdisplaybreaks[4]
\usepackage{hyperref}
\hypersetup{pdfpagemode=UseNone} 
\hypersetup{pdfstartview={XYZ null null 0.85}} 

\DeclareMathOperator{\E}{\mathbb{E}}

\DeclareMathOperator{\sgn}{sgn}

\newcommand{\gd}{d}
\newcommand{\gf}{f}
\newcommand{\gi}{i}

\usepackage{tabularx}

\usetikzlibrary{calc}
\tikzstyle{decision} = [diamond, draw, fill=blue!20, 
    text width=4.5em, text badly centered, node distance=3cm, inner sep=0pt]
\tikzstyle{block} = [rectangle, draw, fill=blue!20, 
    text width=10em, text centered, rounded corners, minimum height=4em, text centered,font=\bf]
\tikzstyle{output} = [rectangle, draw, fill=green!20, 
    text width=10em, text centered, rounded corners, minimum height=4em, text centered,font=\bf]
\tikzstyle{line} = [draw, -latex']
\tikzstyle{cloud} = [draw, ellipse,fill=black!10,node distance=5cm,
    minimum height=3em,text centered,font=\bf]

\author{
\scshape{Andrei Cozma\,}\thanks{\footnotesize\scshape{Mathematical Institute, University of Oxford, OX2 6GG, United Kingdom} \newline
\hspace*{1.8em}\textup{andrei.s.cozma@gmail.com, matthieu.mariapragassam@gmail.com, christoph.reisinger@maths.ox.ac.uk\newline
The first author gratefully acknowledges financial support from the \textsc{EPSRC}.
The second author gratefully acknowledges financial support from the \textsc{Oxford--Man
Institute} and \textsc{BNP Paribas London}.\newline The authors thank three anonymous reviewers for their insightful suggestions and comments.}}
\and
\scshape{Matthieu Mariapragassam\,\footnotemark[1]\;\,\thanks{\footnotesize\scshape{Oxford-Man Institute of Quantitative Finance, University of Oxford, OX2 6ED, United Kingdom}}\ifsiamvar \else \, \Letter \fi}
\and \scshape{Christoph Reisinger\,\footnotemark[1]\;\,\footnotemark[2]}
}

\date{}
\allowdisplaybreaks[0]

\makeatother

\ifsiamvar
\else
\fi

\ifsiamvar
	\newsiamremark{remark}{Remark}
	\newsiamremark{assumption}{Assumption}
\fi
\begin{document}
\title{Calibration of a Hybrid Local-Stochastic Volatility \\ Stochastic Rates Model with a Control Variate Particle Method}
\maketitle

\begin{abstract}
We propose a novel and generic calibration technique for four-factor
foreign-exchange hybrid local-stochastic volatility models (LSV) with stochastic short rates. We build
upon the particle method introduced by Guyon and Henry-Labordère 
[Nonlinear Option Pricing, Chapter 11, Chapman and Hall, 2013]
and combine it with new variance reduction techniques in order to
accelerate convergence. We use control variates derived from: a calibrated pure local volatility model;
a two-factor Heston-type LSV model (both with deterministic rates); the stochastic (CIR) short rates.
The method can be applied to a large class of hybrid LSV models and is
not restricted to our particular choice of the diffusion. However, we address in the paper some specific difficulties
arising from the Heston model, notably by a new PDE formulation and finite element solution to bypass the singularities of the density when zero
is attainable by the variance.
The calibration procedure is performed on market data for the $\text{EUR-USD}$
currency pair  and has a comparable run-time to the PDE calibration of a two-factor LSV model alone. 
\end{abstract}

\section{Introduction}

Efficient pricing and hedging of exotic derivatives requires a model
which is rich enough to re-price accurately a range of liquidly traded
market products. 
Calibration to vanilla options has been
widely documented 
in the literature since the work of Dupire \cite{Dupire} in the context of local volatility (LV).
Nowadays, the
exact re-pricing of call options is a must-have standard, and \emph{Local-Stochastic
Volatility} (LSV) models are the \emph{state-of-the-art} in many financial
institutions.
As discussed in Ren et al.\ \cite{Ren&Madan2007}, Tian et al. \cite{Tian2015}, Van der Stoep et al. \cite{Stoep2014} and
Guyon and Henry-Labordère \cite{Guyon2011}, LSV models
improve the
pricing and risk-management performance when compared to pure local
volatility or pure stochastic volatility models.
The local volatility component allows
a perfect calibration to the market prices of vanilla options. At
the same time, the stochastic volatility component already provides
built-in smiles and skews which give a rough fit, so that
a local volatility component --  the so-called leverage function -- relatively close to one suffices for a perfect
calibration.
Moreover, they exhibit superior dynamic properties over pure
local volatility models.

We focus on a Heston-type
LSV model because %
of the desirable properties of the Cox--Ingersoll--Ross (CIR) process for the variance, such
as mean-reversion and non-negativity, and since semi-analytic formulae
are available for calls and puts under Heston's model  (see \cite{Heston1993a}) and can help
calibrate the Heston parameters easily.
Various sophisticated calibration techniques for the local volatility component
are in use in the financial industry, e.g., based on the Monte Carlo particle method
in \cite{Guyon2011} or the PDE-based approach in
\cite{Ren&Madan2007}. 

In order to improve the pricing and hedging of foreign exchange (FX)
options, we furthermore introduce stochastic domestic and foreign short interest
rates into the model.
Empirical results (see e.g.\ \cite{Haastrecht2009})
have confirmed that 
for long-dated FX products 
the effect of interest rate volatility
can be as relevant as that of the FX rate volatility. 
Extensive research has been carried out in the area of
option pricing with stochastic volatility and interest rates in the
past few years. Van Haastrecht et al. \cite{Haastrecht2009}
extended the model of Schöbel and Zhu \cite{Schobel1999} to currency
derivatives by including stochastic interest rates, a model that benefits
from analytical tractability even in a full correlation setting due
to the processes being Gaussian. On the other hand, Ahlip and Rutkowski
\cite{Ahlip2013}, Grzelak and Oosterlee \cite{Grzelak2011} and Van
Haastrecht and Pelsser \cite{Haastrecht2011} examined Heston--CIR/Vasicek
hybrid models and concluded that they give
rise to non-affine models even under a partial correlation structure of the
driving Brownian motions and are not analytically tractable.

The resulting 4-factor
model 
complicates the calibration routine due to the higher dimensionality, especially when PDEs are used to find the joint
distribution of all factors. 
A few papers
discuss this problem in simpler settings. Deelstra \cite{Deelstra2013}
and Clark \cite{Clark2010} mainly consider 3-factor hybrid local
volatility models and focus on the theoretical rather than the practical
aspects of the calibration, whereas Van der Stoep et al.\ 
\cite{VanderStoep2016} consider an application to a 2-factor hybrid
local volatility. In \cite{Guyon2011}, Guyon and Henry-Labordère discuss an
application of Monte Carlo-based calibration methods to a 3-factor
LSV equity model with stochastic domestic rate and discrete dividends.

The model of Cox et al.~\cite{Cox1985} is popular when modeling
short rates because the (square-root) CIR process admits a unique
strong solution, is mean-reverting and analytically tractable. 
As of late, the non-negativity of the CIR process is considered
to be less desirable when modeling short rates. On one hand,
central banks have significantly reduced the interest rates since
the 2008 financial crisis and it is now commonly accepted that interest
rates need not be positive. On the other hand, if interest rates dropped
too far below zero, then large amounts of money would be withdrawn
from banks and government bonds, putting a severe squeeze on deposits.
Hence, we model the domestic and foreign short rates using the shifted
CIR (CIR\scalebox{.9}{\raisebox{.5pt}{++}}) process of Brigo and
Mercurio \cite{Brigo2001}. The CIR\scalebox{.9}{\raisebox{.5pt}{++}}
model allows the short rates to become negative and can fit any observed
term structure exactly while preserving the analytical tractability
of the original model for bonds, caps, swaptions and other basic interest
rate products.

We note that the CIR process is sometimes considered difficult to simulate in practice.
Moreover, as factor in the Heston model, it leads to singular probability densities
for parameter settings where the variance process can hit zero (i.e., if the so-called
Feller condition is violated),
which cannot be handled easily in the forward Kolmogorov equation by standard
numerical methods. In this paper, we address
both these issues by tailored schemes, but note that the variance and interest rate
processes can be exchanged without significant changes to the main framework
and its benefits, e.g.\ by exponential Ornstein-Uhlenbeck processes for the volatility and
Hull-White processes for the rates, both of which are also popular in the industry.


Based on the above considerations, we study the 4-factor hybrid LSV model defined in
(\ref{eq:Model Definition}) below, which is a Heston-type LSV model with two
shifted CIR short-rate processes.
We give a rigorous proof of the calibration condition for the leverage function given in
\cite[Proposition 12.8]{GuyonLabordere2013} 
for our model specification;
see also the condition given in \cite{Deelstra2013} for a 
4-factor LSV--2CIR++ model.


We propose a calibration approach
which builds on the particle method of \cite{Guyon2011},
and combines
it with a novel and efficient variance reduction technique.
The main
control variate 
is the two-factor LSV model obtained by assuming that the domestic and foreign
rates are deterministic in the original model. In this case, the leverage function is computed
by using a deterministic PDE solver.
This allows us to take
advantage of the efficiency and accuracy of PDE calibration for a low-dimensional model
while keeping the complexity for the high-dimensional model under control by Monte Carlo sampling
with drastically reduced variance.
We find that around 1000 particles are sufficient in practice.
Our numerical experiments suggest that this method recovers
the calibration speed from the corresponding 2-factor LSV model
with deterministic rates defined in (\ref{eq:modelDefinitionLSVDetRates}).


As a result of independent interest, we explain how to effectively deal with violation of the
Feller condition for the Heston-type LSV Kolmogorov forward 
equation and numerically solve the PDE using a
finite element method with a \emph{Backward Differentiation Formula}
(BDF) time-stepping scheme and an appropriate non-Dirichlet boundary
condition. 
To the best of our knowledge, this represents a new approach
which complements the literature on the
use of ADI schemes \cite{Clark2010,Ren&Madan2007,Wyns2016FV} to handle the PDE
calibration of an LSV model with deterministic rates. 

For Heston type models, the
CIR variance process can reach zero if the Feller condition
is violated, as  is often the case in FX markets
(we refer to Table 6.5 in \cite{Clark2010}
for examples on a large range of currency pairs and maturities).
As a consequence, 
the density is singular at the boundary $V=0$.
In \cite{Tian2015},
the authors propose to reduce the problem by considering $\log\left(V_{t}/v_{0}\right)$,
whereas \cite{Clark2010} suggests to refine the mesh near $V=0$.
While these methods alleviate the problem to some extent,
we propose to use a different boundary condition as well as a change
of variables which results in a bounded solution in a neighbourhood of $V=0$.

Moreover, a
main advantage of the finite element method compared to ADI schemes,
besides the greater flexibility in the mesh construction,
is that the Dirac delta initial condition can be handled naturally in
the weak formulation. This methodology yields an accurate calibration
of the Heston-type LSV model with deterministic rates for a broad set of market data.


Finally, we provide empirical evidence that the inclusion of stochastic rates is 
important for the pricing of some specific exotic derivatives.
In particular, in Section \ref{sec:pricing} we consider the pricing problem for a \emph{Target Accrual Redemption Note} (TARN) and a no-touch option.
We demonstrate that the impact of stochastic rates is comparable to the difference between pricing a 5-year no-touch option under a LV or LSV model.
Other exotics with similar features, not considered here, are Accumulators and \emph{Power Reverse Dual-Currency notes} (PRDC).
Moreover, stochastic rates become necessary for any hybrid product
which embeds the rates explicitly. Examples are spread options between
an FX rate and the Libor rate.

The remainder of this paper is organised as follows. In Section \ref{sec:model}, 
we specify the model and calibration framework and
provide a necessary and sufficient condition for a perfect
calibration to vanilla quotes.
A 
rigorous proof 
emphasising the use of local times and possible moment explosions is given in Appendix \ref{app:proof}.
In Section \ref{sec:Fast-Calibration},
we introduce the particle method used and detail how the control variates for both conditional expectations and 
standard expectations are constructed.
In Section \ref{sub:2-Factor-Heston-LSV-Calibration}, we describe the calibration of
the LSV model with deterministic rates using a carefully constructed finite element method.
In Section \ref{sec:Calibration-Results},
we present numerical results and show that a low number of particles
suffices to provide a very good fit to market quotes, which demonstrates the computational efficiency of the method.
The impact of stochastic rates for the pricing of a TARN and no-touch option is presented.
Section \ref{sec:Conclusion} concludes
with a brief discussion.

\section{Model definition and calibration}
\label{sec:model}

We consider a domestic and a foreign market with stochastic short rates
$r^{d}$ and $r^{f}$, and exchange rate $S$.
The spot $S_{T}$ is associated with the currency pair $\text{ccy1ccy2}$
(following the notations in \cite{Clark2010}) and denotes the amount
of units of $\text{ccy2}$ (domestic currency) needed to buy one unit
of $\text{ccy1}$ (foreign currency) at time $T$. 
We denote by $D^{d}$ and $D^{f}$ the domestic
and foreign discount factors associated with their respective money
market accounts,
\[
D_{t}^{d}=\text{e}^{-\int_{0}^{t}r_{u}^{d}du},\quad D_{t}^{f}=\text{e}^{-\int_{0}^{t}r_{u}^{f}du}\,.
\]

\subsection{Models}

We assume the existence of a filtered probability space ({\Large{}$\chi$},
$\mathcal{F},\left\{ \mathcal{F}_{t}\right\} _{t\geq0},\mathbb{Q}^{d})$
with a domestic risk-neutral measure $\mathbb{Q}^{d}$.
For future reference we also define a foreign risk-neutral measure $\mathbb{Q}^{f}$.
Under $\mathbb{Q}^{d}$,
$S$, $r^d$ and $r^f$ follow a system of SDEs
\begin{equation}
\begin{array}{c}
\begin{cases}
\cfrac{dS_{t}}{S_{t}}=\left(r_{t}^{d}-r_{t}^{f}\right)\,dt+\alpha\left(S_{t},t\right)\sqrt{V_{t}}\,dW_{t}\\
r_{t}^{d}=g_{t}^{d}+h^{d}\left(t\right)\\
r_{t}^{f}=g_{t}^{f}+h^{f}\left(t\right)\\
dg_{t}^{d}=\kappa_{d}\left(\theta_{d}-g_{t}^{d}\right)\,dt+\xi_{d}\sqrt{g_{t}^{d}}\,dW_{t}^{\gd}\\
dg_{t}^{f}=\left(\kappa_{f}\left(\theta_{f}-g_{t}^{f}\right)-\rho_{Sf}\xi_{f}\sqrt{g_{t}^{f}}\alpha\left(S_{t},t\right)\sqrt{V_{t}}\right)\,dt+\xi_{f}\sqrt{g_{t}^{f}}\,dW_{t}^{\gf}\\
dV_{t}=\kappa\left(\theta-V_{t}\right)\,dt+\xi\sqrt{V_{t}}\,dW_{t}^{V},
\end{cases}\end{array}\label{eq:Model Definition}
\end{equation}
where $V$ is the stochastic variance process and the four-dimensional standard Brownian motion $(W,W^{V},W^{\gd},W^{\gf})$ has the correlation structure
\begin{eqnarray*}
d\langle W_{t},W_{t}^{V} \rangle \; = \; \rho \, dt, \quad
d\langle W_{t},W_{t}^{\gd}  \rangle \; = \; \rho_{Sd} \, dt, \quad
d\langle W_{t},W_{t}^{\gf}  \rangle \; = \; \rho_{Sf} \, dt, \quad
d\langle W_{t}^{\gd},W_{t}^{\gf}  \rangle \; = \; \rho_{df} \, dt,
\end{eqnarray*}
with $\rho, \rho_{Sd}, \rho_{Sf}, \rho_{df} \in (-1,1)$,
the other correlations being zero (and such that the correlation matrix is positive definite),
and for given functions
$\alpha: \mathbb{R}^+ \times [0,T] \rightarrow \mathbb{R^+}$,
$h^{d/f}: [0,T] \rightarrow \mathbb{R}$, and non-negative numbers
$\kappa, \theta, \xi, \kappa_d, \theta_d, \xi_d, \kappa_f, \theta_f, \xi_f$,
as well as initial values $S_0, g_0^d, g_0^f, V_0$.

Let the call option price under model (\ref{eq:Model Definition})
for a notional of one unit of $ccy1$, with strike $K>0$ and maturity $T>0$, be 
\[
C\left(K,T\right)=\mathbb{E}^{\mathbb{Q}^{d}}\left[D_{T}^{d}\left(S_{T}-K\right)^{+}\right]\,.
\]

If the leverage function $\alpha \equiv1$ in (\ref{eq:Model Definition}), we recover a Heston model with shifted CIR
domestic and foreign short rates. We will refer to this model as Heston-2CIR++ model.
As this model will only be used for intermediate calibration steps, we will make the additional simplification that the
interest rate dynamics are independent of the dynamics of the spot FX rate and the variance process, for analytical tractability
(see \cite{Ahlip2013}).

We also define two simpler models which we will refer
to in the remainder of the article. 
In both these models, rates are deterministic,
$\bar{r}^{d}\left(t\right)=-\partial\ln P^{d}\left(0,t\right)/\partial t$
and 
$\bar{r}^{f}\left(t\right)=-\partial\ln P^{f}\left(0,t\right)/\partial t$,
with $P^{d/f}\left(0,T\right)$ the market zero coupon bond prices
for the domestic and foreign money market accounts, respectively.

We can then write the related
2-factor Heston-type LSV model with deterministic rates as
\begin{equation}
\begin{array}{c}
\begin{cases}
\cfrac{dS_{t}^{2D}}{S_{t}^{2D}}=\left(\bar{r}^{d}\left(t\right)-\bar{r}^{f}\left(t\right)\right)\,dt+\alpha^{2D}\left(S_{t}^{2D},t\right)\sqrt{V_{t}^{2D}}\,dW_{t},\quad S_{0}^{2D}=S_{0},\\
dV_{t}^{2D}=\kappa\left(\theta-V_{t}^{2D}\right)\,dt+\xi\sqrt{V_{t}^{2D}}\,dW_{t}^{V}, \quad V_{0}^{2D}=V_{0},
\end{cases}\end{array}\label{eq:modelDefinitionLSVDetRates}
\end{equation}
for a given function $\alpha^{2D}:\,\mathbb{R}^{+}\times\left[0,T\right]\rightarrow\mathbb{R}^{+}$,
and the pure \emph{Local Volatility} (LV) model{ as 
\begin{equation}
\frac{dS_{t}^{LV}}{S_{t}^{LV}}=\left(\bar{r}^{d}\left(t\right)-\bar{r}^{f}\left(t\right)\right)\,dt+\sigma_{LV}\left(S_{t}^{LV},t\right)\,dW_{t},\quad S_{0}^{LV}=S_{0},\label{eq:modelDefinitionLocalVolatility}
\end{equation}
with a given function $\sigma_{LV}:\,\mathbb{R}^{+}\times\left[0,T\right]\rightarrow\mathbb{R}^{+}$.

Note that while the volatility is ``local'', i.e., a function of spot FX and time, the short rates are assumed to be a function of time only.
We also note for future reference that
under the pure LV model (\ref{eq:modelDefinitionLocalVolatility}), call prices $C_{LV}$ satisfy
 the forward Dupire PDE (see \cite{Dupire})
\begin{equation}
\frac{\partial C_{LV}}{\partial T}-\frac{1}{2}\sigma_{LV}\left(K,T\right)^{2}K^{2}\frac{\partial^{2}C_{LV}}{\partial K^{2}}+K\left(\bar{r}^{d}\left(T\right)-\bar{r}^{f}\left(T\right)\right)\frac{\partial C_{LV}}{\partial K} +\bar{r}^{f}\left(T\right)C_{LV} =0\,.
\label{eq:dupirePDE}
\end{equation}


\subsection{Calibration outline}
The purpose of this paper is to calibrate $h^{d}$, $h^{f}$, $\kappa_{d}$,
$\kappa_{f}$, $\theta_{d}$, $\theta_{f}$, $\xi_{d}$, $\xi_{f}$,
$\kappa$, $\theta$, $\xi$, $\rho$ and especially $\alpha$ in (\ref{eq:Model Definition}).
We will use calibration of (\ref{eq:modelDefinitionLSVDetRates})
and (\ref{eq:modelDefinitionLocalVolatility}) as ``stepping stones''.
More precisely, the full calibration process consists of the following steps,
illustrated in
Figure \ref{fig:Calibration-routine-Flowchart}.

\begin{figure}[h]
\begin{centering}
\includegraphics[scale=0.75]{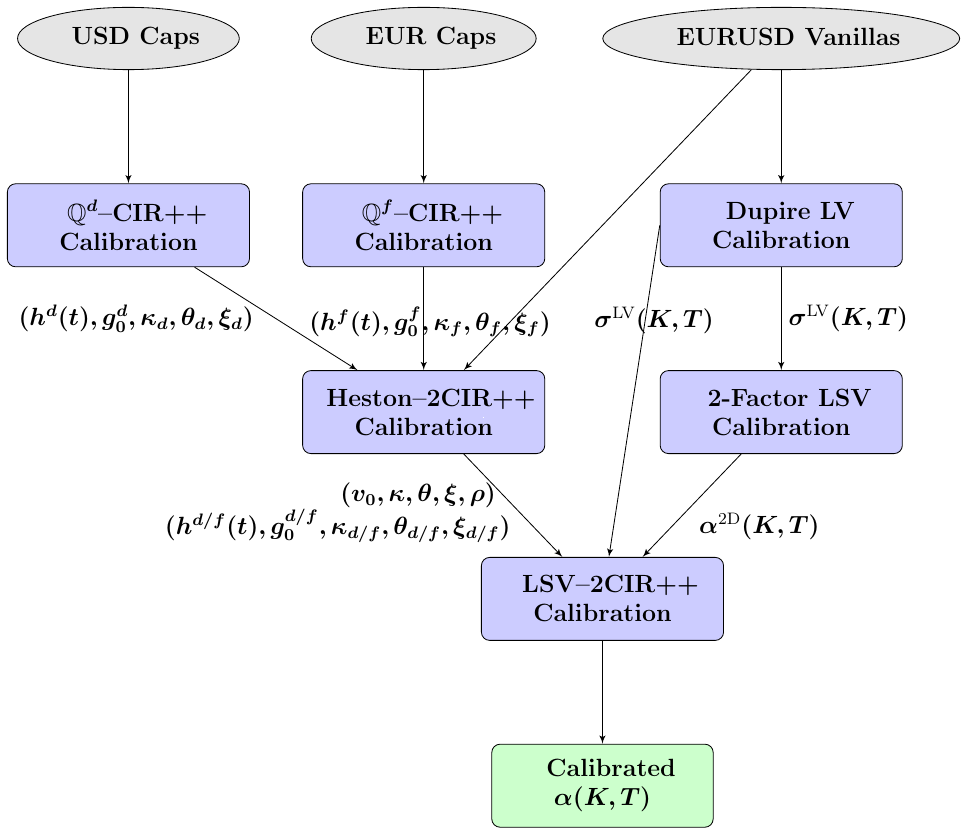}
\par\end{centering}

\caption{\label{fig:Calibration-routine-Flowchart}Full calibration routine
flowchart.}

\end{figure}

\begin{enumerate}
\item
Calibration of Heston-2CIR++ parameters:
\begin{enumerate}
\item
calibrate shifted CIR++ model for domestic and foreign short rates separately
(Appendix
\ifarxivvar
\ref{subsec:rate});
\else
D in \cite{cmr18});
\fi
\item
calibrate Heston-2CIR++ SV model assuming volatility, domestic and foreign
short rates are all independent processes
(Appendix
\ifarxivvar
\ref{subsec:hybridsv});
\else
F in \cite{cmr18});
\fi
\end{enumerate}
\item
Calibrate local volatility assuming time-dependent domestic and foreign short
rates (Appendix
\ifarxivvar
\ref{sub:Appendix-Local-volatility-calibration});
\else
E in \cite{cmr18});
\fi
\item
Calibration of Heston-2CIR++ LSV model:
\begin{enumerate}
\item
calibrate the leverage function of 2-factor Heston LSV model (\ref{eq:modelDefinitionLSVDetRates}) assuming
time-dependent domestic and foreign rates using local volatility from
Step 2 (Section~\ref{sub:2-Factor-Heston-LSV-Calibration});
\item
calibrate the leverage function of 4-factor LSV-2CIR++ LSV model (\ref{eq:Model Definition})
using Heston-2CIR++ parameters from Step 1, the local volatility from
Step 2 and leverage function of 2-factor Heston LSV model as a control
variate (Sections \ref{sec:Fast-Calibration} and~\ref{sec:Calibration-Results}).
\end{enumerate}
\end{enumerate}

\subsection{A necessary and sufficient condition for exact calibration}
\label{sec:nec_suff}

In the following, we give the main formula that links market call prices, via the Dupire local volatility, to prices under (\ref{eq:Model Definition}). 

%


In \cite{Guyon2011}, the following calibration condition is given\footnote{The (equivalent) context there is an equity with stochastic short rate and dividends.}: 
\begin{eqnarray}
\alpha^{2}\left(K,T\right) &=& \frac{\mathbb{E}^{\mathbb{Q}^{d}}\left[D_{T}^{d}\,|\,S_{T}=K\right]}{\mathbb{E}^{\mathbb{Q}^{d}}\left[D_{T}^{d}V_{T}\,|\,S_{T}=K\right]}\left(\sigma_{LV}\left(K,T\right)^{2}+ 
\frac{\mathbb{E}^{\mathbb{Q}^{d}}[Q_T]}{\frac{1}{2}K^{2}\frac{\partial^{2}C_{LV}}{\partial K^{2}}}
\right)\,,\label{eq:alphaFormula}
\end{eqnarray}
where 
$\sigma_{LV}$ is 
a local volatility as in (\ref{eq:modelDefinitionLocalVolatility}), and
\begin{eqnarray}
\label{eqn:Q}
Q_T &=&
D_{T}^{d} \left(r_{T}^{f}-\bar{r}^{f}\left(T\right)\right)\left(S_{T}-K\right)^{+}
-K D_{T}^{d}\mathbf{1}_{S_{T}\geq K}\left[\left(r_{T}^{d}-\bar{r}^{d}\left(T\right)\right)-\left(r_{T}^{f}-\bar{r}^{f}\left(T\right)\right)\right]. 
\end{eqnarray}

\begin{assumption}
\label{main-ass}
$\alpha$ is Lipschitz and uniformly
bounded by $\alpha_{\max}$, $h^{d,f}$ are uniformly bounded and
both the marginal density $\phi(\cdot,T)$ of $S_T$
in (\ref{eq:Model Definition}) and $\mathbb{E}^{\mathbb{Q}^{d}}\left[D_{T}^{d}V_{T}\,|\,S_{T}= \cdot \right]$
are continuous.
\end{assumption}
We note that the continuity and positivity in $\mathbb{R}^3_+$ of the joint density $\psi(\cdot,\cdot,\cdot,T)$ of $\left(S_{T}, V_{T}, D_{T}^{d}\right)$ in (\ref{eq:Model Definition})
is sufficient for $\mathbb{E}^{\mathbb{Q}^{d}}\left[D_{T}^{d}V_{T}\,|\,S_{T}= \cdot \right]$ to be continuous. 

We also define
$\varphi=2+\sqrt{2}$, $\zeta=\xi\alpha_{\max}$, and
\begin{equation}
\label{explosion_time}
\begin{cases}
T^{*}=\frac{2}{\sqrt{\varphi^{2}\zeta^{2}-\kappa^{2}}}\left[\frac{\pi}{2}+\arctan\left(\frac{\kappa}{\sqrt{\varphi^{2}\zeta^{2}-\kappa^{2}}}\right)\right], & \text{if\, }\kappa<\varphi\zeta,\\
T^{*}=\infty, & \text{if\, }\kappa\geq\varphi\zeta,
 \end{cases}
\end{equation}
which is a lower bound for the explosion time of $S_t^2$ (see \cite{Cozma2015}).

We prove the following theoretical results in Appendix \ref{app:proof}:

\begin{proposition}
\label{prop:fwd_eq}
Under Assumption \ref{main-ass}, 
the call
price $C\left(K,T\right)$ under model (\ref{eq:Model Definition})
satisfies 
\begin{eqnarray}
&&\frac{\partial C\left(K,T\right)}{\partial T}-\frac{1}{2}\alpha^2\left(K,T\right) K^{2}\frac{\mathbb{E}^{\mathbb{Q}^{d}}\left[D_{T}^{d}V_{T}\,|\,S_{T}=K\right]}{\mathbb{E}^{\mathbb{Q}^{d}}\left[D_{T}^{d}\,|\,S_{T}=K\right]}\frac{\partial^{2}C\left(K,T\right)}{\partial K^{2}}\label{eq:stoVolDupirePDE}\\
&&\hspace{4 cm} +\;\; \mathbb{E}^{\mathbb{Q}^{d}}\left[D_{T}^{d}r_{T}^{f}\left(S_{T}-K\right)^{+}\right]-\mathbb{E}^{\mathbb{Q}^{d}}\left[D_{T}^{d}\mathbf{1}_{S_{T}\geq K}K\left(r_{T}^{d}-r_{T}^{f}\right)\right]  =  0\,\nonumber 
\end{eqnarray}
for any strike $K>0$ and maturity $T<T^{*}$, with $T^{*}$ given by (\ref{explosion_time}).
\end{proposition}

\begin{theorem}
\label{thm:alphaFormula}  
Under Assumption \ref{main-ass},
the call
price $C\left(K,T\right)$ under model (\ref{eq:Model Definition})
matches the 
price $C^{LV}$ under the local volatility model (\ref{eq:modelDefinitionLocalVolatility})
for any strike $K>0$ and maturity $T<T^{*}$ only if (\ref{eq:alphaFormula}) holds for all $K,T>0$.

If (\ref{eq:stoVolDupirePDE}) has a unique solution, then the condition (\ref{eq:alphaFormula}) is also sufficient.
\end{theorem}

Uniqueness of the solution $C$ to the heat equation (\ref{eq:stoVolDupirePDE}) is normally expected under sufficient regularity of the diffusion coefficient and under a growth condition.

The condition (\ref{eq:alphaFormula}) expresses when a model of the form (\ref{eq:Model Definition}) with \emph{exogenously given} $\alpha$ is consistent with market prices, which are expressed through the local volatility function $\sigma_{LV}$.
We make no claim about the existence of such a model (see also Remark \ref{rem:high-xi} below), and note that $\alpha$ enters (\ref{eq:alphaFormula}) not only explicitly but also through the $\mathbb{Q}^{d}$-expectations.
Existence of a calibrated model is linked to the existence of a solution to the McKean-Vlasov equation which results when inserting $\alpha$ \emph{defined endogenously} by (\ref{eq:alphaFormula}) in terms of $\sigma_{LV}$ and the model itself into (\ref{eq:Model Definition}). In \cite{abergel2010nonlinear}, the existence of a short-time solution of the associated Fokker-Planck equation for the density of LSV processes of this type is shown under certain regularity assumptions. The upper bound on the time in \cite{abergel2010nonlinear} is needed to guarantee that the density stays strictly positive from an assumed strictly positive initial condition, and has no direct link to $T^*$ in this paper.

The ratio on the right-hand side of (\ref{eq:alphaFormula}) accounts for the stochastic volatility;
if there is no stochastic volatility
(i.e.\ $V_{T}=1$), we recover the formula in \cite{Clark2010}.
The term $Q_T$ 
accounts for the stochastic rates and,
if rates are deterministic, $Q_T=0$ and we recover the formula derived in \cite{Dupire},
\begin{equation}
\alpha^{2D}\left(K,T\right)=\frac{\sigma_{LV}\left(K,T\right)}{\sqrt{\mathbb{E}^{\mathbb{Q}^{d}}\left[V_{T}^{2D}\,|\,S_{T}^{2D}=K\right]}}\,.
\label{eq:alphaFormula2D}
\end{equation}
At time $T=0$, 
from $r_{0}^{d}-\bar{r}^{d}\left(0\right)=r_{0}^{f}-\bar{r}^{f}\left(0\right)=0$ we get $\alpha\left(K,0\right)=\sigma_{LV}\left(K,0\right)/\sqrt{v_{0}}.$

Theorem \ref{thm:alphaFormula} provides technical conditions for the formula presented in
\cite{Guyon2011}, where a formal proof is given without specification of the rates processes.
Here, we consider specifically
an LSV--2{CIR\scalebox{.9}{\raisebox{.5pt}{++}}}
model and derive the result rigorously.
In Lemma \ref{prop:TrueMartingale} we provide a sufficient condition for the
process 
\[
\int_{0}^{t}\mathbf{1}_{S_{u}\geq K}D_{u}^{d}S_{u}\alpha\left(S_{u},u\right)\sqrt{V_{u}}\,dW_{u}
\]
to be a true martingale up to $T^{*}$, which is an important step
in the proof of Theorem \ref{thm:alphaFormula}. 
On the one hand, $T^{*}$ is a lower bound for the explosion time of the second moment
of the discounted spot process $D_{t}^{d}S_{t}$. On the other hand,
\cite{Andersen2007a} show that the moment explodes in finite time
for the Heston model, a property that is inherited by our Heston-type
LSV--2{CIR\scalebox{.9}{\raisebox{.5pt}{++}}}
model (\ref{eq:Model Definition}) as well as the Heston-type LSV--2Hull--White model
in \cite{Deelstra2013}. Therefore, the formula may not hold for certain
values of the model parameters and for large maturities $T$. However,
in practice, $T^{*}$ is very large. For instance, from our calibration
given in Section \ref{sec:Calibration-Results} we obtain
$\kappa=1.4124$, $\xi=0.2988$, $\alpha_{\max}=1.40$, such that $T^{*}=28.6.$

\begin{remark}
\label{rem:high-xi}
A numerical experiment in \cite{Guyon2011} raises the question of the
existence of a calibrated 2-factor LSV model for large $\xi$ (there, $\xi\approx350\%$
is used to match forward smiles).
In this particular case and with the other model parameters kept the
same, we find $T^{*}=0.20$, which indicates that moment explosions
may occur sooner.
\end{remark}

\section{Fast calibration with a new control variate particle method\label{sec:Fast-Calibration}}

In this and the next section, we describe two of the main components of the calibration routine. 
We recall the calibration condition (\ref{eq:alphaFormula}),
which involves conditional expectations as well as standard expectations, which have to be estimated under model (\ref{eq:Model Definition}).

First, we describe the basic particle method for the estimation of these expectations. Then, we present
the various control variates, building on intermediary calibration steps, which we use in order to reduce
the computational cost of the calibration of $\alpha$ in the $4$-factor model (\ref{eq:Model Definition}).

Therefore, we require the prior calibration of the interest rate models in (\ref{eq:Model Definition}), the LV model (\ref{eq:modelDefinitionLocalVolatility}), the
Heston-2CIR++ model, 
and the LSV model (\ref{eq:modelDefinitionLSVDetRates}).
The calibration of the latter via a PDE is detailed in Section \ref{sub:2-Factor-Heston-LSV-Calibration},
while we refer to Appendices 
\ifarxivvar
\ref{sub:CIR++-Model-Calibration}, \ref{sub:Appendix-Local-volatility-calibration}, and \ref{sub:Heston-2CIR++-Calibration} 
\else
D, E, and F of \cite{cmr18}
\fi
for the former three.


Equation (\ref{eq:alphaFormula}) contains the local volatility, which
can be obtained from derivatives of market prices
from (\ref{eq:dupirePDE})
by re-arranging it (into Dupire's formula), and explicitly the second derivative of market prices with respect to strike.
Different approximation approaches are used in practice, e.g., one writes the formulae in terms of the implied volatility,
and uses a smooth parametrisation for the differentiation.
Here, we first calibrate a parametrisation of the local volatility model with
a fixed-point iteration as in \cite{Rehai2006,Tur2014} 
and then use $\frac{\partial^{2}C_{LV}}{\partial K^{2}}$ obtained from
the solution of the forward PDE (\ref{eq:dupirePDE}) with a smoothing scheme
(see Appendix
\ifarxivvar
\ref{sub:Appendix-Local-volatility-calibration}).
\else
E of \cite{cmr18}).
\fi



\subsection{Calibration by particle method\label{sub:Calibration-Algorithm-with-CVParticleMethod}}


A calibrated $\alpha$ is implicitly defined by (\ref{eq:alphaFormula}),
where the right-hand side depends on $\alpha$ in a non-linear way
through the (conditional) expectations. 
Formal insertion of the calibration formula into the SDE (\ref{eq:Model Definition})
leads to a process where the diffusion coefficient depends on the distribution of the joint process
$X_{t}=(S_{t},V_{t},r_{t}^{d},r_{t}^{f}, D_t^d)$.
The process
thus falls in the class of 
McKean-Vlasov processes 
\cite{McKean1966}.

The existence and uniqueness of the solution for this McKean-Vlasov SDE are not
established theoretically, to the best of our knowledge.
From an empirical perspective, in \cite{Guyon2011} and in Section 11.8 of \cite{GuyonLabordere2013} the authors
encountered problems for very high values of $\xi$; 
see Remark \ref{rem:high-xi}.
%
In our case, for $\xi$ 
calibrated to market smiles ($\approx30\%$)  we are able to reach a high
accuracy. 


The particle method for processes of this type was introduced{
in \cite{McKean1966} and is discussed in Chapter 2, Section 3 of  \cite{Sznitman1991}; it was applied to LSV model calibration
in} \cite{Guyon2011} and in Section 11.6 of \cite{GuyonLabordere2013}.

We define $N$-sample path approximations of $X_{t}$ as
$\left(X_{t}^{i,N}\right)_{i\in\llbracket1,N\rrbracket}=(S_{t}^{i},V_{t}^{i},r_{t}^{d,i},r_{t}^{f,i},D_t^{d,i})_{i\in\llbracket1,N\rrbracket}$
by the ($5\times N$)-dimensional
SDE 
\begin{eqnarray}
\qquad
\begin{cases}
\cfrac{dS_{t}^{i}}{S_{t}^{i}}=\left(r_{t}^{d,i}-r_{t}^{f,i}\right)\,dt+\hat{\alpha}_{N}\!\left(S_{t}^{i},t,\left(X_{t}^{j,N}\right)_{j\leq N}\right)\sqrt{V_{t}^{i}}\,dW_{t}^{i}\\
r_{t}^{d,i}=g_{t}^{d,i}+h^{d}\left(t\right)\\
r_{t}^{f,i}=g_{t}^{f,i}+h^{f}\left(t\right)\\
dg_{t}^{d,i}=\kappa_{d}\left(\theta_{d}-g_{t}^{d,i}\right)\,dt+\xi_{d}\sqrt{g_{t}^{d,i}}\,dW_{t}^{\gd,i}\\
dg_{t}^{f,i}=\left(\kappa_{f}\left(\theta_{f}-g_{t}^{f,i}\right)-\rho_{Sf}\xi_{f}\sqrt{g_{t}^{f,i}}\hat{\alpha}_{N}\!\left(S_{t}^{i},t,\left(X_{t}^{j,N}\right)_{j\leq N}\right)\sqrt{V_{t}^{i}}\right)\,dt+\xi_{f}\sqrt{g_{t}^{f,i}}\,dW_{t}^{\gf,i}\\
dV_{t}^{i}=\kappa\left(\theta-V_{t}^{i}\right)\,dt+\xi\sqrt{V_{t}^{i}}\,dW_{t}^{V,i} \\
dD_t^{d,i}=- r_t^{d,i} D_t^{d,i} \, dt \,, &
\end{cases}
\label{part_sys}
\end{eqnarray}
where $(W_{t}^{i},W_{t}^{\gd,i},W_{t}^{\gf,i},W_{t}^{V,i})$, ${i\in\llbracket1,N\rrbracket}$
are $N$ i.i.d.\ copies of the four correlated Brownian motions,
and $\hat{\alpha}_{N}$ is an estimator for $\alpha$ based on $\left(X_{t}^{i,N}\right)_{i\leq N}$,
\begin{eqnarray}
\label{estimator-alpha}
\hat{\alpha}_{N}\left(K,T,\left(X_{t}^{i,N}\right)_{i\leq N}\right)=\sqrt{\frac{\sigma_{LV}\left(K,T\right)^{2}}{\hat{p}_{N}\left(K,T\right)}+
\frac{\widehat{Q}\left(K,T\right)}{\frac{1}{2}\hat{p}_{N}\left(K,T\right)K^{2}\frac{\partial^{2}C_{LV}}{\partial K^{2}}}}\,,
\end{eqnarray}
with 
\begin{eqnarray}
\label{pNqN}
\widehat{Q}
=\frac{1}{N}\sum_{i=1}^{N} Q_{T}^{i} \qquad
\text{and} \qquad
\hat{p}_{N}\left(K,T\right)=\frac{\sum_{i=1}^{N}D_{T}^{d,i}V_{T}^{i}\delta_{N}\left(S_{T}^{i}-K\right)}{\sum_{i=1}^{N}D_{T}^{d,i}\delta_{N}\left(S_{T}^{i}-K\right)},
\end{eqnarray}
where
$\hat{p}_{N}$ is an estimator for
$\mathbb{E}^{\mathbb{Q}^{d}}\!\!\left[D_{T}^{d}V_{T}\,|\,S_{T}=K\right]/\,\mathbb{E}^{\mathbb{Q}^{d}}\!\!\left[D_{T}^{d}\,|\,S_{T}=K\right]$,
with $\delta_{N}$ a kernel function, and
$Q_{T}^{i}$ is the $i$-th sample of $Q_{T}$ from (\ref{eqn:Q}) based on $X_T^{i,N}$.

The paths of the $5\times N$-dimensional process
$(X_{t}^{i,N})_{i\leq N}$ are now entangled due to the dependence
on $\hat{\alpha}_{N}$ in $(X_{t}^{i,N})_{i\leq N}$. The
process can be seen as a system of $N$ interacting particles evolving
in a $5$-dimensional space, where particle $i$ is defined by its
position $X_{t}^{i,N}$. As in \cite{Guyon2011}, we will therefore
use the term ``particle'' instead of ``path''. 
Because of the four driving factors, we will keep referring to this as a 4-factor model in spite of the extra state variable $D^d$.

A central ingredient for proving convergence of the particle method is the chaos propagation property
(see Chapter 2, Section 3 of \cite{Sznitman1991}), which is not proven for the present case.


\subsection{Variance reduction for the Markovian projection}
\label{sub:Variance-Reduction-for-Markovian-Projection}

Our goal here is to reduce the variance of the estimator $\hat{p}_{N}$ from (\ref{pNqN}) to be able to use a minimal number of particles.

We assume that the 2-factor LSV model
(\ref{eq:modelDefinitionLSVDetRates}) is perfectly calibrated to
market call prices, i.e.\ that (\ref{eq:alphaFormula2D}) is satisfied.
Then we will use 
\begin{equation}
\hat{p}_{N}^{2D}\left(K,T\right)=\sum_{i=1}^{N}\frac{V_{T}^{2D,i}\delta_{N}\left(S_{T}^{2D,i}-K\right)}{\sum_{i=1}^{N}\delta_{N}\left(S_{T}^{2D,i}-K\right)},
\label{eq:p2D}
\end{equation}
which is an estimator for
\[
p^{2D}(K,T)=\mathbb{E}^{\mathbb{Q}^{d}}\left[V_{T}^{2D}\,|\,S_{T}^{2D}=K\right],
\]
as a control variate for $\hat{p}_{N}$, 
and $p^{2D}$ will be computed using a PDE solver.
The Kolmogorov forward equation for $p^{2D}$ is commonly used for the calibration of LSV models (see \cite{Clark2010,Ren&Madan2007,Wyns2016FV}),
and we propose in Section \ref{sub:2-Factor-Heston-LSV-Calibration} a new method which is tailored to the specific difficulties associated with density functions in Heston-style models.

We thus define a new estimator $p_{N}^{*}$ by
\begin{equation}
p_{N}^{*}(K,T)=\hat{p}_{N}\left(K,T\right)+\lambda\left(\hat{p}_{N}^{2D}\left(K,T\right)-p^{2D}\left(K,T\right)\right).\label{eq:Control-Variate-Formula}
\end{equation}
The latter has an asymptotically diminishing bias if we assume the particle method
to converge in distribution (and neglect the time stepping bias).

In order to get a good estimate for the optimal $\lambda$, we can rewrite the above estimator as 
\[
p_{N}^{*}=\frac{1}{N}\sum_{i=1}^{N}m_{i}+\lambda\left(\frac{1}{N}\sum_{i=1}^{N}m_{i}^{2D}-p^{2D}\right)
\]
with
\[
\begin{array}{cc} m_{i}=\frac{D_{T}^{d,i}V_{T}^{i}\delta_{N}\left(S_{T}^{i}-K\right)}{\frac{1}{N}\sum_{i=1}^{N}D_{T}^{d,i}\delta_{N}\left(S_{T}^{i}-K\right)},\quad & 
m_{i}^{2D}=\frac{V_{T}^{2D,i}\delta_{N}\left(S_{T}^{2D,i}-K\right)}{\frac{1}{N}\sum_{i=1}^{N}\delta_{N}\left(S_{T}^{2D,i}-K\right)}\end{array} \, ,
\]
which mimics the standard Monte Carlo control variate form. We can
think of $m_{i}$ and $m_{i}^{2D}$ roughly as samples of two random variables
$m$ and $m^{2D}$ respectively (but note they are not independent, although for large $N$ the correlation is very low),
and for the best variance reduction (see Section 4.1 in \cite{Glasserman2004}), we take
\[
\lambda=-\frac{\text{Cov}\left(m,m^{2D}\right)}{\text{Var}\left(m^{2D}\right)}\,,
\]
which we can estimate by 
\begin{equation}
\hat{\lambda}_{N}=-\frac{\sum_{i=1}^{N}\left(m_{i}-\hat{p}_N \right)\left(m_{i}^{2D}-\hat{p}_N^{2D}\right)}{\sum_{i=1}^{N}\left(m_{i}^{2D}-\hat{p}_N^{2D}\right)^{2}}\,.
\label{eq:lambda-Markov-Proj}
\end{equation}
We recall that the expected variance reduction factor is
\begin{equation}
\frac{1}{1-\text{Corr}\left(m,m^{2D}\right)^{2}}\,.\label{eq:variance-reduction-factor}
\end{equation}
 Hence, if the stochastic rates are not highly volatile, i.e. if $\xi_{f}$
and $\xi_{d}$ are small enough, we expect a very good variance reduction
as the correlation between the particles generated by the 4-factor
hybrid LSV model (\ref{eq:Model Definition}) and by the 2-factor
LSV model (\ref{eq:modelDefinitionLSVDetRates}) will be high. Our
numerical tests performed on a model calibrated to recent EUR and
USD market data exhibit a correlation
between model (\ref{eq:Model Definition}) and
(\ref{eq:modelDefinitionLSVDetRates}) of $95\%$ up to $1.5$ years
and $50\%$ around $5$ years. Additionally, our stress test in Subsection
\ref{sub:Stress-scenario} suggests that even under high volatility
regimes for the rate processes, i.e. when $\xi_{f}$ and $\xi_{d}$ are large,
 the variance reduction brought by this control variate is significant.

\subsection{Variance reduction for standard expectations}
\label{sub:Variance-Reduction-for-Standard-Expectations}

Here, we discuss the variance reduction for the estimator
$\widehat{Q}$
in (\ref{pNqN}) for $\mathbb{E}^{\mathbb{Q}^{d}} \left[Q_T \right]$, where we repeat 
\begin{eqnarray}
\label{QX}
Q_T = 
\underbrace{D_{T}^{d}\left(r_{T}^{f}-\bar{r}^{f}\left(T\right)\right)\left(S_{T}-K\right)^{+}}_{=: X_{1,T}}
-K \underbrace{D_{T}^{d}\mathbf{1}_{S_{T}\geq K}\left[\left(r_{T}^{d}-\bar{r}^{d}\left(T\right)\right)-
\left(r_{T}^{f}-\bar{r}^{f}\left(T\right)\right)\right]}_{=: X_{2,T}},
\end{eqnarray}
from (\ref{eqn:Q}) for the convenience of the reader.
This is an estimator for a standard expectation (in contrast to conditional expectations).
We introduce control variates
\[
Y_{1,T} = D_{T}^{d}\left(S_{T}-K\right)^{+}, \qquad Z_{1,T}=\widehat{r}_{T}^{f}-\bar{r}^{f}\left(T\right)\,
\]
for $X_{1,T}$ defined in (\ref{QX}), and
\begin{eqnarray*}
Y_{2,T}=D_{T}^{d}\mathbf{1}_{S_{T}\geq K}, \qquad \quad
Z_{2,T}=\left(r_{T}^{d}-\bar{r}^{d}\left(T\right)\right)-\left(\widehat{r}_{T}^{f}-\bar{r}^{f}\left(T\right)\right)\,
\end{eqnarray*}
for $X_{2,T}$, and where $\widehat{r}^{f}$ is the foreign rate process without the quanto adjustment.\footnote{The last quantity
is introduced because the expectation of ${r}^{f}$ is not analytically available.}
We know that if the model (\ref{eq:Model Definition}) is perfectly calibrated
to call option prices, 
\begin{eqnarray*}
\mathbb{E}^{\mathbb{Q}^{d}}\left[Y_{1,T}\right] &=& C_{LV}\left(K,T\right), \\
\mathbb{E}^{\mathbb{Q}^{d}}\left[Y_{2,T}\right] &=&-\frac{\partial C_{LV}}{\partial K}\left(K,T\right),
\end{eqnarray*}
estimated from market data via a calibrated LV model.
The following are also analytically available:
\begin{eqnarray*}
\zeta_1  \equiv  \mathbb{E}^{\mathbb{Q}^{d}}\left[Z_{1,T}\right]&=&g_{\text{0}}^{f}e^{-\kappa_{f}T}+\theta_{f}\left(1-e^{-\kappa_{f}T}\right)+h^{f}\left(T\right)-\bar{r}^{f}\left(T\right), \\
\zeta_2  \equiv \mathbb{E}^{\mathbb{Q}^{d}}\left[Z_{2,T}\right]&=&g_{\text{0}}^{d}e^{-\kappa_{d}T}+\theta_{d}\left(1-e^{-\kappa_{d}T}\right)+h^{d}\left(T\right)-\bar{r}^{d}\left(T\right)-\zeta_{1}.
\end{eqnarray*}

We denote the Monte Carlo estimators of the corresponding $\mathbb{Q}^{d}$-expectations as
$\widehat{X}_1$, $\widehat{X}_2$, $\widehat{Y}_1$, $\widehat{Y}_2$, $\widehat{Z}_1$, $\widehat{Z}_2$,
respectively, 
using the same Brownian paths for $W,\;W^{V},\;W^{\gd},\;W^{\gf}$ in all estimators.

We can define a new Monte Carlo estimator ${Q}^{*}$
for $\mathbb{E}^{\mathbb{Q}^{d}}\left[Q_T\right]$ as
\begin{eqnarray}
{Q}^{*} & = & {X}_1^* - K {X}_2^{*},
\label{eq:Control-Variates-Formula-StoRatesPart}
\end{eqnarray}
with
\begin{eqnarray*}
{X}_1^{*} & = & \widehat{X}_1 +\lambda_{1} \left(\widehat{Y}_1 -C_{LV}\left(K,T\right)\right)+\eta_{1}\left(\widehat{Z}_1 - \zeta_{1}\right), \\
{X}_2^{*} & = & \widehat{X}_2 +\lambda_{2} \left(\widehat{Y}_2 +\frac{\partial C_{LV}\left(K,T\right)}{\partial K}\right)+\eta_{2}\left(\widehat{Z}_2-\zeta_{2}\right)\,.
\end{eqnarray*}
The weights $\lambda_{1}$, $\lambda_{2}$, $\eta_{1}$, $\eta_{2}$ above are chosen to minimize the variance of ${Q}^{*}$ (see \cite{Glasserman2004}).

This approach is particularly useful for out-of-the-money options
and digital options as the Monte Carlo estimator will exhibit higher
variance in these settings.

A computational problem arises if there is no particle with
$S_{T}>K$, which happens if $K$ is large and the total number of particles is relatively small (as will be the case with control variates),
since the estimators of
$\mathbb{V}\left[D_{T}^{d}\mathbf{1}_{S_{T}\geq K}\right]$ and $\mathbb{V}\left[D_{T}^{d}\left(S_{T}-K\right)^{+}\right]$
are then zero. In that case, we pick
\begin{eqnarray*}
\lambda_{1}=-\zeta_{1},\quad\eta_{1}=-C_{LV}\left(K,T\right)\,,\\
\lambda_{2}=-\zeta_{2},\quad\eta_{2}=\frac{\partial C_{LV}\left(K,T\right)}{\partial K}\,,
\end{eqnarray*}
such that both control variates are of the same order of magnitude.

\subsection{Implementation details}

{The leverage function $\alpha$ can in principle be computed for any $K$ and $T$ by the 
estimator (\ref{estimator-alpha}). However, for computational purposes, we
defined it in this way on a grid of points and interpolate it from there with cubic splines
in spot and piecewise constant in time.
We denote by $N_{T}$ the number of maturities. 
Then there are $N_{T}+1$ volatility
``slices'' in total such that we denote the $m$-th time slice
$\alpha\left(\cdot,T_{m}\right)$, by $\alpha_{m}$, represented numerically
as splines with $N_{S}$ nodes. While having $N_{S}$ too
small will lead to accuracy problems, choosing it too large will make
the surface rougher due to over-fitting. We find 25-30 points to provide
a good trade-off between accuracy and smoothness.
For a given $T_{m}$, the leverage function is thus defined on some interval
$\left[S_{\min}^{m},\,S_{\max}^{m}\right]$
and is extrapolated constant outside these bounds. Because we need more
grid points around the forward value and less around $S_{\min}^{m}$
and $S_{\max}^{m}$, we use a hyperbolic grid (with $\eta=0.05$, see
Appendix \ref{sub:Mesh-construction} for more details) refined
around the forward value 
\[
F_{m}=S_{0}e^{\int_{0}^{T_{m}}\left(\bar{r}^{d}\left(t\right)-\bar{r}^{f}\left(t\right)\right)dt}, \qquad
\text{with}  \qquad
S_{\min}^{m}=F_{m}e^{-3 \sigma_{F}\left(T_{m}\right)\sqrt{T_{m}}},\quad S_{\max}^{m}=F_{m}e^{3\sigma_{F}\left(T_{m}\right)\sqrt{T_{m}}},
\]
where $\sigma_{F}\left(T_{m}\right)$ is the at-the-money forward
market volatility for maturity $T_{m}$ (interpolated linearly in
variance).
Each of the grid values can be seen as a parameter and we denote them
by $\left(\alpha_{m,j}\right)_{m\leq N_{T},\,j\leq N_{S}}$ with the
associated spot grid values $\left(s_{m,j}\right)_{m\leq N_{T},\,j\leq N_{S}}$.}

We now give the calibration algorithm. 
As previously, we denote the 
particle system at time $T$ for the model (\ref{eq:Model Definition}){
by {\small{}$\left(S_{T}^{i},V_{T}^{i},r_{T}^{d,i},r_{T}^{f,i},D_{T}^{d,i}\right)_{i\leq N}$}.
Similarly, we denote the 2-factor particle system at time $T$ for the model
}(\ref{eq:modelDefinitionLSVDetRates}){
by {\small{}$\left(S_{T}^{2D,i},V_{T}^{2D,i}\right)_{i\leq N}$}. }

We work with an Gaussian kernel
\[
\delta_{N}\left(x,T\right)=\frac{e^{-\frac{1}{2}\left(\frac{x}{h_{N}\left(T\right)}\right)^{2}}}{h_{N}\left(T\right)\sqrt{2\pi}}\,,
\]
with a bandwidth given by a Silverman-type rule (see \cite{Silverman1986})
\[
h_{N}\left(T\right)=\eta S_{0}\sigma_{LV}\left(S_{0},T\right)\sqrt{\max\left(T,T_{\min}\right)}N^{-\frac{1}{5}},
\]
where 
$\eta=1.5$ and $T_{\min}=0.25$ in our tests.

 The step-by-step calibration is detailed in Algorithm \ref{alg:ParticleMethod}.

\begin{algorithm}[!htp]
\caption{$\alpha\left(s,T\right)$ Calibration with control variate particle
method\label{alg:ParticleMethod}}

\begin{algor}
\item [{{*}}] $\alpha\left(s,T_{1}=0\right)=\frac{\sigma_{LV}\left(s,0\right)}{\sqrt{v_{0}}}$
\item [{for}] ( $m=1\,;\,m\leq N_{T}\,;\,m++$) 

\begin{algor}
\item [{{*}}] \begin{raggedright}
\textbf{generate} $\left(Z,Z_{v},Z_{d},Z_{f}\right)_{i\leq N}$
and $\left(U\right)_{i\leq N}$, i.e. $4\times N$ independent draws
from $\mathcal{N}\left(0,1\right)$ and $N$ draws from $\mathcal{U}\left(\left[0,1\right]\right)$,
respectively\\

\par\end{raggedright}
\item [{{*}}] \textbf{evolve }the {4-factor particle
system} from $T_{m}$ to $T_{m+1}$ with $QE-$Scheme (\ref{eq:QE-Scheme})
where $\alpha\left(s,\left[T_{m},T_{m+1}\right[\right)=\alpha\left(s,T_{m}\right)$\\

\item [{{*}}] \textbf{evolve }the {2-factor particle}
system from $T_{m}$ to $T_{m+1}$ with $QE-$Scheme (\ref{eq:QE-Scheme})
with pre-computed $\alpha^{2D}$ and using $\left(Z,Z_{v},U\right)_{i\leq N}$\\

\item [{{*}}] \textbf{solve} the Dupire forward PDE (\ref{eq:dupirePDE})
from $T_{m}$ to $T_{m+1}$ for $C_{LV}$, $\frac{\partial C_{LV}}{\partial K}$
and $\frac{\partial^{2}C_{LV}}{\partial K^{2}}$\\

\item [{{*}}] \textbf{set }$T=T_{m+1}$
\item [{for}] ( $j=1\,;\,j\leq N_{S}\,;\,j++$)

\begin{algor}
\item [{{*}}] \textbf{set} $K=s_{m+1,j}$

\item [{{*}}] \textbf{compute }as in (\ref{eq:alphaFormula2D})
\vspace{-1em}
\[
p^{2D}=\mathbb{E}^{\mathbb{Q}^{d}}\left[V_{T}^{2D}\,|\,S_{T}^{2D}=K\right]=\left(\frac{\sigma_{LV}\left(K,T\right)}{\alpha^{2D}\left(K,T\right)}\right)^{2}
\]
\vspace{-1em}
\item [{{*}}] \textbf{compute} as in (\ref{eq:p2D})
\vspace{-1em}
\textbf{ 
\[
\hat{p}_{N}^{2D}\left(K,T\right)=\frac{\sum_{i=1}^{N}V_{T}^{2D,i}\delta_{N}\left(S_{T}^{2D,i}-K\right)}{\sum_{i=1}^{N}\delta_{N}\left(S_{T}^{2D,i}-K\right)}
\]
}
\vspace{-1em}
\item [{{*}}] \textbf{compute} $\hat{\lambda}_{N}$ as in (\ref{eq:lambda-Markov-Proj})
and

\[
p_{N}^{*}\left(K,T\right)=\frac{\sum_{i=1}^{N}D_{T}^{d,i}V_{T}^{i}\delta_{N}\left(S_{T}^{i}-K\right)}{\sum_{i=1}^{N}D_{T}^{d,i}\delta_{N}\left(S_{T}^{i}-K\right)}+\hat{\lambda}_{N}\left(\hat{p}_{N}^{2D}\left(K,T\right)-p^{2D}\left(K,T\right)\right)
\]

as in (\ref{eq:Control-Variate-Formula})
\item [{{*}}] \textbf{compute} as in (\ref{eq:Control-Variates-Formula-StoRatesPart})
\vspace{-1em}
\begin{eqnarray*}
{Q}^{*} & = & \left({X}_1^{*}-K {X}_2^{*}\right)
\end{eqnarray*}
\vspace{-1em}
with

\begin{eqnarray*}
{X}_1^{*} & = & \widehat{X}_1 +\lambda_{1}\left(\widehat{Y}_1 -C_{LV}\left(K,T\right)\right)+\eta_{1}\left(\widehat{Z}_1-\mathbb{E}^{\mathbb{Q}^{d}}\left[\left(r_{T}^{f}-\bar{r}^{f}\left(T\right)\right)\right]\right)\\
{X}_2^{*} & = & \widehat{X}_2 +\lambda_{2}\left(\widehat{Y}_2+\frac{\partial C_{LV}\left(K,T\right)}{\partial K}\right)+\eta_{2}\left(\widehat{Z}_2-\mathbb{E}^{\mathbb{Q}^{d}}\left[\left(r_{T}^{d}-\bar{r}^{d}\left(T\right)\right)-\left(r_{T}^{f}-\bar{r}^{f}\left(T\right)\right)\right]\right)
\end{eqnarray*}

\item [{{*}}] \textbf{compute} 
\[
\alpha_{m+1,j}=\sqrt{\frac{1}{p_{N}^{*}\left(K,T\right)}\left(\sigma_{LV}\left(K,T\right)^{2}+\frac{{Q}^{*}}{\frac{1}{2}K^{2}\frac{\partial^{2}C_{LV}\left(K,T\right)}{\partial K^{2}}}\right)}\,
\]

\end{algor}
\item [{endfor}]~
\end{algor}
\item [{endfor}]~\end{algor}
\end{algorithm}
Our empirical findings suggest that Quasi-Monte Carlo sampling of the random
numbers does not provide significant accuracy gains.
In order to speed up the computation of the sums involving kernel
functions such as
\[
\sum_{i=1}^{N}V_{T}^{2D,i}\delta_{N}\left(S_{T}^{2D,i}-K\right)\,,
\]
it is advised (see \cite{Guyon2011}) to sort the particle state vector
by spot value and select only the relevant particles that fall inside
an interval $\left[K-\Delta K,\,K+\Delta K\right]$, where we choose
\[
\Delta K=\sqrt{-2h_{N}^{2}\left(T\right)\ln\left(\epsilon\sqrt{2\pi}h_{N}\right)}\,,
\]
with $\epsilon=10^{-5}$.

\section{Two-factor Heston-type LSV model calibration by PDE 
\label{sub:2-Factor-Heston-LSV-Calibration}}

Here, we describe the calibration of the 2-factor sub-model
of (\ref{eq:Model Definition}) defined in (\ref{eq:modelDefinitionLSVDetRates})
by solution of the forward PDE.\footnote{In this section only, we write $S$,$V$
and $\alpha$ in lieu of $S^{2D}$, $V^{2D}$ and $\alpha^{2D}$,
respectively, to ease notation.}
{

}

\subsection{Transformation and weak formulation}

The following is a small variation of the main result in \cite{Lucic2008},
and the proof is therefore omitted.
Note that a new non-Dirichlet boundary condition appears at $z=0$.

\begin{theorem}
\label{thm:Fokker-Planck-HestonLSV}Define the region $\Omega=\mathbb{R}_{+}^{2}$ 
and assume that the density $\phi$ (under $\mathbb{Q}^{d}$) of the
Markovian process $\left(S_{t},V_{t}\right)$ started at $\left(S_{0},v_{0}\right)$
at time 0 exists 
and is $C^{2,2,1}\left(\Omega \times \mathbb{R}_{+}\right)$.
Then $\phi$ is the solution to the Kolmogorov forward equation
\begin{equation}
\begin{cases}
\frac{\partial\phi}{\partial t}+\left(\bar{r}^{d}\left(t\right)-\bar{r}^{f}\left(t\right)\right)\frac{\partial x\phi}{\partial x}+\frac{\partial\kappa\left(\theta-z\right)\phi}{\partial z}\\
\hspace{1 cm} -\; \frac{1}{2}\left(\frac{\partial^{2}x^{2}\alpha^{2}\left(x,t\right)z\phi}{\partial x^{2}}+\frac{\partial^{2}\xi^{2}z\phi}{\partial z^{2}}+2\frac{\partial^{2}\rho\xi x\alpha\left(x,t\right)z\phi}{\partial x\partial z}\right)=0, & \left(x,z\right)\in\Omega,\,t>0\,,\\
\left.\left(\frac{\xi^{2}}{2}\frac{\partial z\phi}{\partial z}-\kappa\left(\theta-z\right)\phi+\rho\xi z\frac{\partial x\alpha\left(x,t\right)\phi}{\partial x}\right)\right\rfloor _{z=0}=0, & z=0,\,x\geqslant0,\,t>0\,,\\
\lim_{z\rightarrow\infty}\phi(x,z,t)=\lim_{x\rightarrow\infty}\phi(x,z,t)=\phi(0,z,t)=0, & \left(x,z\right)\in\Omega,t>0\,,\\
\lim_{t\rightarrow0}\phi(x,z,t)=\delta(x-S_{0},z-v_{0}), & \left(x,z\right)\in\Omega\,.
\end{cases}\label{eq:Fokker-Planck-HestonLSV}
\end{equation}
\end{theorem}
\begin{proof}
Similar to the proof of Lemma 4.1 and Theorem 4.1 in \cite{Lucic2008}.
\end{proof}

The marginal density function of the CIR process at $t$ is 
(see \cite{Cox1985})
\[
\phi_{v}\left(z,t\right)=ce^{-u-cz}\left(\frac{cz}{u}\right)^{\beta/2}I_{\beta}\left(2\sqrt{cuz}\right)\,,
\]
with $c=\frac{2\kappa}{\left(1-e^{-\kappa t}\right)\xi^{2}}$ , $u=cv_{0}e^{-\kappa t}$
and $\beta=\frac{2\kappa\theta}{\xi^{2}}-1$, where $I_{\beta}$ is
the modified Bessel function of the first kind of order $\beta$.
We can write an asymptotic expression for small $z$ using the asymptotic
formula for the modified Bessel function found in \cite{Abramowitz1974},
\begin{eqnarray}
\phi_{v}\left(z,t\right) 
 & \sim & \frac{c^{\left(\beta+1\right)}e^{-u-cz}}{\Gamma\left(\beta+1\right)}z^{\beta}\,,\label{eq:cirDensitySmallZ}
\end{eqnarray}
such that $\phi_{v}\left(z,t\right)$ diverges for $z=0^{+}$ when
$2\kappa\theta<\xi^{2}$. This agrees with the well-known density $\phi_{v}^{\infty}$ of the stationary distribution (see \cite{Cox1985})
given by
\begin{equation}
\phi_{v}^{\infty}\left(z\right)=\lim_{t\rightarrow\infty}\phi_{v}\left(z,t\right)=\frac{\omega^{\left(\beta+1\right)}z^{\beta}e^{-\omega z}}{\Gamma\left(\beta+1\right)} \,, \label{eq:hestonDensityStationary}
\end{equation}
with $\omega=2\kappa/\xi^{2}$. 
This motivates a scaling of the density $\phi$ of the form $p=\phi z^{-\beta}$,
and to solve a new PDE for $p$ which we define hereafter. 
By insertion in (\ref{eq:Fokker-Planck-HestonLSV}) we get the following.
\begin{corollary}
For any $\beta \in \mathbb{R}$,
$p=\phi z^{-\beta}$ satisfies the
initial boundary value problem
\begin{equation}
\begin{cases}
\frac{\partial p}{\partial t}+(\bar{r}^{d}\left(t\right)-\bar{r}^{f}\left(t\right))\frac{\partial xp}{\partial x}-\beta\kappa p\\
\hspace{2 cm} + \; \frac{\partial\kappa\left(\theta-z\right)p}{\partial z}-\left(\beta+1\right)\left(\frac{\partial\xi^{2}p}{\partial z}+\frac{\partial\rho\xi x\alpha\left(x,t\right)p}{\partial x}\right)\\
\hspace{2 cm} -\; \frac{1}{2}z\left[\frac{\partial^{2}x^{2}\alpha^{2}\left(x,t\right)p}{\partial x^{2}}+\frac{\partial^{2}\xi^{2}p}{\partial z^{2}}+2\frac{\partial^{2}\rho\xi x\alpha\left(x,t\right)p}{\partial x\partial z}\right]=0\,, & \left(x,z\right)\in\Omega,\,t>0\,,\\
\frac{\xi^{2}z}{2}\left.\frac{\partial p}{\partial z}\right\rfloor _{z=0}+\left.\kappa zp\right\rfloor _{z=0}+\rho\xi z\left.\frac{\partial x\alpha\left(x,t\right)p}{\partial x}\right\rfloor _{z=0}=0, & z=0,\,t>0\,,\\
\lim_{z\rightarrow\infty}p(x,z,t)=\lim_{x\rightarrow\infty}p(x,z,t)=p(0,z,t)=0, & \left(x,z\right)\in\Omega,\,z\neq0,\,t>0\,,\\
\lim_{t\rightarrow0}p(x,z,t)=z^{-\beta}\delta(x-S_{0},z-v_{0}), & \left(x,z\right)\in\Omega\,.
\end{cases}\label{eq:KolmogorovForward-HestonLSV-Reformulated}
\end{equation}
\end{corollary}
While this PDE is easier to handle numerically, one wants
to work with the original density function $\phi$ for most of the
applications. 
There are two main calculations one would like to achieve:
the expected payoff $f\left(S_{T}\right)$ for a given function $f$;
the Markovian projection $\mathbb{E}\left[V_{T}\,|\,S_{T}=K\right]$.

As $\phi$ is still intractable for small $z$ and computing $z^{\beta}p\left(x,z,t\right)$
is not numerically feasible, we perform an integration by parts (noticing $\lim_{z\rightarrow0}z^{\beta+1}p=\lim_{z\rightarrow \infty}z^{\beta+1}p=0$ since $\beta+1>0$)
to obtain, for deterministic rates,
\begin{equation}
\mathbb{E}\left[D_{T}f\left(S_{T}\right)\right]=-D_{T}\int_{0}^{\infty}f\left(x\right)\int_{0}^{\infty}\frac{z^{\beta+1}}{\beta+1}\frac{\partial p(x,z,T)}{\partial z}\,dz\,dx \label{eq:dampenedPDEMarkovProj0}
\end{equation}
and
\begin{equation}
\mathbb{E}\left[V_{T}\,|\,S_{T}=K\right]=-\left(\beta+1\right)\frac{\int_{0}^{\infty}z^{\beta+1}p\left(K,z,T\right)\,dz}{\int_{0}^{\infty}z^{\beta+1}\frac{\partial p(K,z,T)}{\partial z}\,dz}\,.\label{eq:dampenedPDEMarkovProj}
\end{equation}

\subsection{Finite element method with a two-step BDF time scheme}

We combine a finite element approximation in space with a \emph{Backward
Differentiation Formula} (BDF) scheme in time,
since Crank-Nicholson time-stepping or ADI schemes can give rise to
instabilities for Dirac initial data (see \cite{pooley2003convergence, wyns2015convergence};
we refer to \cite{bokanowski2017stability} for a stability analysis and to \cite{Floch} for some financial applications of BDF schemes).

Equation (\ref{eq:KolmogorovForward-HestonLSV-Reformulated}) can be written as
{
\begin{eqnarray*}
\label{pdeabstract}
&&\frac{\partial p}{\partial t}-z\nabla\mathbf{\cdot}\mathbf{u}+\nabla\cdot\left(\mathbf{b} -\left(\beta+1\right)\mathbf{w} \right)+ c p=0\,,
\vspace{0.3 cm} \\
\nonumber
&&\mathbf{u=}\frac{1}{2}\begin{bmatrix}\frac{\partial x^{2}\alpha^{2}\left(x,t\right)p}{\partial x}+\frac{\partial\rho\xi x\alpha\left(x,t\right)p}{\partial z}\\
\frac{\partial\xi^{2}p}{\partial z}+\frac{\partial\rho\xi x\alpha\left(x,t\right)p}{\partial x}
\end{bmatrix}\!, \;
\mathbf{b}=\begin{bmatrix}(\bar{r}^{d}\left(t\right)-\bar{r}^{f}\left(t\right))x\\
\kappa\left(\theta-z\right)
\end{bmatrix}\!, \;
\mathbf{w}=\begin{bmatrix}\rho\xi x\alpha\left(x,t\right)p\\
\xi^{2}p
\end{bmatrix}\!, \;
c = -\beta\kappa\,.
\end{eqnarray*}
Denote 
by $\Gamma_{R}=\left\{ \left(x,z\right)\in\partial\Omega:\,z=0\right\} $
the subset of the boundary of $\Omega$ with Robin boundary condition.
We derive a weak formulation in the usual way (see, e.g., \cite{quarteroni2008numerical}), i.e.,
we multiply the PDE by a test function $v\in H^{1}\left(\Omega\right)$,
integrate over $\Omega$, using the divergence theorem and boundary conditions, to obtain the weak form of (\ref{pdeabstract}),
\begin{eqnarray*}
\int_{\Omega}\frac{\partial p}{\partial t}v\,d\Omega+a(p,v) &=& 0,
\end{eqnarray*}
with the bi-linear form 
\begin{eqnarray*}
a\left(p,v\right)&=&\int_{\Omega} 
\mathbf{u}\mathbf{\cdot\nabla} (z v)  + 
\left(\nabla\cdot\left(\mathbf{b}-\left(\beta+1\right)\mathbf{w}\right)+  c p \right) v\,d\Omega \,- 
\, \int_{\Gamma_{R}}\left(\kappa zp + \frac{1}{2}z\frac{\partial\rho\xi x\alpha\left(x,t\right)p}{\partial x}\right)v\,d\Gamma_{R},
\end{eqnarray*}
}
where the last term contains the new boundary condition.

%
Let us define a
uniform time mesh with $t_{m}=m\Delta_{t}$, $m\in\llbracket0,M\rrbracket$. 
We denote $p_{m}=p\left(\cdot,\cdot,t_{m}\right)$, in which case the BDF scheme can be written as
\[
\int_{\Omega}\left(p_{m+2}-\frac{4}{3}p_{m+1}+\frac{1}{3}p_{m}\right)v\,d\Omega+\frac{2}{3}\Delta_{t}a(p_{m+2},v)=0, \qquad m=0,\ldots, M-2,
\]
where the first time step is divided into two standard fully implicit
time steps (see \cite{gilescarter}). 
Then, for the first time step, using the Dirac delta initial condition, 
\[
\int_{\Omega}p_{1}v\,d\Omega+\Delta_{t}a(p_{1},v)=\int_{\Omega}p_{0}v\,d\Omega=v_{0}^{-\beta}v\left(S_{0},v_{0}\right).
\]
Note that we initially only assumed $v \in H^1 \not\subset C$, and therefore the operation above with the Dirac delta is not defined for all such $v$.
However, we will next use continuous basis functions.
If $\left(S_{0},v_{0}\right)$ coincides with a mesh point,
this is equivalent to solving a linear system where the right-hand
side vector is $v_{0}^{-\beta}$ for the source point node and zero otherwise.

The PDE solution is approximated by a conforming finite element method with $\mathbb{P}_{2}$ elements, i.e., a polynomial
of order two on a triangle cell. The computation of the finite element solutions was performed with the FEniCS library \cite{Fenics2015}. 
Each triangle is characterised by
6 local degrees of freedom (nodes) as displayed in Figure \ref{fig:fem-element}
(see \cite{FEMBook2005} for details).
We describe the mesh construction in detail in Appendix \ref{sub:Mesh-construction}.
An example of a thus generated mesh
with 30 spot steps and 30 variance steps is illustrated in Figure \ref{fig:Finite-Element-Mesh-1}.

\noindent
\begin{figure}
	\hspace{-1.5 cm}
	\begin{centering}
		\includegraphics[scale=0.37]{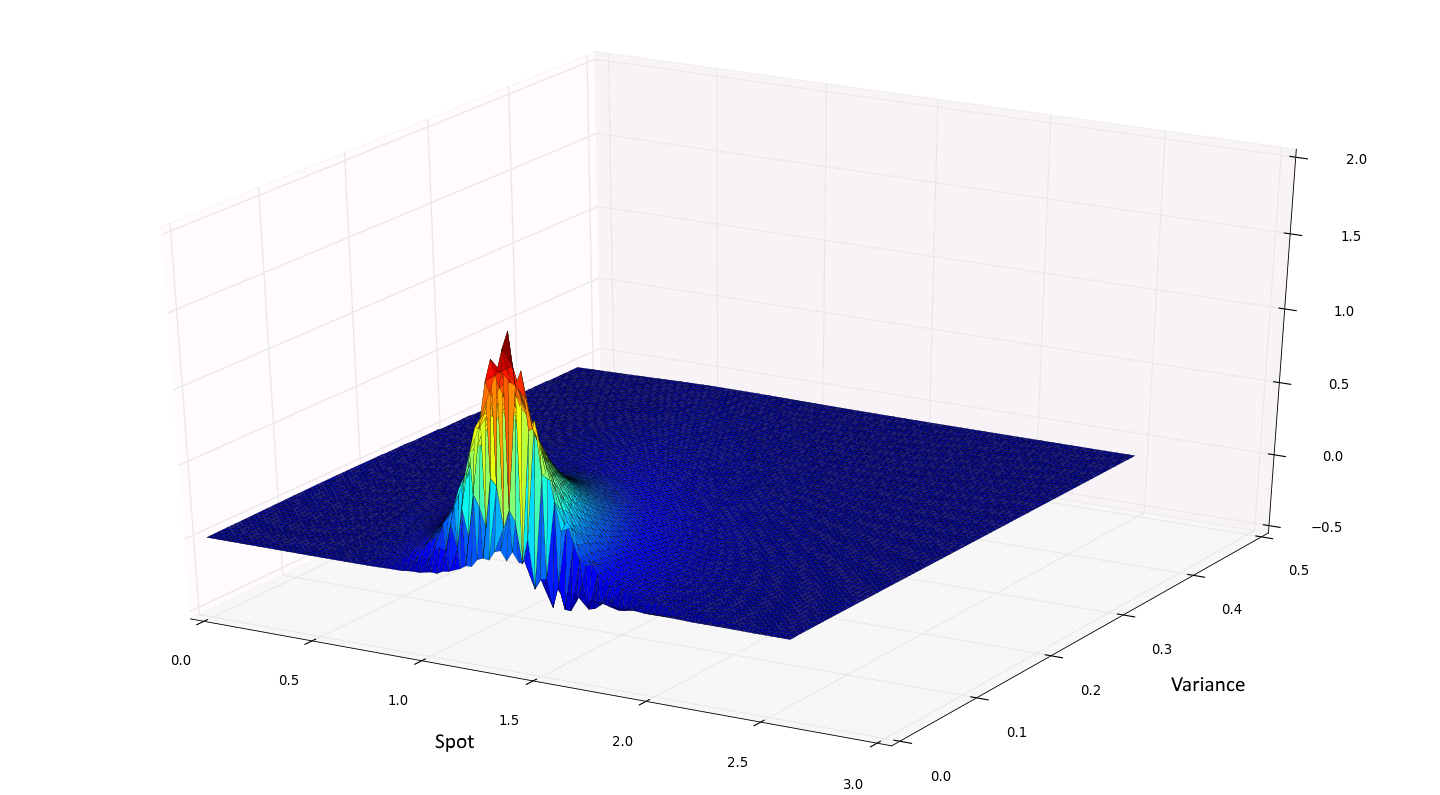}
		\hspace{-1.5 cm}
		\includegraphics[scale=0.25]{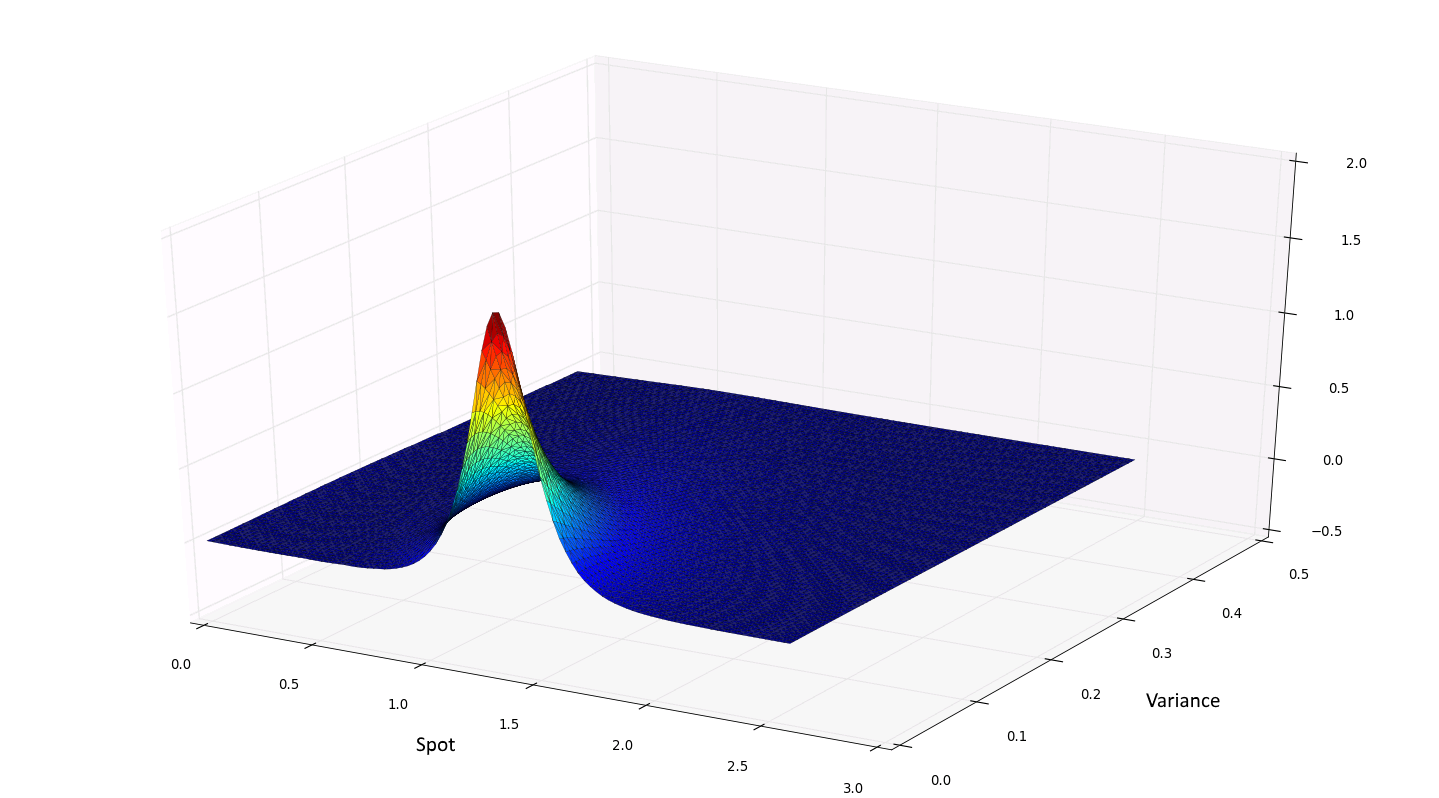}
		\captionof{figure}{$p=z^{-\beta}\phi$ using $\phi$ computed with the standard 
			(left) and change of variables PDE (right)}\label{fig:-computed-=0000E0-posteriori}%
	\end{centering}
\end{figure}

\begin{minipage}[c]{1\columnwidth}%
\begin{minipage}[c]{0.23\columnwidth}%
\begin{center}
~\\
~\\
\includegraphics[scale=0.7]{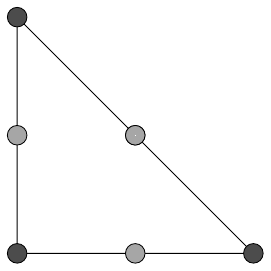}\\

\par\end{center}

\captionof{figure}{$\mathbb{P}_{2}$ element with 6 degrees of freedom}

\label{fig:fem-element}%
\end{minipage}~~~~%
\begin{minipage}[c]{0.75\columnwidth}%
	\begin{centering}
		\includegraphics[scale=0.32]{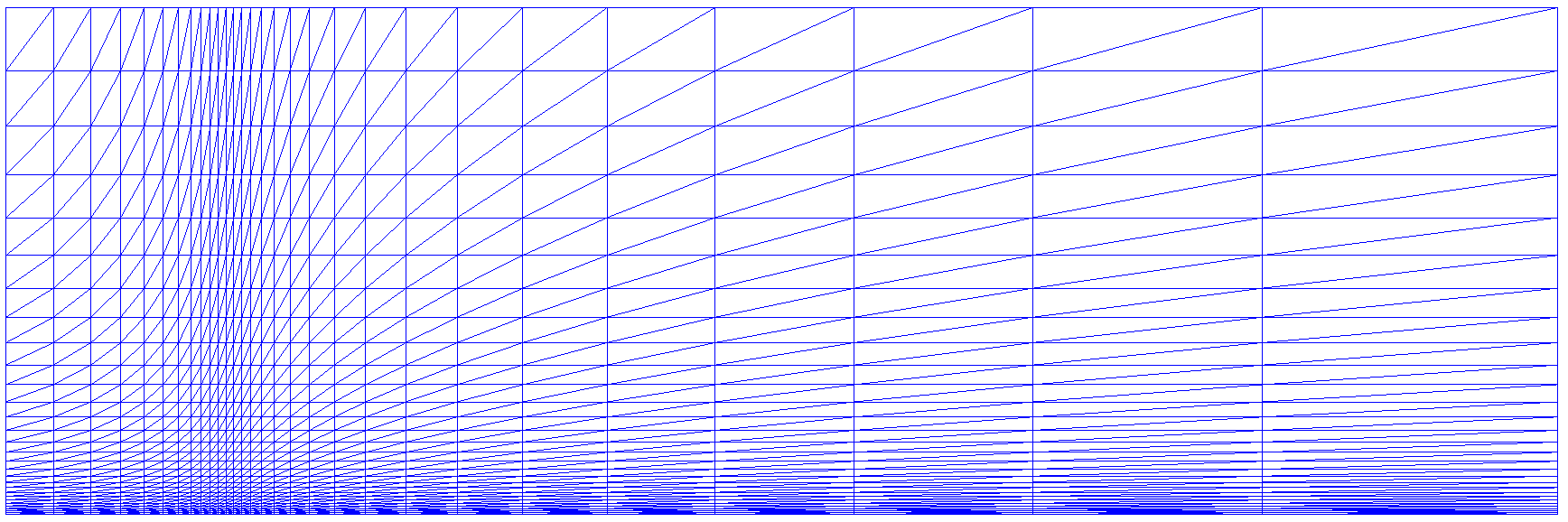}\captionof{figure}{Finite element triangular mesh refined around $x=S_{0}$ and $z=0$}\label{fig:Finite-Element-Mesh-1}%
	\end{centering}
\end{minipage}%
\end{minipage}

In order to see the improvement due to the transformed PDE (\ref{eq:KolmogorovForward-HestonLSV-Reformulated})
for $p$, we plot both $z^{-\beta}\phi$ and $p$ for a pure
Heston model with
\[
r=3\%,\quad q=1\%,\quad\kappa=1,\quad\theta=v_{0}=0.04,\quad\rho=-0.3,\quad\xi=0.5,\quad t=1\,,
\]
which corresponds to a Feller ratio of $0.32$. We use 100 time steps
as well as 30 spot steps and 30 variance steps. The solution for $p=z^{-\beta}\phi$
computed with no change of variables is presented in Figure \ref{fig:-computed-=0000E0-posteriori},
where we notice significant numerical instabilities. The tranformed PDE is solved for the same problem and $p$ is also plotted in
Figure \ref{fig:-computed-=0000E0-posteriori}. 

%
%
%

\subsection{Calibration algorithm}

The calibration of the 2-factor LSV model (\ref{eq:modelDefinitionLSVDetRates})
is performed by finding the leverage function $\alpha$ defined in 
(\ref{eq:alphaFormula2D}).
We compute $\mathbb{E^{\mathbb{Q}}}\left[V_{T}\mid S_{T}=K\right]$
from (\ref{eq:dampenedPDEMarkovProj})
with the solution $p$ of the dampened PDE (\ref{eq:KolmogorovForward-HestonLSV-Reformulated}).
Both integrals can
be computed by double adaptive Clenshaw-Curtis quadrature rules to
handle singularities properly when $T$ is small (see \cite{Gough2009}).

Furthermore, for very small or very large values of  $K$, both the numerator
and denominator will be very small. So we define a smooth extrapolation
rule by
\[
\mathbb{E^{\mathbb{Q}}}\left[V_{T}\,|\,S_{T}=K\right] \approx \frac{\left(\beta+1\right)\left(\int_{0}^{\infty}z^{\beta+1}p\left(K,z,T\right)\,dz+\epsilon\left(v_{0}e^{-\kappa T}+\left(\theta-v_{0}\right)\left(1-e^{-\kappa T}\right)\right)\right)}{\epsilon\left(\beta+1\right)-\int_{0}^{\infty}z^{\beta+1}\frac{\partial p(K,z,T)}{\partial z}\,dz}\, ,
\]
where we pick $\epsilon=10^{-14}$ in our numerical tests.

%
%
%
%
%

The calibration will be done forward in time. Denote by $\left(\Delta T\!\left(i\right)\right)_{i\leq N_{T}}$
the interval lengths between maturities and by $N_{T}$ the number of maturities.

The leverage function $\alpha$ is again defined by splines as detailed in Section
\ref{sub:Calibration-Algorithm-with-CVParticleMethod}.
This approach allows us
to compute $\mathbb{E^{\mathbb{Q}}}\left[V_{T}\,|\,S_{T}=K\right]$
only on the nodes, reducing the computational time considerably.
In the calibration routine, we use forward constant interpolation of the leverage function between maturities
to handle the non-linearity of the problem.
We can then write the calibration procedure as in Algorithm \ref{alg:calibWithDampedPDE}.

\begin{algorithm}
\caption{$\alpha\left(s,T\right)$ calibration with the dampened Kolmogorov
forward PDE}
\label{alg:calibWithDampedPDE}

\begin{algor}
\item [{{*}}] $\alpha\left(s,0\right)=\frac{\sigma_{LV}\left(s,0\right)}{\sqrt{v_{0}}}$
\item [{{*}}] $T=0$
\item [{for}] ( $i=1\,;\,i\leq N_{T}\,;\,i++$) 

\begin{algor}
\item [{{*}}] \begin{raggedright}
\textbf{solve} (\ref{eq:KolmogorovForward-HestonLSV-Reformulated})
for $p$ on $\left[T,T+\Delta T\right]$ with $\alpha\left(s,\left[T,T+\Delta T\left(i\right)\right]\right)=\alpha\left(s,T\right)$
\par\end{raggedright}
\item [{for}] ( $j=1\,;\,j\leq N_{S}\,;\,j++$)

\begin{algor}
\item [{{*}}] 
\[
E_{V_{T}}=-\left(\beta+1\right)\frac{\int_{0}^{\infty}z^{\beta+1}p\left(s_{i,j},z,T+\Delta T\left(i\right)\right)\,dz}{\int_{0}^{\infty}z^{\beta+1}\frac{\partial p(s_{i,j},z,T+\Delta T\left(i\right))}{\partial z}\,dz}
\]

\item [{{*}}] 
\[
\alpha_{i,j}=\frac{\sigma_{LV}\left(s_{i,j},T+\Delta T\left(i\right)\right)}{\sqrt{E_{V_{T}}}}
\]

\end{algor}
\item [{endfor}]~
\item [{{*}}] $T=T+\Delta T\left(i\right)$
\end{algor}
\item [{endfor}]~\end{algor}
\end{algorithm}


In the test, we use the Heston parameters calibrated in Appendix
\ifarxivvar
\ref{sub:Heston-2CIR++-Calibration}
\else
F of \cite{cmr18}
\fi
for the Heston-2CIR++ model,
i.e.,
\[
v_{0}=0.0094,\quad\theta=0.0137,\quad\kappa=1.4124,\quad\rho=-0.1194,\quad\xi=0.2988\,,
\]
where the Feller ratio is 
\[
\frac{\mbox{2\ensuremath{\kappa\theta}}}{\xi^{2}}\approx0.4335<1\,,
\]
which violates the Feller condition.

The calibrated leverage function $\alpha$ is plotted in Figure \ref{fig:leverage2D}.
For the solution of the forward PDE between maturities, we use a BDF scheme with constant
stepsize and find that $50$ time steps per year and a $80\times80$
spot-variance finite element mesh give very accurate results.

\begin{figure}[h]
\begin{centering}
\includegraphics[scale=0.23]{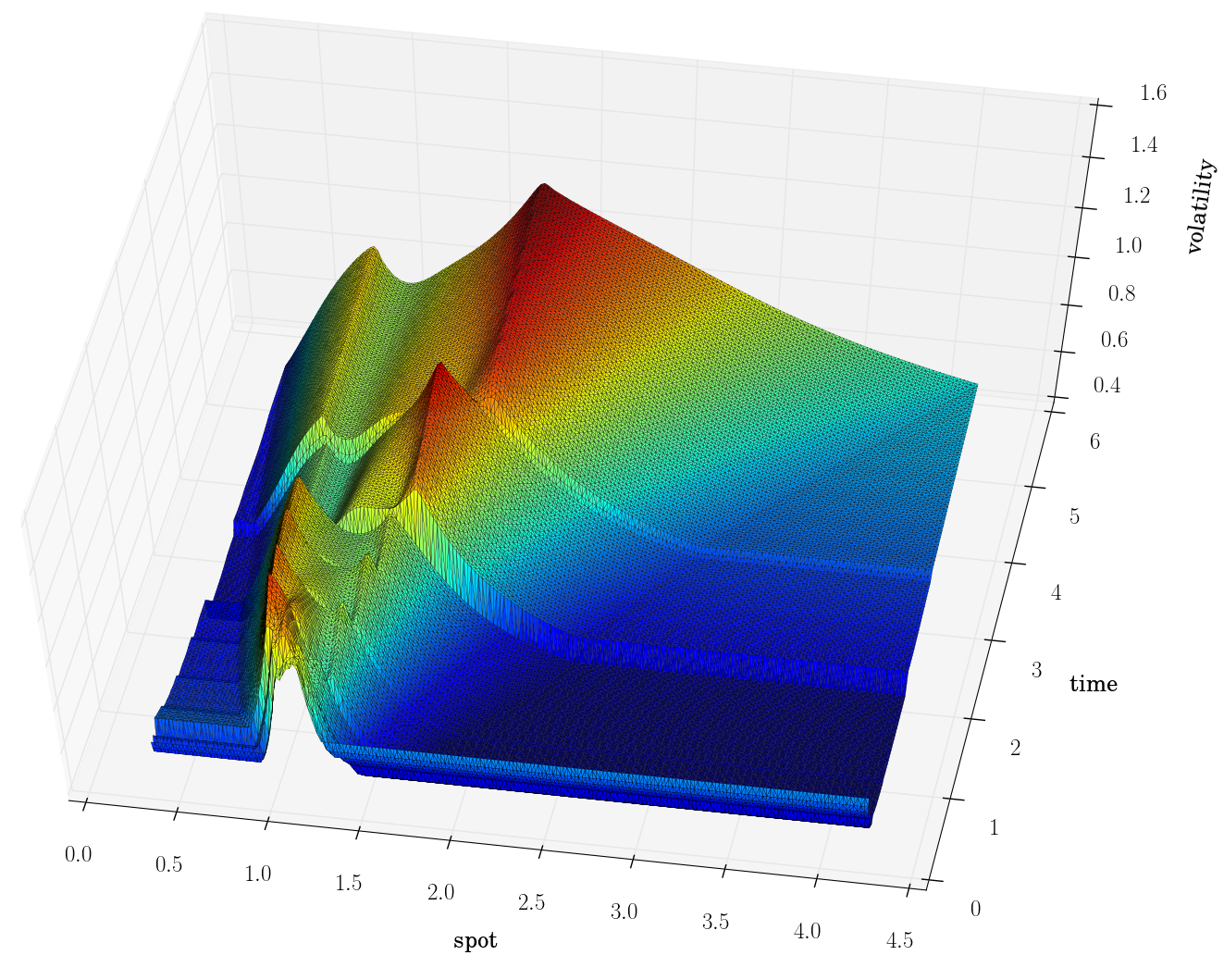} \hfill
\includegraphics[scale=0.33]{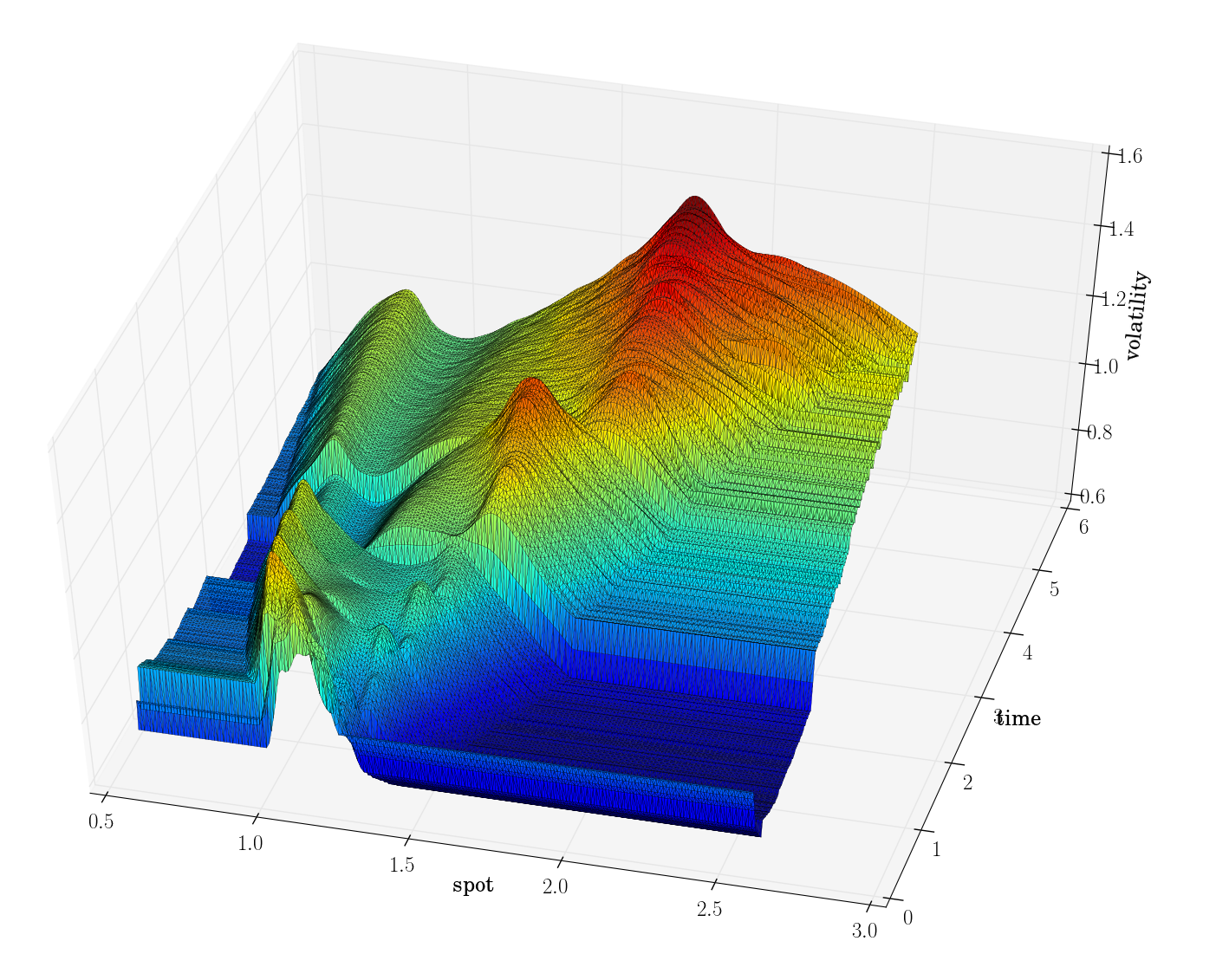}
\par\end{centering}

\caption{Calibrated leverage function $\alpha$ for (left) the 2-factor LSV model
(\ref{eq:modelDefinitionLSVDetRates})
and (right) the 4-factor LSV model
(\ref{eq:Model Definition}).
\label{fig:4DLSVLeverage}
\label{fig:leverage2D}}
\end{figure}

%
%

\section{Four-factor Heston-type LSV 2CIR++ model calibration\label{sec:Calibration-Results}}

In this section, we give the results for the calibration of the main
model (\ref{eq:Model Definition}) and test the efficiency of the algorithm.

We use vanilla options implied volatility data from Bloomberg from the $18/03/2016$ for the currency
pair $\text{EURUSD}$,
namely,
$10\text{D-Put},25\text{D-Put},50\text{D},\,25\text{D-Call},\,10\text{D-Call}$,
for the maturities\footnote{We skip the $\text{7Y}$ and $\text{10Y}$ quotes as they were not liquid enough.}
$\text{3W, 1M, 2M, 3M, 6M, 1Y, 1Y6M, 2Y, 3Y, 5Y.}$

%
%

%
We use historical correlations estimated in \cite{DeGraff2016} from weekly time series data from 2012--2014,
\[ \rho_{Sd}=-0.3024, \qquad \rho_{Sf}=0.1226, \qquad \rho_{df}= 0.6293 \,.\]
We assume that both CIR\scalebox{.9}{\raisebox{.5pt}{++}} processes
are calibrated under their own risk-neutral measure as in
Appendix
\ifarxivvar
\ref{subsec:rate},
\else
D in \cite{cmr18},
\fi
and that the Heston--2CIR\scalebox{.9}{\raisebox{.5pt}{++}}
model is calibrated as in 
Appendix
\ifarxivvar
\ref{sub:Heston-2CIR++-Calibration}
\else
F of \cite{cmr18}
\fi
for the Heston-2CIR++ model.
The parameters are as follows 
\[
\begin{cases}
v_{0}=0.0094,\quad\theta=0.0137,\quad\kappa=1.4124,\quad\rho=-0.1194,\quad\xi=0.2988\,,\\
g_{0}^{d}=0.0001,\quad\theta_{d}=0.5469,\quad\kappa_{d}=0.0837,\quad\rho_{Sd}=0,\quad\xi_{d}=0.0274\,,\\
g_{0}^{f}=0.0001,\quad\theta_{f}=1.1656,\quad\kappa_{f}=0.0110,\quad\rho_{Sf}=0,\quad\xi_{f}=0.0370\,.
\end{cases}
\]

In order to approximate the particle system (\ref{part_sys}), 
we use an extension of the QE-scheme from \cite{Andersen2008} to model (\ref{eq:Model Definition}).
A full description of the time marching scheme is provided in Appendix \ref{appendix:Monte-Carlo-QE-scheme}.

\subsection{Calibration results and efficiency}

For a first illustration of the model fit and
the improvement through the control variates, we calibrate the 4-factor model with $800$ particles
with and without control variates. 
The associated leverage function is plotted in Figure \ref{fig:4DLSVLeverage} (right).

We then plot, in Figure \ref{fig:calibrationFit-4000CVVSNOCV},
the model implied volatility slices for 3M, 1Y, 2Y, and 5Y.
The figure shows a significantly improved fit due to the control variates.
We will analyse the accuracy and convergence in detail in Subsection \ref{subsec:error_red}.

\begin{figure}[h]
\begin{centering}
\includegraphics[scale=0.27]{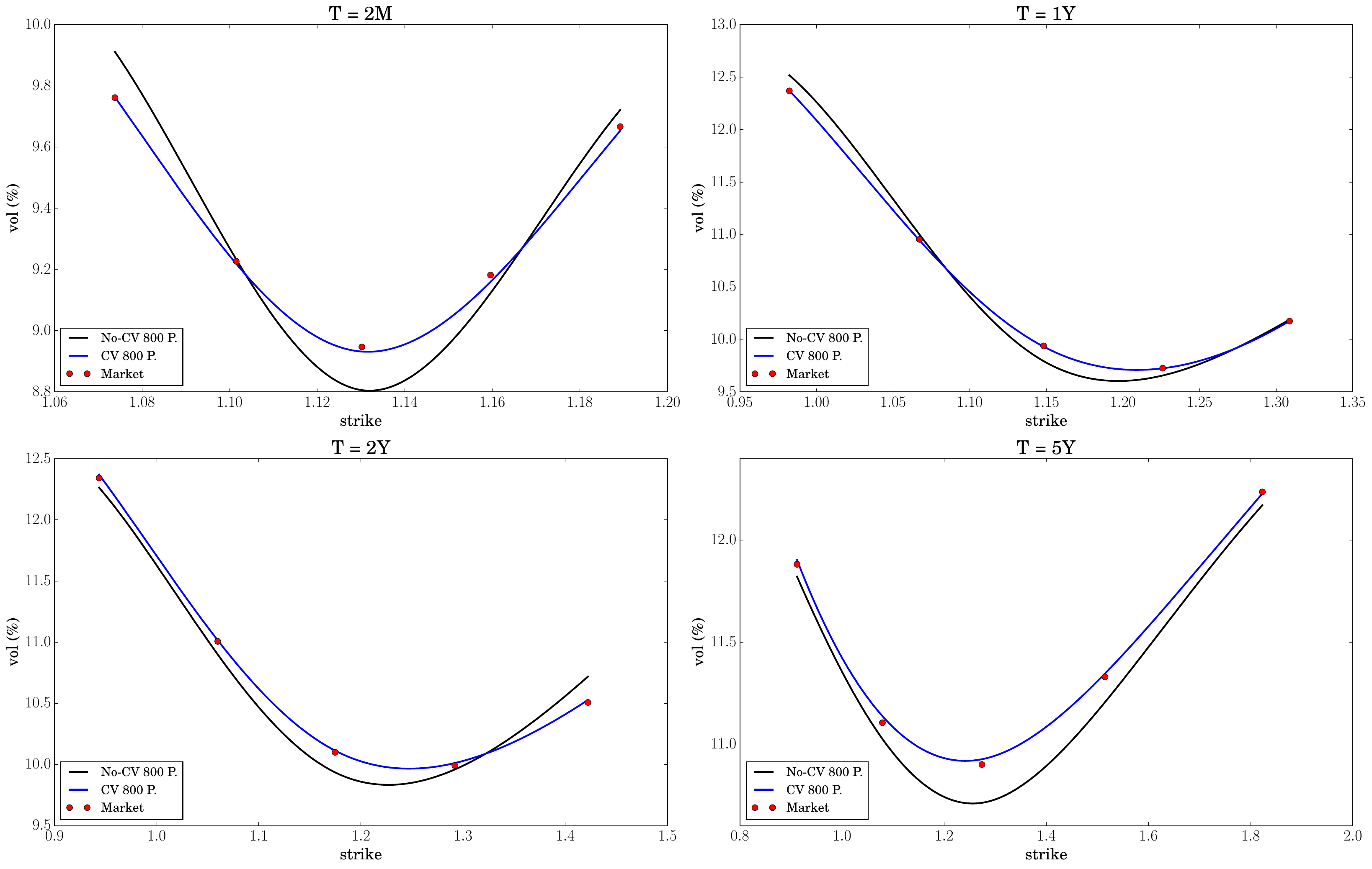}
\par\end{centering}

\centering{}\caption{Calibration fit with 800 particles for the 4-factor LSV model with and
without control variates \label{fig:calibrationFit-4000CVVSNOCV}}
\end{figure}

As regards computational time, the control variate particle method with 20000 particles, which are enough for a reasonably converged solution (see Subsection \ref{subsec:error_red}), took approximately $10$\% of the overall time spent in the PDE calibration of the LSV(2D) model. Most practitioners are familiar with the computational time required to calibrate a 2D LSV model by a forward PDE, a rough estimate being below one second for short- to mid-term
expiries (5Y).
The extra cost to calibrate the full 4D LSV--2CIR++ model is almost negligible with control variates, while the same accuracy without control variates requires more than 60 times the cost of the 2D PDE solver.

\subsection{Variance and error reduction}
\label{subsec:error_red}

In this subsection, we 
compare the results with control
variates (CV) to those with the plain particle method (No-CV) as
a function of the number of particles.
The error measure we use is the absolute error in volatility (in $\%$ units).
For instance, a maximum error
(taken over all quoted deltas and maturities)
 of $0.03\%$ for a $20\%$ market volatility
signifies that the calibrated volatility can be $20.00\%\pm0.03\%$
in the worst case scenario.

We use the calibration routine described in Subsection \ref{sub:Calibration-Algorithm-with-CVParticleMethod}
with and without control variates for 160, {$\text{800}$,}
$\text{4\,000}$, $\text{20\,000}$, $\text{100\,000}$, $\text{500\,000}$
and $\text{2\,500\,000}$ particles.


The results are presented in Figure \ref{fig:CV VS NoCV}.
On the basis that higher correlations between spot-rates and rate-rate will make our control variates more efficient,
we also display in Figure \ref{fig:CV VS NoCV} the calibration errors with no correlations between spot-rates and rate-rate as
a presumed worst-case.

\begin{figure}[h]
\begin{centering}
\includegraphics[scale=0.43]{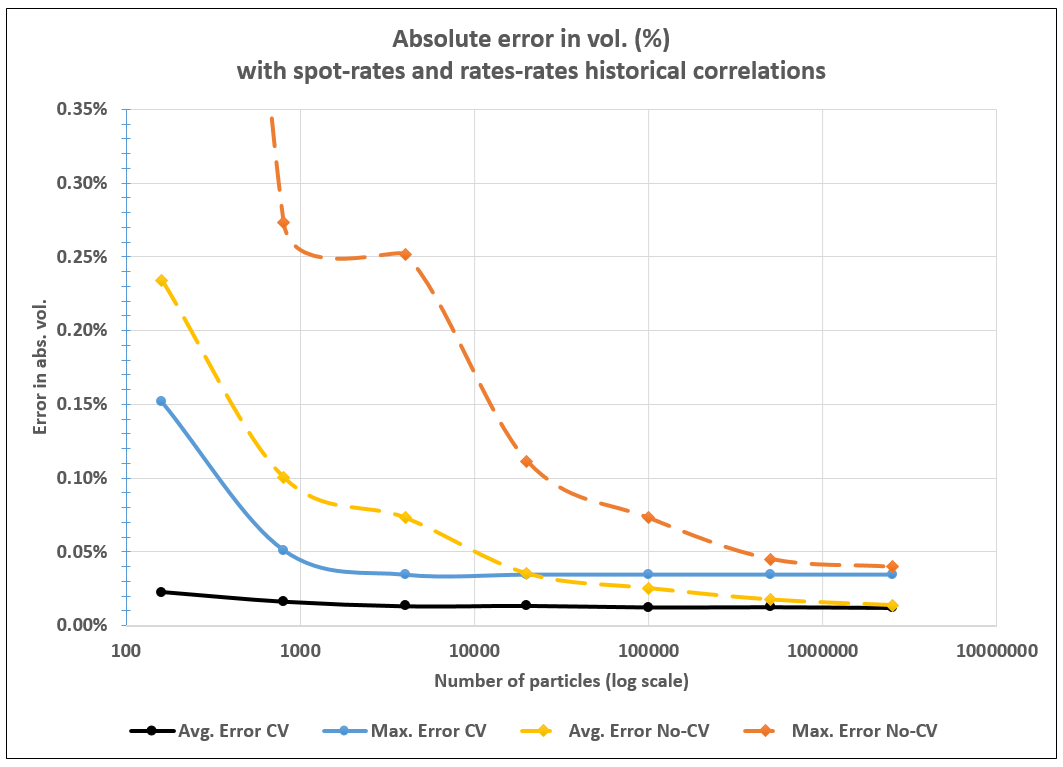}
\includegraphics[scale=0.43]{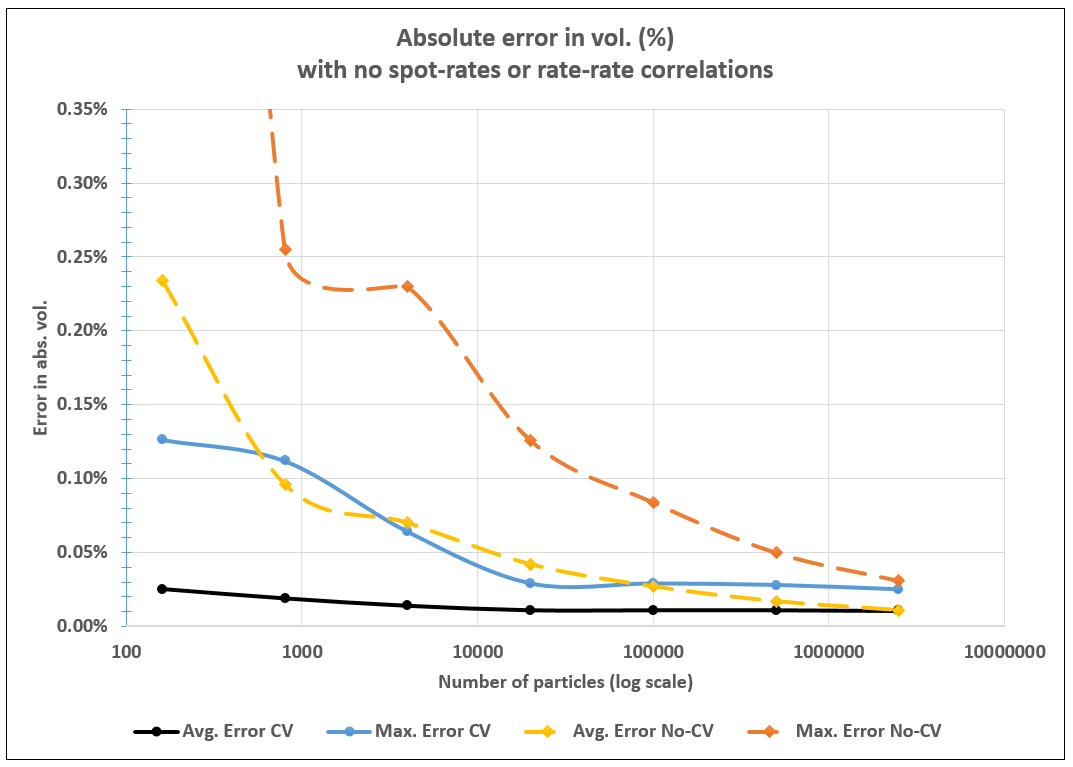}

\par\end{centering}

\caption{Error convergence with and without control variates, either with historical $\rho_{Sd}, \rho_{Sf}, \rho_{df}$ correlations (left) or $\rho_{Sd}=0, \rho_{Sf}=0, \rho_{df}=0$ correlations (right). Error is computed
on all quoted deltas and maturities (log-scale) \label{fig:CV VS NoCV}}
\end{figure}
We infer from the data in Figure \ref{fig:CV VS NoCV} that the use
of control variates greatly improves the general calibration routine.
Convergence in both the maximum error ($0.032\%$)
and in the average error ($0.012\%$) is reached
(i.e., the error from here on is dominated by other
sources, such as the time discretisation error)
for $\text{4\,000}$
particles when using control variates, while the plain
particle method (without variance reduction) only reaches the same
accuracy with $\text{2\,500\,000}$ particles. 
We infer that the control
variates give a $625$-fold speed-up.

Moreover,
from careful data analysis of Figure \ref{fig:CV VS NoCV}, we find
a convergence rate of $0.3$ in the number of particles for
the plain particle method, for both the average and maximum error.
The addition of the control variates preserves the convergence rate but reduces the absolute size of the error significantly. 

We note that the error can be further reduced
by increasing the number of time steps per year. We used $250$ time steps
(such that, on average, there is one
time step per open day) as it already yields a very
accurate calibration at a reasonable computational cost.

Part of the error reduction is due to the conditional control variate
described in Subsection \ref{sub:Variance-Reduction-for-Markovian-Projection}
which can provide very good results for short-term horizons (the other part being due to
the control variates for standard expectations).
In order
to analyse this further, we plot the variance reduction factor from (\ref{eq:variance-reduction-factor}),
estimated with $\text{500\,000}$ particles,
as a function of time in Figure \ref{fig:varReductionFactorInTime}. We estimate
a trend line $1+\frac{C}{T^{1.6}}$, for a given constant $C$, and
thus a very good variance reduction for short-term options.
For longer terms, this still yields a good variance reduction factor
of $5.2$ for the $3Y$ maturity and $2.34$ for $5Y$.

In addition, the other two control variates presented in Subsection
\ref{sub:Variance-Reduction-for-Standard-Expectations} help
control the short- to long-term behaviour as well. 
We plot
both variance reduction factors for $X_{1}^{*}$ and $X_{2}^{*}$ in Figure \ref{fig:varReductionFactorInTime}.
They 
seem to reach
a steady state for maturities around $1.6$ and $6.5$, respectively.
We note that we displayed here average variance reduction values over all strikes, 
while $X_{1}^{*}$ was found to provide significant variance reduction for small strikes.

\begin{figure}[h]
	\begin{centering}
		\includegraphics[scale=0.43]{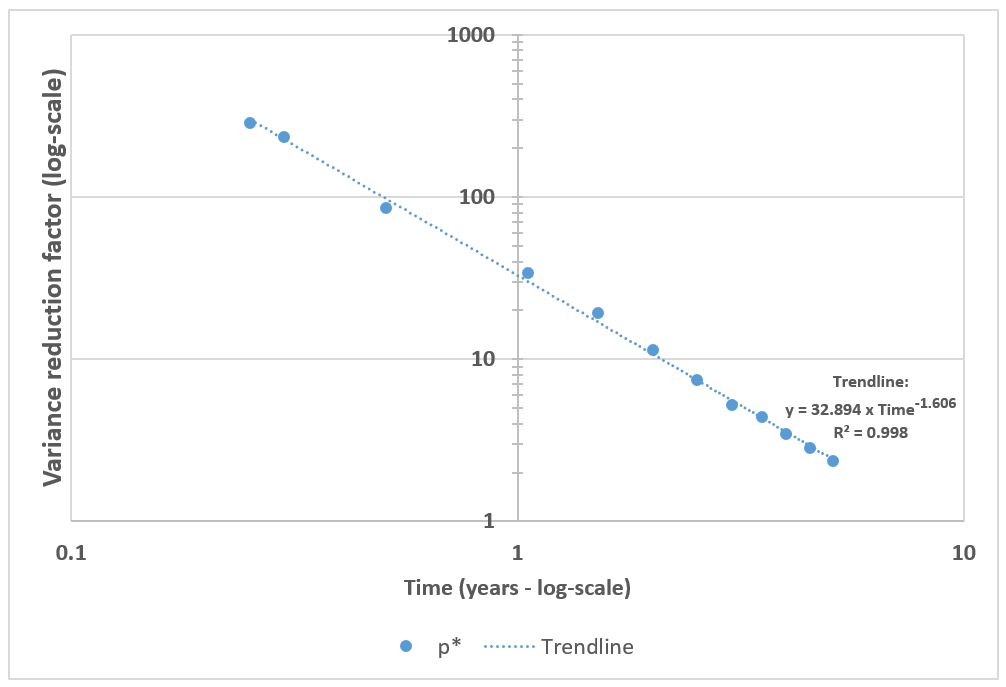} \hfill
		\includegraphics[scale=0.395]{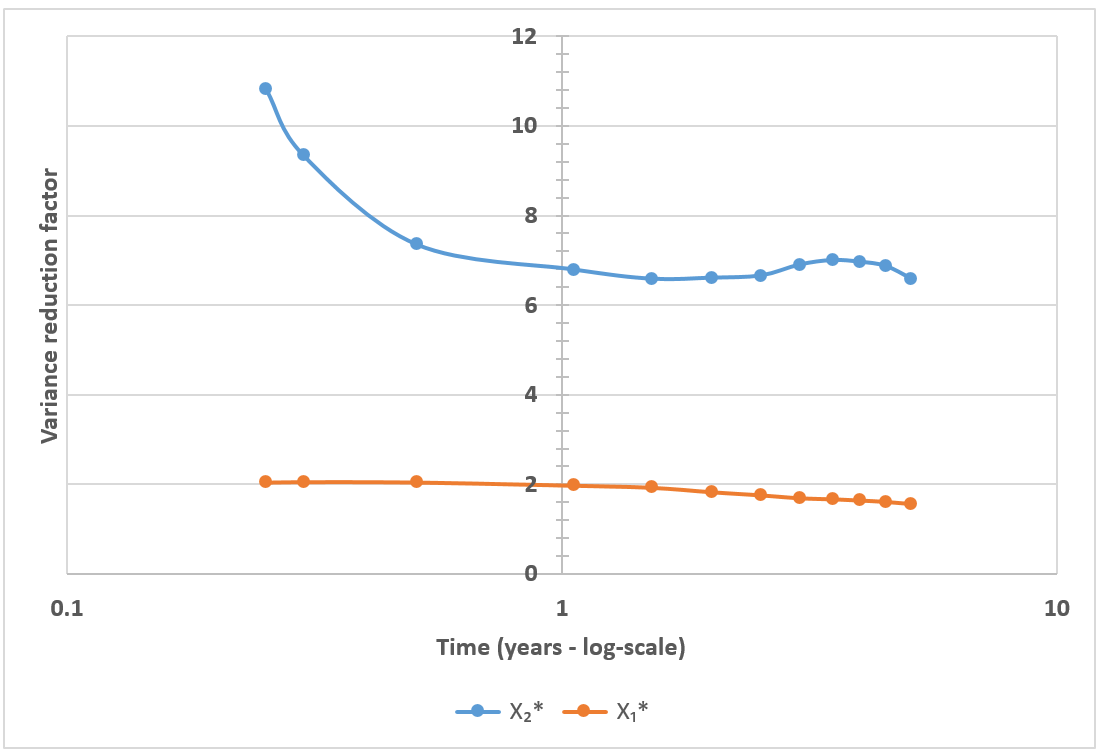}
		\par\end{centering}
	
	\caption{Variance reduction factor as a function of time for the conditional
		expectation estimator $p_{N}^{*}$ (left, using a log-log scale) and the standard
		expectations estimators $X_{1}^{*}$ and $X_{2}^{*}$ (right, using a log scale)
		\label{fig:varReductionFactorInTime}}
\end{figure}

\subsection{Stress scenario\label{sub:Stress-scenario}}

The last calibration result we present is for an extremely stressed
set of rates parameters. This will serve as a robustness test of the
control-variate particle method for very volatile short rate processes.
In order to perform the test, we multiply the calibrated $\xi_{d}$
and $\xi_{f}$ by $20$ as well as divide $\theta_{d}$ and $\theta_{f}$
by $20$. This leads to strongly violated Feller conditions and very high volatilities for the two CIR\scalebox{.9}{\raisebox{.5pt}{++}}
rate processes. We also set $\rho_{Sd}=0, \rho_{Sf}=0, \rho_{df}=0$ as this has shown to be more challenging for our method (as it makes the control variates less effective).  The stress scenario parameter values are 
\[
\begin{cases}
v_{0}=0.0094,\quad\theta=0.0137,\quad\kappa=1.4124,\quad\rho=-0.1194,\quad\xi=0.2988\,,\\
g_{0}^{d}=0.0001,\quad\theta_{d}=\mathbf{0.0273},\quad\kappa_{d}=0.0837,\quad\rho_{Sd}=0,\quad\xi_{d}=\mathbf{0.5480}\,,\\
g_{0}^{f}=0.0001,\quad\theta_{f}=\mathbf{0.0582},\quad\kappa_{f}=0.0110,\quad\rho_{Sf}=0,\quad\xi_{f}=\mathbf{0.7400}\,.
\end{cases}
\]
We emphasise that in practice, the volatility parameters $\xi_{d/f}$
are rarely above $0.06$. We refer the reader to \cite{Brigo2006} for more
details. A calibration summary for the average error in absolute volatility
is displayed in Table \ref{tab:stress_test}. In this stress scenario,
the calibration via control-variate particle method reaches an error of $0.0165\%$
for $20\,000$ particles, whereas $2\,500\,000$ particles are required
with the plain particle method (without control variates) to reach
the same accuracy.
Hence, the control-variate particle method shows a consistent improvement
over the plain particle method even under stress scenarios. The conditional control variate for $p_{N}^{*}$ yields
a variance reduction factor of almost $2$ for the last maturity
($5\text{Y}$), similar to the correlated case. 

\begin{table}[h]
\begin{centering}
\begin{tabular}{cccccc}
\toprule 
\textbf{$N$} & \textbf{$4\,000$} & \textbf{$20\,000$} & \textbf{$100\,000$} & \textbf{$500\,000$} & \textbf{$2\,500\,000$}\tabularnewline
\midrule 
\textbf{Plain particle method} & $0.0855\%$ & $0.0440\%$ & $0.0320\%$ & $0.0255\%$ & $0.0160\%$\tabularnewline
\textbf{Control variate particle method} & $0.0312\%$ & $0.0165\%$ & $0.0133\%$ & $0.0133\%$ & $0.0133\%$\tabularnewline
\bottomrule
\end{tabular}
\par\end{centering}

\centering{}\caption{Average error in absolute volatility (\% unit) for a high-volatility
stress scenario on the rate processes. $N$ is the number of particles.
\label{tab:stress_test}}
\end{table}

\subsection{Impact of stochastic rates}
\label{sec:pricing}
In this subsection, we discuss the pricing of more exotic products,
namely a no-touch option and a target accrual redemption
note (TARN); see Chapter 8 in \cite{Clark2010} and Section 2.2 in \cite{Wystup2007} respectively for a discussion of these products.
Specifically, we assess the impact of stochastic rates on
products embedding knock-out features with mid- to long-term expiries.

\subsubsection*{No-touches}
The foreign no-touch up option pays one EUR at maturity if the exchange
rate has not breached an upper barrier during the product lifespan.
The payout under an arbitrage-free model
	with foreign risk-neutral measure $\mathbb{Q}^{f}$ ,
	for a premium expressed in foreign units (EUR), a notional $N_{\text{EUR}}$ in
	foreign units (EUR) and for a maturity $T$ is 
	\[
	N_{\text{EUR}}\mathbb{E}^{\mathbb{Q}^{f}}\left[D_{T}^{f}\mathbf{1}_{M_{T}<B_{1}}\right]\,,
	\]
	where $M_t = \sup_{0\le u\le t} S_u$ is the running-maximum of the spot $S$ and $B_{1}$ is the
	upper barrier. For the tests, we pick 
	\[
	T=5.0,\qquad B_{1}=1.4\times S_{0}\,.
	\]

\subsubsection*{TARNs}
The specification of the TARN is here as follows: the buyer receives
the forward value $K-S_{t_{i}}$ at fixing date $t_{i}$ if $K>S_{t_{i}}$;
the buyer has to pay $S_{t_{i}}-K$ if $K<S_{t_{i}}$. At each fixing
date, if the amount received is positive, the accrued value is increased
by the paid amount. If at some point in the deal life-cycle, the accrued
amount breaches a target $H_{target}$, the deal is terminated early.
To protect the buyer, an additional knock-out barrier redeems
the deal early if the spot fixes at or above an upper barrier $B_{2}$.
{The payout }under an arbitrage-free model
with domestic risk-neutral measure $\mathbb{Q}^{d}$ , for a premium
expressed in domestic units (USD) and a notional $N_{EUR}$ in foreign
units (EUR), for a maturity $T$ is
\[
N_{\text{EUR}}\mathbb{E}^{\mathbb{Q}^{d}}\left[D_{T}^{d}\sum_{i=0}^{n_{f}}\mathbf{1}_{\tau>t_{i}}\left(K-S_{t_{i}}\right)\right]\,,
\]
where $\tau$ is the (early) redemption date that is either triggered
by the accrual breaching the target $H_{target}$ or the spot $S$
breaching the upper barrier $B_{2}$. In this test, we pick monthly
fixings and 
\[
T=5.0,\qquad n_{f}=12\times5,\qquad K=1.09\times S_{0},\qquad B_{2}=1.45\times S_{0},\qquad H_{target}=6\times\left(K-S_{0}\right)\,.
\]
We display in Table \ref{tab:5Y-NT-TARN} the prices computed with
1 048 575 quasi-Monte Carlo paths (with Sobol sequences and Brownian bridge
construction from \cite{SobolBB1997}) and 365 time steps per year for
the foreign no-touch and TARN contracts under the LV, LSV
and LSV--2CIR++ models. The running maximum is sampled with the
Brownian bridge technique, as described in Chapter 6 of \cite{Glasserman2004}.

\begin{table}[h]
	\begin{centering}
		\begin{tabular}{ccc}
			\toprule 
			\textbf{Model} & \textbf{NPV - No-touch} & \textbf{NPV - TARN}\tabularnewline
			\midrule
			\textbf{LV} & $72.42\%$ & $29.16\%$\tabularnewline
			\textbf{LSV} & $74.59\%$ & $14.69\%$\tabularnewline
			\textbf{LSV--2CIR}\scalebox{.9}{\raisebox{.5pt}{++}} & $73.60\%$ & $17.50\%$\tabularnewline
			\bottomrule
		\end{tabular}
		\par\end{centering}
	
	\centering{}\caption{Monte Carlo NPVs, in \% of $N_\text{EUR}$ and $N_\text{USD}$ respectively, for the 5Y foreign no-touch
		and target accrual redemption note\label{tab:5Y-NT-TARN}}
\end{table}

We infer from the data in Table \ref{tab:5Y-NT-TARN} a relative difference
of $1.35\%$ in price for the foreign no-touch and $-16.07\%$ for
the TARN between the LSV and LSV--2{CIR\scalebox{.9}{\raisebox{.5pt}{++}}}
models, which confirms that adding stochastic rates has a significant
impact even for a 5Y contract. 

Interestingly, in this setup, the price of the foreign no-touch option under the LSV--2{CIR\scalebox{.9}{\raisebox{.5pt}{++}}
	lies in between the LV and LSV models.
	It is common to observe a market price that lies between the LV model price and a Heston-type LSV model price, when the Heston parameters
	have been calibrated to the same vanilla options. Practitioners therefore introduce a mixing
	factor to control the amount of stochastic volatility of the LSV model and to manually match
	the no-touch options quotes (see
	\cite{Clark2010} for details). The introduction of stochastic rates
	seems to achieve a similar behaviour, at least in this example.

\section{Conclusion\label{sec:Conclusion}}

In this paper, we have provided a new and numerically effective method to
calibrate a 4-factor LSV model to vanilla options. In our
numerical tests with market data, we managed
to achieve an approximate $625$-fold speed-up for the calibration using control variates, as compared to the plain particle
method. We have shown that a high accuracy can be obtained with as few
as $\text{4\,000}$ particles (with a maximum error in absolute volatility
of $0.03\%$), and we were able to get a good fit with only $\text{800}$
particles (with a maximum error in absolute volatility of $0.05\%$). 

Using the calibrated leverage function from this paper, we showed that the addition of stochastic rates has a significant impact
on structured products, even more so when barrier features and coupon
detachments are combined for longer-dated contracts. Stochastic rates
become necessary in the modelling if one wishes to price hybrid products
where the rates appear explicitly (for instance, a spread option on
the FX performance and the Libor rate). One could use a second factor
in the CIR\scalebox{.9}{\raisebox{.5pt}{++}} processes to improve
the fit to caps and use the method presented in this article
to calibrate the leverage function.

{\small{}\bibliographystyle{siam}
\bibliography{GeneralBiblio}
}{\small \par}

\appendix



\section{Proof of Proposition \ref{prop:fwd_eq}
and Theorem \ref{thm:alphaFormula}}
\label{app:proof}

We first state a necessary auxiliary result, which is an adaptation of Tanaka's formula in Chapter 4 of \cite{Rogers&Williams2000}, where the integrand $D$ in the local time integral is $1$.
\begin{proposition}
\label{prop:DDeltaD2x_Dl}
On a filtered probability space ({\Large{}$\chi$},
$\mathcal{F},\left\{ \mathcal{F}_{t}\right\} _{t\geq0},\mathbb{Q}^{d})$,
let $D$ and $X$ be two $\mathcal{F}_{t}$-adapted continuous semi-martingales, with $D$ positive and integrable, $\left(l_{t}^{a}\right)_{t\geq0}$ the local time of $X$ at
level a and, for all $n>0$, 
\begin{equation}
\delta_{n}^{a}\left(x\right)=\begin{cases}
0, & \left|x-a\right|>\frac{1}{n},\\
\frac{n}{2}, & \left|x-a\right|\leq\frac{1}{n}\,.
\end{cases}\label{eq:2delta_n^K}
\end{equation}
Then, for any $T>0$ and $0 < t \leq T$,  $\int_{0}^{t}D_{s}\delta_{n}^{a}\left(X_{s}\right)\,d\langle X\rangle_{s}$
converges almost surely, and uniformly in time, to $\int_{0}^{t}D_{s}\,dl_s^{a}$.
\end{proposition}
\begin{proof}

\ifarxivvar

We closely follow Section 45 of \cite{Rogers&Williams2000}, but use a more concrete expression of the regularisation 
function. The {local time at level $a$
is defined as a continuous adapted increasing process such that
\begin{equation}
\left|X_{t}-a\right|-\left|X_{0}-a\right|=\int_{0}^{t} \sgn\left(X_{s}-a\right)\,dX_{s}+l_{t}^{a}\label{eq:localTimeDefByTanaka}\,,
\end{equation}
 with $\sgn\left(x\right)=-1$ for $x\leq0$ and $\sgn\left(x\right)=1$
for $x>0$. We define a sequence of functions $\left(f_{n}\right)_{n\geq0}$ for all $x\in\mathbb{R}$, as in Chapter 4 of \cite{Oksendal2014}
\[
f_{n}\left(x\right)=\begin{cases}
\left|x-a\right|, & \left|x-a\right|>\frac{1}{n},\\
\frac{1}{2}\left(\frac{1}{n}+n\left(x-a\right)^{2}\right), & \left|x-a\right|\leq\frac{1}{n}\,.
\end{cases}
\]
Hence, for all $n>\text{0}$, 
$\frac{1}{2}f_{n}^{''}=\delta_{n}^{a}$ a.e.} We recall
from the proof of Tanaka's formula in \cite{Rogers&Williams2000} 
 that $\int_{0}^{t}\delta_{n}^{a}\left(X_{s}\right)\,d\left\langle X\right\rangle _{s}$
converges almost surely to $l_{t}^{a}$ (uniformly in $t$)%
{. }%
{}%
{}%
We note that the sequence $f_{n}$ converges
uniformly to $x\rightarrow\left|x-a\right|$ and $f_{n}^{'}$ converges
point-wise to $\sgn\left(x-a\right)$. By the Itô-Doeblin
formula we can write
\begin{eqnarray}
D_{t}f_{n}\left(X_{t}\right)-D_{0}f_{n}\left(X_{0}\right) & = & \int_{0}^{t}f_{n}\left(X_{s}\right)\,dD_{s}+\int_{0}^{t}D_{s}f_{n}^{'}\left(X_{s}\right)\,dX_{s}\label{eq:itoSmoothedTanaka}\\
 & + & \frac{1}{2}\int_{0}^{t}D_{s}f_{n}^{''}\left(X_{s}\right)\,d\left\langle X\right\rangle _{s}+\int_{0}^{t}\,d\left\langle D,f_{n}\left(X\right)\right\rangle _{s}.\nonumber
\end{eqnarray}
We denote 
\[
C_{t}^{n}=\frac{1}{2}\int_{0}^{t}D_{s}f_{n}^{''}\left(X_{t}\right)\,d\left\langle X\right\rangle _{s}
\]
and, since 
$f_{n}^{''}\left(x\right)=0$
for any $x$ such that $\left|x-a\right|\geq\frac{1}{n}$, we have
\[
\int_{0}^{t}\mathbf{1}_{\left|X_{s}-a\right|>\frac{1}{n}}\,dC_{s}^{n}=0\,.
\]
{Also, from the definition of $f_n$, for all $x\in\mathbb{R}$,
\[
\sgn\left(x-a\right)-f_{n}^{'}\left(x\right)=\begin{cases}
0, & \left|x-a\right|\geq\frac{1}{n},\\
\sgn\left(x-a\right)-n\left(x-a\right), & \left|x-a\right|<\frac{1}{n},
\end{cases}
\]
 and then for any given $x$, and $n>0$, both $D_{t}\left|\sgn\left(x-a\right)-f_{n}^{'}\left(x\right)\right|$
and $D_{t}\left|\left|x-a\right|-f_{n}\left(x\right)\right|$ are smaller than $D_{t}$ which is integrable. }%
 Let $X_{t}=X_{0}+M_{t}+A_{t}$ be the canonical
decomposition of $X$ and $D_{t}=D_{0}+N_{t}+R_{t}$ the canonical
decomposition of $D$. 
Localisation allows us to reduce the problem
to the case where $M$ and $N$ are bounded and $A$ and $R$ are
of bounded variation. Then,

\[
\left\Vert \int_{0}^{T}D_{s}\left(\sgn\left(X_{s}-a\right)-f_{n}^{'}\left(X_{s}\right)\right)\,dM_{s}\right\Vert _{2}^{2}=\mathbb{E}\left[\int_{0}^{T}\left(D_{s}\left(\sgn\left(X_{s}-a\right)-f_{n}^{'}\left(X_{s}\right)\right)\right)^{2}\,d\left\langle M\right\rangle _{s}\right]\,,
\]
for which the right-hand-side goes to zero when $n$ goes to infinity. By Doob's $L^{2}$
martingale inequality, we can write 
\begin{eqnarray*}
\left\Vert \sup_{t\in\left[0,T\right]}\left|\int_{0}^{t}D_{s}\left(\sgn\left(X_{s}-a\right)-f_{n}^{'}\left(X_{s}\right)\right)\,dM_{s}\right|\right\Vert _{2}&\leq&2\left\Vert \int_{0}^{T}D_{s}\left(\sgn\left(X_{s}-a\right)-f_{n}^{'}\left(X_{s}\right)\right)\,dM_{s}\right\Vert _{2}\,,
\end{eqnarray*}
to conclude that
\begin{eqnarray}
\sup_{t\in\left[0,T\right]}\left|\int_{0}^{t}D_{s}\left(\sgn\left(X_{s}-a\right)-f_{n}^{'}\left(X_{s}\right)\right)\,dM_{s}\right|\rightarrow0\, \label{eq:supcvgL2}
\end{eqnarray}
in $L^{2}$ and in probability. We may then assume that \eqref{eq:supcvgL2} also holds almost surely (since we could work with a subsequence for which the statement is true instead).
{Similarly, we have
\[
\sup_{t\in\left[0,T\right]}\left|\int_{0}^{t}\left(\left|X_{t}-a\right|-f_{n}\left(X_{s}\right)\right)\,dN_{s}\right|\rightarrow0 \;\;a.s.
\]
Also,
\begin{eqnarray*}
\left|\int_{0}^{t}D_{s}\left(\sgn\left(X_{s}-a\right)-f_{n}^{'}\left(X_{s}\right)\right)\,dA_{s}\right| & \leq & \int_{0}^{t}D_{s}\left|\left(\sgn\left(X_{s}-a\right)-f_{n}^{'}\left(X_{s}\right)\right)\right|\left|\,dA_{s}\right|\\
 & \leq & \int_{0}^{t}D_{s}\left|\,dA_{s}\right|\,,
\end{eqnarray*}
and
\[
\sup_{t\in\left[0,T\right]}\left|\int_{0}^{t}D_{s}\left(\sgn\left(X_{s}-a\right)-f_{n}^{'}\left(X_{s}\right)\right)\,dA_{s}\right|\leq\int_{0}^{T}D_{s}\left|\,dA_{s}\right|\,.
\]
The monotone-convergence theorem allows us to conclude that
\[
\sup_{t\in\left[0,T\right]}\left|\int_{0}^{t}D_{s}\left(\sgn\left(X_{s}-a\right)-f_{n}^{'}\left(X_{s}\right)\right)\,dA_{s}\right|\rightarrow0\,
\]
in $L^{1}$, in probability and almost surely (on passing to a sub-sequence on $n$ if necessary). Hence,
\begin{equation}
\int_{0}^{t}D_{s}f_{n}^{'}\left(X_{s}\right)\,dA_{s}\rightarrow\int_{0}^{t}D_{s}\sgn\left(X_{s}\right)\,dA_{s}\label{eq:fv_monotone}
\end{equation}
 almost surely and uniformly in time}%
{. Similarly,
\begin{eqnarray*}
\left|\int_{0}^{t}\left(f_{n}\left(X_{s}\right)-\left|X_{s}-a\right|\right)\,dR_{s}\right| & \leq & \int_{0}^{t}\,\left|dR_{s}\right|\,,
\end{eqnarray*}
 and we get 
\begin{eqnarray*}
\int_{0}^{t}f_{n}\left(X_{s}\right)\,dR_{s} & \rightarrow & \int_{0}^{t}\left|X_{s}-a\right|\,dR_{s}
\end{eqnarray*}
almost surely and uniformly in time}%
{. Additionally, we can write
\[
\int_{0}^{t}\,d\left\langle D,f_{n}\left(X\right)\right\rangle _{s}=\int_{0}^{t}f_{n}^{'}\left(X_{s}\right)\,d\left\langle N,M\right\rangle _{s}\,.
\]
From the Kunita-Watanabe inequality,

\begin{eqnarray*}
\int_{0}^{t}\,\left|\left(\sgn\left(X_{s}-a\right)-f_{n}^{'}\left(X_{s}\right)\right)\right|\,d\left\langle N,M\right\rangle _{s} & \leq & \int_{0}^{t}\left|\left(\sgn\left(X_{s}-a\right)-f_{n}^{'}\left(X_{s}\right)\right)\right|\,\left|d\left\langle N,M\right\rangle _{s}\right|\\
 & \hspace{-1 cm} \leq & \hspace{-0.5 cm} \sqrt{\int_{0}^{t}\left|\left(\sgn\left(X_{s}-a\right)-f_{n}^{'}\left(X_{s}\right)\right)\right|^{2}\,d\left\langle M\right\rangle _{s}}\sqrt{\int_{0}^{t}\,d\left\langle N\right\rangle _{s}}.
\end{eqnarray*}
Since $\left\langle N\right\rangle _{s}$ and $\left\langle M\right\rangle _{s}$
are increasing processes of finite variation, we proceed
as in (\ref{eq:fv_monotone}) and conclude
\[
\int_{0}^{t}f_{n}^{'}\left(X_{s}\right)\,d\left\langle N,M\right\rangle _{s}\rightarrow\int_{0}^{t}\sgn\left(X_{s}-a\right)\,d\left\langle N,M\right\rangle _{s}\,a.s.
\]
Hence, from (\ref{eq:itoSmoothedTanaka}), $C^{n}$ converges
to a limit $\zeta$ almost surely (uniformly in time $t$). }%
{Applying integration
by parts to the Tanaka formula (\ref{eq:localTimeDefByTanaka}), we can write
\begin{eqnarray*}
D_{t}\left|X_{t}-a\right|-D_{0}\left|X_{0}-a\right| & = & \int_{0}^{t}\left|X_{s}-a\right|\,dD_{s}+\int_{0}^{t}D_{s}\sgn\left(X_{s}-a\right)\,dX_{s}\\
 & + & \int_{0}^{t}D_{s}\,dl_s^{a}+\int_{0}^{t}\,d\left\langle D,\left|X-a\right|\right\rangle _{s} \,,
\end{eqnarray*}
where Tanaka's formula also allows us to write 

\[
\int_{0}^{t}\,d\left\langle D,\left|X-a\right|\right\rangle _{s}=\int_{0}^{t}\sgn\left(X_{s}-a\right)\,d\left\langle N,M\right\rangle _{s}\,,
\]
and conclude that 
\[
\zeta_{t}=\int_{0}^{t}D_{s}\,dl_s^{a}\,.
\]
}

\else

Follows as in Section 45 of \cite{Rogers&Williams2000} with a few additional steps. A complete proof is given in Appendix \ref{app:proof}
of \cite{cmr18}.

\fi

\end{proof}

We now state and prove two lemmas necessary for the derivation of Theorem \ref{thm:alphaFormula}.
Lemma \ref{prop:EDdl} provides a link between the local time of a process and
its density function.


\begin{lemma}
\label{prop:EDdl}
Given a filtered probability space ({\Large{}$\chi$}, $\mathcal{F},\left\{ \mathcal{F}_{t}\right\} _{t\geq0},\mathbb{Q}^{d})$,
let $W$ be a standard Brownian motion and $\mu$, $Y$ two $\mathcal{F}_t$-adapted processes with $Y$ continuous,
with finite second moment and
\[
\int_{0}^{t} \left( \left| \mu_{u} \right| + Y_{u}^{2} \right) \,d{u} < \infty \,.
\]
Consider a continuous Itô process $X$ given by
\begin{eqnarray*}
X_{t} & = & X_{0}+\int_{0}^{t}\mu_{u}\,du+\int_{0}^{t}Y_{u}\,dW_{u}
\end{eqnarray*}
whose marginal density function $\phi(\cdot,t)$ and $\mathbb{E}^{\mathbb{Q}^{d}}\left[D_{t}Y_{t}^{2}\,|\,X_{t}=\cdot \right]$ are
assumed to be continuous.
{Further denote by $\left(l_{t}^{a}\right)_{t\geq0}$
the local time} of $X$ at level $a$. Then, for
any continuous, integrable and positive $\mathcal{F}_t$-adapted semi-martingale $D$ and any $a \in \mathbb{R}$,
\[
\mathbb{E}^{\mathbb{Q}^{d}}\left[\int_{0}^{t}D_{u}\,dl_{u}^{a}\right]=\int_{0}^{t} \mathbb{E}^{\mathbb{Q}^{d}}\left[D_{u}Y{}_{u}^{2}\,|\,X_{u}=a\right]\phi\left(a,u\right)\,du\,.
\]
\end{lemma}
\begin{proof} 
From Proposition \ref{prop:DDeltaD2x_Dl}, we know that the following
holds almost surely:
\[
\int_{0}^{t}D_{u}\,dl_{u}^{a}=\lim_{n\rightarrow\infty}\int_{0}^{t}D_{u}\delta_{n}^{a}\left(X_{u}\right)Y_{u}^{2}\,du\,.
\]
This implies convergence in distribution, so that we can write
\[
\mathbb{E}^{\mathbb{Q}^{d}}\left[\int_{0}^{t}D_{u}\,dl_{u}^{a}\right]=\lim_{n\rightarrow\infty}\mathbb{E}^{\mathbb{Q}^{d}}\left[\int_{0}^{t}D_{u}\delta_{n}^{a}\left(X_{u}\right)Y_{u}^{2}\,du\right]\,,
\]
and, by the stochastic Fubini theorem, we get 
\begin{eqnarray*}
\mathbb{E}^{\mathbb{Q}^{d}}\left[\int_{0}^{t}D_{u}\,dl_{u}^{a}\right] & = & \lim_{n\rightarrow\infty}\int_{0}^{t}\mathbb{E}^{\mathbb{Q}^{d}}\left[D_{u}\delta_{n}^{a}\left(X_{u}\right)\,Y_{u}^{2}\right]\,du\\
 & = & \lim_{n\rightarrow\infty}\int_{0}^{t}\mathbb{E}^{\mathbb{Q}^{d}}\left[\delta_{n}^{a}\left(X_{u}\right)\mathbb{E}^{\mathbb{Q}^{d}}\left[D_{u}Y_{u}^{2}\,|\,X_{u}\right]\right]\,du\,.
\end{eqnarray*}
We denote $\gamma\left(x,u\right)=\mathbb{E}^{\mathbb{Q}^{d}}\left[D_{u}Y_{u}^{2}\,|\,X_{u}=x\right]$,
 such that
\begin{eqnarray*}
\mathbb{E}^{\mathbb{Q}^{d}}\left[\int_{0}^{t}D_{u}\,dl_{u}^{a}\right]
&=&\lim_{n\rightarrow\infty}\int_{0}^{t}\int_{0}^{\infty}\delta_{n}^{a}\left(x\right)\gamma\left(x,u\right)\phi\left(x,u\right)\,dxdu \\
&=& \lim_{n\rightarrow\infty}\int_{0}^{\infty}\delta_{n}^{a}\left(x\right)\left(\int_{0}^{t}\gamma\left(x,u\right)\phi\left(x,u\right)\,du\right)\,dx, 
\end{eqnarray*}
where we have used Fubini's Theorem in the second line. 
By the continuity assumptions on $\phi$ and $\gamma$, we deduce that
\begin{eqnarray*}
\mathbb{E}^{\mathbb{Q}^{d}}\left[\int_{0}^{t}D_{u}\,dl_{u}^{a}\right] & = & \int_{0}^{t}\left(\mathbb{E}^{\mathbb{Q}^{d}}\left[D_{u}Y_{u}^2\,|\,X_{u}=a\right]\phi\left(a,u\right)\right)\,du\,.
\end{eqnarray*}

\end{proof}


\begin{lemma}
\label{prop:TrueMartingale}
Given the set-up of Theorem \ref{thm:alphaFormula},

\[
M_t = \int_{0}^{t}\mathbf{1}_{S_{u}\geq K}D_{u}^{d}\alpha\left(S_{u},u\right)S_{u}\sqrt{V_{u}}\,dW_{u}
\]
is a true martingale up to $T^*$ given by (\ref{explosion_time}).
\end{lemma}
\begin{proof}
Since $\alpha$
and $\mathbf{1}_{S_{u}\geq K}$ are bounded, the process \[M_t = \int_{0}^{t}\mathbf{1}_{S_{u}\geq K}D_{u}^{d}\alpha\left(S_{u},u\right)S_{u}\sqrt{V_{u}}\,dW_{u}, \,\]
is a true martingale if 
\[
\mathbb{E}^{\mathbb{Q}^{d}}\left[\int_{0}^{t}\left(D_{u}^{d}S_{u}\right)^{2}V_{u}\,du\right]<\infty\,.
\]
On the one hand, since $t<T^{*}$, from Proposition 3.13 in \cite{Cozma2015},
we can find $\omega>2$ such that 
\[
\sup_{u\in\left[0,t\right]}\mathbb{E}^{\mathbb{Q}^{d}}\left[\left(D_{u}^{d}S_{u}\right)^{\omega}\right]<\infty\,.
\]
On the other hand, from Theorem 3.1 in \cite{Hurd2008}, 
\[
\sup_{u\in\left[0,t\right]}\mathbb{E}^{\mathbb{Q}^{d}}\left[V_{u}^{\frac{\omega}{\omega-2}}\right]<\infty\,.
\]
Using Hölder's inequality with the pair $\left(\frac{\omega}{2},\frac{\omega}{\omega-2}\right)$,
\[
\mathbb{E}^{\mathbb{Q}^{d}}\left[\left(D_{u}^{d}S_{u}\right)^{2}V_{u}\right]\leq\mathbb{E}^{\mathbb{Q}^{d}}\left[\left(D_{u}^{d}S_{u}\right)^{\omega}\right]^{\frac{2}{\omega}}\mathbb{E}^{\mathbb{Q}^{d}}\left[V_{u}^{\frac{\omega}{\omega-2}}\right]^{\frac{\omega-2}{\omega}}<\infty.
\]
Finally, using the Fubini theorem, $\mathbb{E}^{\mathbb{Q}^{d}}\left[\int_{0}^{t}\left(D_{u}^{d}S_{u}\right)^{2}V_{u}\,du\right]<\infty$
and hence 
$M$
is a true martingale of zero expectation.
\end{proof}

Combining Lemmas \ref{prop:EDdl} and \ref{prop:TrueMartingale},
we can derive Proposition \ref{prop:fwd_eq}.

\begin{proof}[Proof of Proposition \ref{prop:fwd_eq}]
Let $K\in\mathbb{R}^{+}$, $0<t<T^{*}$ and $H_{t}=\left(S_{t}-K\right)^{+}$.
{The Trotter-Meyer theorem \cite{Rogers&Williams2000}
gives}
\begin{eqnarray*}
\left(S_{t}-K\right)^{+}-\left(S_{0}-K\right)^{+} 
 & = & \int_{0}^{t}\mathbf{1}_{S_{u}\geq K}\,dS_{u}+\frac{1}{2}l_{t}^{K}\,,
\end{eqnarray*}
which we can write in differential form as
\begin{eqnarray*}
dH_{t} 
 & = & \mathbf{1}_{S_{t}\geq K}S_{t}\left(r_{t}^{d}-r_{t}^{f}\right)\, dt+\frac{1}{2}dl_{t}^{K}+\mathbf{1}_{S_{t}\geq K}\alpha\left(S_{t},t\right)S_{t}\sqrt{V_{t}}\, dW_{t}\,.
\end{eqnarray*}
Also,
\begin{eqnarray}
\nonumber
d\left(D_{t}^{d}H_{t}\right) 
  &=& D_{t}^{d}\left[-r_{t}^{d}H_{t}+\mathbf{1}_{S_{t}\geq K}S_{t}\left(r_{t}^{d}-r_{t}^{f}\right)\right]\,dt+\frac{1}{2}D_{t}^{d}\,dl_{t}^{K} \\
  && \hspace{3 cm} 
  + \, \mathbf{1}_{S_{t}\geq K}D_{t}^{d}\alpha\left(S_{t},t\right)S_{t}\sqrt{V_{t}}\,dW_{t}\,.\;\; \label{eq:mainIto}
\end{eqnarray}
Hence, by applying Lemma \ref{prop:EDdl} with $X_{t}=S_{t}$,
$D_{t}=D_{t}^{d}$ and $Y_{t}=\alpha\left(S_{t},t\right)S_{t}\sqrt{V_{t}}$,
we can write 
\begin{eqnarray}
\mathbb{E}^{\mathbb{Q}^{d}}\left[\int_{0}^{t}D_{u}^{d}\,dl_{u}^{K}\right] & = & \int_{0}^{t}\left(\alpha^2\left(K,u\right) K^{2}\mathbb{E}^{\mathbb{Q}^{d}}\left[D_{u}^{d}V_{u}\,|\,S_{u}=K\right]\phi\left(K,u\right)\right)\,du\,,\label{eq:EC_Intermediary}
\end{eqnarray}
where $\phi$ is the marginal density function of $S$ at time $t${.
Furthermore, one can define $\bar{\phi}_{n}$ as} 
\[
\bar{\phi}_{n}\left(K,u\right)=\mathbb{E}^{\mathbb{Q}^{d}}\left[D_{u}^{d}\delta_{n}^{K}\left(S_{u}\right)\right]=\int_{0}^{\infty}\delta_{n}^{K}\left(x\right)\mathbb{E}^{\mathbb{Q}^{d}}\left[D_{u}^{d}\,|\,S_{u}=x\right]\phi\left(x,u\right)\,dx\,,
\]
with $\delta_{n}^{K}$ defined as in {(\ref{eq:2delta_n^K}})
by 
\[
\delta_{n}^{K}\left(x\right)=\begin{cases}
0, & \left|x-K\right|>\frac{1}{n},\\
\frac{n}{2}, & \left|x-K\right|\leq\frac{1}{n}\,,
\end{cases}
\]
and by a similar reasoning to that of Lemma \ref{prop:EDdl} we get

\[
\lim_{n\rightarrow\infty}\bar{\phi}_{n}=\mathbb{E}^{\mathbb{Q}^{d}}\left[D_{u}^{d}\,|\,S_{u}=K\right]\phi\left(K,u\right)\,.
\]
Since 
\[
\frac{\partial^{2}C\left(K,u\right)}{\partial K^{2}}=\lim_{n\rightarrow\infty}\,\mathbb{E}^{\mathbb{Q}^{d}}\left[D_{u}^{d}\delta_{n}^{K}\left(S_{u}\right)\right]\,,
\]
we write 
\begin{equation}
\frac{\partial^{2}C\left(K,u\right)}{\partial K^{2}}=\mathbb{E}^{\mathbb{Q}^{d}}\left[D_{u}^{d}\,|\,S_{u}=K\right]\phi\left(K,u\right).\label{eq:D2CDK2Exp}
\end{equation}
Combining (\ref{eq:EC_Intermediary}) and (\ref{eq:D2CDK2Exp}) allows
to write
\[
\mathbb{E}^{\mathbb{Q}^{d}}\left[\int_{0}^{t}D_{u}^{d}\,dl_{u}^{K}\right]=\int_{0}^{t}\alpha^2\left(K,u\right) K^{2}\frac{\mathbb{E}^{\mathbb{Q}^{d}}\left[D_{u}^{d}V_{u}\,|\,S_{u}=K\right]}{\mathbb{E}^{\mathbb{Q}^{d}}\left[D_{u}^{d}\,|\,S_{u}=K\right]}\frac{\partial^{2}C\left(K,u\right)}{\partial K^{2}}\,du\,.
\]
Hence, integrating (\ref{eq:mainIto}), 

\begin{eqnarray}
C\left(K,t\right)=\mathbb{E}^{\mathbb{Q}^{d}}\left[D_{t}^{d}H_{t}\right] & = & \int_{0}^{t}\left(-\mathbb{E}^{\mathbb{Q}^{d}}\left[D_{u}^{d}r_{u}^{d}\left(S_{u}-K\right)^{+}\right]+\mathbb{E}^{\mathbb{Q}^{d}}\left[D_{u}^{d}\mathbf{1}_{S_{u}\geq K}S_{u}\left(r_{u}^{d}-r_{u}^{f}\right)\right]\right)\,du\nonumber \\
 & & + \;\; \frac{1}{2}\int_{0}^{t}\alpha^2\left(K,u\right) K^{2}\frac{\mathbb{E}^{\mathbb{Q}^{d}}\left[D_{u}^{d}V_{u}\,|\,S_{u}=K\right]}{\mathbb{E}^{\mathbb{Q}^{d}}\left[D_{u}^{d}\,|\,S_{u}=K\right]}\frac{\partial^{2}C\left(K,u\right)}{\partial K^{2}}\,du\label{eq:callInDeferentialForm}\\
 & & + \;\; \mathbb{E}^{\mathbb{Q}^{d}}\left[\int_{0}^{t}\mathbf{1}_{S_{u}\geq K}D_{u}^{d}\alpha\left(S_{u},u\right)S_{u}\sqrt{V_{u}}\,dW_{u}\right]\,.\nonumber 
\end{eqnarray}
Furthermore, on a fixed time interval $\left[0,T^{*}\right]$, $D_{t}^{d}$
is uniformly bounded by $\exp({T^{*}\max_{u\in[0,T^{*}]}\left|h^{d}\left(u\right)\right|})$.
Then, from Lemma \ref{prop:TrueMartingale} we know that $\int_{0}^{t}\mathbf{1}_{S_{u}\geq K}D_{u}^{d}\alpha\left(S_{u},u\right)S_{u}\sqrt{V_{u}}\,dW_{u}$
is a true martingale of zero expectation.

We write (\ref{eq:callInDeferentialForm}) at time $T$, differentiate
with respect to $T$ and, upon noticing that $\mathbf{1}_{S_{T}\geq K}S_{T}=\left(S_{T}-K\right)^{+}+\mathbf{1}_{S_{T}\geq K}K$,
we get (\ref{eq:stoVolDupirePDE}).
\end{proof}

We are now ready to give the proof of Theorem \ref{thm:alphaFormula}.

\begin{proof}
[Proof of Theorem \ref{thm:alphaFormula}] 
First, we want to ensure that (\ref{eq:alphaFormula}) is a necessary condition for 
\begin{equation}
C\left(K,T\right)= 
C_{LV}\left(K,T\right)\,.\label{eq:CModel_CLV}
\end{equation}
Hence, by subtracting the Dupire PDE (\ref{eq:dupirePDE}) from (\ref{eq:stoVolDupirePDE}),
we obtain
\begin{eqnarray*}
\frac{1}{2}K^{2}\left(\alpha^2\left(K,T\right) \frac{\mathbb{E}^{\mathbb{Q}^{d}}\left[D_{T}^{d}V_{T}\,|\,S_{T}=K\right]}{\mathbb{E}^{\mathbb{Q}^{d}}\left[D_{T}^{d}\,|\,S_{T}=K\right]}-\sigma_{LV}^2\left(K,T\right)\right)\frac{\partial^{2}C_{LV}}{\partial K^{2}} &=& 
\mathbb{E}^{\mathbb{Q}^{d}}\left[D_{T}^{d}r_{T}^{f}\left(S_{T}-K\right)^{+}\right] \\
&&\hspace{-8 cm} 
-\bar{r}^{f}\left(T\right)C_{LV}-\mathbb{E}^{\mathbb{Q}^{d}}\left[D_{T}^{d}\mathbf{1}_{S_{T}\geq K}K\left(r_{T}^{d}-r_{T}^{f}\right)\right]-K\left(\bar{r}^{d}\left(T\right)-\bar{r}^{f}\left(T\right)\right)\frac{\partial C_{LV}}{\partial K}\,,
\end{eqnarray*}
so 
\begin{eqnarray*}
\alpha^2\left(K,T\right) 
 &=&
 {\frac{\mathbb{E}^{\mathbb{Q}^{d}}\left[D_{T}^{d}\,|\,S_{T}=K\right]}{\mathbb{E}^{\mathbb{Q}^{d}}\left[D_{T}^{d}V_{T}\,|\,S_{T}=K\right]}\!\left(\sigma_{LV}^2\!\left(K,T\right)+\bar{q}(K,T)\!\right)},
\end{eqnarray*}
where
\begin{eqnarray*}
&&\hspace{-0.8 cm} \bar{q}(K,T) =
\frac{\mathbb{E}^{\mathbb{Q}^{d}}\left[\overline{Q}_{T}\right]}{\frac{1}{2}K^{2}\frac{\partial^{2}C_{LV}}{\partial K^{2}}}\,, \\
&&\hspace{-0.8 cm} \overline{Q}_{T} =
D_{T}^{d}r_{T}^{f}\left(S_{T}-K\right)^{+}-\bar{r}^{f}\left(T\right)C_{LV}-K\left(D_{T}^{d}\mathbf{1}_{S_{T}\geq K}\left(r_{T}^{d}-r_{T}^{f}\right)+\left(\bar{r}^{d}\left(T\right)-\bar{r}^{f}\left(T\right)\right)\frac{\partial C_{LV}}{\partial K}\right)\,.
\end{eqnarray*}
It remains to show that we can replace $\overline{Q}_{T}$ by $Q_T$ in $\bar{q}$ . 
First, $D_{T}^{d}\left(S_{T}-K\right)^{+}$ is weakly differentiable
with respect to $K$ with $\frac{\partial D_{T}^{d}\left(S_{T}-K\right)^{+}}{\partial K}=D_{T}^{d}\mathbf{1}_{S_{T}\geq K}$,
which is bounded by the integrable process $D_{T}^{d}$. We can interchange
differentiation and expectation to get $\mathbb{E}^{\mathbb{Q}^{d}}\left[D_{T}^{d}\mathbf{1}_{S_{T}\geq K}\right]=-\frac{\partial C}{\partial K}$.
Since the models agree, 
$\mathbb{E}^{\mathbb{Q}^{d}}\left[D_{T}^{d}\mathbf{1}_{S_{T}\geq K}\right]=-\frac{\partial C_{LV}}{\partial K}$
and $\mathbb{E}^{\mathbb{Q}^{d}}\left[D_{T}^{d}\left(S_{T}-K\right)^{+}\right]=C_{LV}$, and (\ref{eq:alphaFormula}) holds.

By re-tracing the steps in reverse order, one sees that (\ref{eq:alphaFormula}) is also a  sufficient condition for (\ref{eq:CModel_CLV})
provided the solution to \eqref{eq:stoVolDupirePDE} is unique.
\end{proof}


%

\section{Monte Carlo QE-scheme\label{appendix:Monte-Carlo-QE-scheme}}

The \emph{Quadratic-Exponential (QE)} scheme \cite{Andersen2008} that
is used to discretise the square-root process, employs moment-matching techniques and can significantly reduce the Monte
Carlo discretisation error. While the full truncation Euler scheme
and the QE scheme have shown to perform well in our tests, we experienced
a faster convergence in time for the QE scheme when the Feller condition
is broken. Hence, we choose the QE scheme for the variance process and
the full truncation Euler for both stochastic rates, as the computational
cost will be smaller. We briefly write a generalised QE scheme based
on the original scheme from \cite{Andersen2008} to incorporate a
leverage function and stochastic rates in the discretisation.

We follow our time interpolation rule for the calibration of $\alpha$
and interpolate forward-flat in time. We assume for simplicity that
each Monte Carlo time step belongs to the $\alpha$ time grid. We
can write
\[
V_{t+\Delta t}=V_{t}+\int_{t}^{t+\Delta t}\kappa\left(\theta-V_{u}\right)\,du+\xi\int_{t}^{t+\Delta t}\sqrt{V_{u}}\,dW_{u}^{V}\,,
\]
and hence
\[
\int_{t}^{t+\Delta t}\sqrt{V_{u}}\,dW_{u}^{V}=\frac{V_{t+\Delta t}-V_{t}-\int_{t}^{t+\Delta t}\kappa\left(\theta-V_{u}\right)\,du}{\xi}\,,
\]
and 
\[
d\ln S_{t}=\left(r_{t}^{d}-r_{t}^{f}-\frac{1}{2}\alpha^{2}\left(S_{t},t\right)V_{t}\right)\,dt+\alpha\left(S_{t},t\right)\rho\sqrt{V_{t}}\,dW_{t}^{V}+\alpha\left(S_{t},t\right)\sqrt{1-\rho^{2}}\sqrt{V_{t}}\,dW_{t}^{S}\,,
\]
where $W_{t}^{S}$ is a Brownian motion independent of $W_{t}^{V}$. Therefore,
\begin{eqnarray*}
\ln S_{t+\Delta t} & = & \ln S_{t}+\int_{t}^{t+\Delta t}\left(r_{u}^{d}-r_{u}^{f}\right)\,du-\frac{1}{2}\alpha^{2}\left(S_{t},t\right)\int_{t}^{t+\Delta t}V_{u}\,du\\
 & + & \frac{\alpha\left(S_{t},t\right)\rho\left(V_{t+\Delta t}-V_{t}-\kappa\theta\Delta t+\kappa\int_{t}^{t+\Delta t}V_{u}\,du\right)}{\xi}\\
 & + & \alpha\left(S_{t},t\right)\sqrt{1-\rho^{2}}\int_{t}^{t+\Delta t}\sqrt{V_{u}}\,dW_{u}^{S}\,.
\end{eqnarray*}
We approximate $\int_{t}^{t+\Delta t}V_{u}\,du$ by $\left(\frac{V_{t+\Delta t}+V_{t}}{2}\right)\Delta t$
, and note that conditional on $V_{t}$ and $\int_{t}^{t+\Delta t}V_{u}\,du$, since $W_{u}^{V}$ and $W_{u}^{S}$ are independent,
the Itô integral $\int_{t}^{t+\Delta t}\sqrt{V_{u}}\,dW_{u}^{S}$
is normally distributed with mean zero and variance $\int_{t}^{t+\Delta t}V_{u}\,du$. We write the full
scheme below

\begin{equation}
\begin{cases}
g_{t+\Delta t}^{d} & =g_{t}^{d}+\kappa_{d}\left(\theta_{d}-\left(g_{t}^{d}\right)^{+}\right)\Delta t+\xi_{d}\sqrt{\left(g_{t}^{d}\right)^{+}}\sqrt{\Delta t}\,Y_{d}\\
g_{t+\Delta t}^{f} & =g_{t}^{f}+  \left( \kappa_{f}\left(\theta_{f}-\left(g_{t}^{f}\right)^{+}\right) - \left(\rho_{Sf}\xi_{f}\sqrt{\left(g_{t}^{f}\right)^{+}}\alpha\left(S_{t},t\right)\sqrt{V_{t}}\right) \right)\Delta_{t}\\
& + \xi_{f}\sqrt{\left(g_{t}^{f}\right)^{+}}\sqrt{\Delta t}\,Y_{f}\\
\ln S_{t+\Delta t} & =\ln S_{t}+\left(\frac{\left(r_{t+\Delta t}^{d}-r_{t+\Delta t}^{f}\right)+\left(r_{t}^{d}-r_{t}^{f}\right)}{2}-\frac{1}{4}\alpha^{2}\left(S_{t},t\right)\left(V_{t+\Delta t}+V_{t}\right)\right)\Delta t\\
 & +\frac{\alpha\left(S_{t},t\right)\rho\left(V_{t+\Delta t}-V_{t}+\kappa\left(\frac{V_{t+\Delta t}+V_{t}}{2}-\theta\right)\Delta t\right)}{\xi}\\
 & +\alpha\left(S_{t},t\right)\sqrt{1-\rho^{2}}\sqrt{\frac{V_{t+\Delta t}+V_{t}}{2}}\sqrt{\Delta t}\,Z\\
V_{t+\Delta t}: & \begin{cases}
\text{if }\psi\leq\psi_{c} & :\,V_{t+\Delta t}=a\left(b+Z_{v}\right)^{2}\\
\text{else } & \begin{cases}
\text{if}\quad U\leq p & :\,V_{t+\Delta t}=0\\
\text{else} & :\,V_{t+\Delta t}=\ln\left(\frac{1-p}{1-U}\right)\frac{m}{1-p}\,,
\end{cases}
\end{cases}
\end{cases}\label{eq:QE-Scheme}
\end{equation}
with
\[
\begin{cases}
m & =\theta+\left(V_{t}-\theta\right)e^{-\kappa\Delta t}\,,\\
\gamma^{2} & =\frac{V_{t}\xi^{2}e^{-\kappa\Delta t}}{\kappa}\left(1-e^{-\kappa\Delta t}\right)+\frac{\theta\xi^{2}}{2\kappa}\left(1-e^{-\kappa\Delta t}\right)^{2}\,,\\
\psi & =\frac{\gamma^{2}}{m^{2}},\quad p=\frac{\psi-1}{\psi+1},\quad \beta=\frac{1-p}{m}\,,\\
b^{2} & =\frac{2}{\psi}-1+\sqrt{\frac{2}{\psi}}\sqrt{\frac{2}{\psi}-1},\quad a=\frac{m}{1+b^{2}}\,,\\
\psi_{c} & =1.5\,,
\end{cases}
\]
Let the Cholesky decomposition of the correlation matrix\footnote{If the correlation matrix is not positive definite, one can rely on spectral decomposition instead; see 2.3 in \cite{Glasserman2004} for details.}

\[
\left[\begin{array}{cccc}
	1 & \rho & 0 & 0\\
	\rho & 1 & \rho_{Sd} & \rho_{Sf}\\
	0 & \rho_{Sd} & 1 & \rho_{df}\\
	0 & \rho_{Sf} & \rho_{df} & 1
\end{array}\right]\,,\]
be $\mathbf{LL}^{T}$. $Y_{d}$
and $Y_{f}$
are defined as 
\[\left[\begin{array}{c}
	Y_{v}\\
	Y_{s}\\
	Y_{d}\\
	Y_{f}
\end{array}\right]=\mathbf{L}\left[\begin{array}{c}
	Z_{v}\\
	Z\\
	Z_{d}\\
	Z_{f}
\end{array}\right] \,,\]
where $Z,\,Z_{v},\,Z_{d},\,Z_{f}$ are independent draws from a standard
normal distribution and $U$ is a draw from a uniform distribution.

\section{Finite element mesh construction\label{sub:Mesh-construction}}

In order to refine the mesh in the most relevant area,
we use an exponential mesh on the variance axis and a hyperbolic
mesh (see \cite{White2013}) in the spot direction. This makes the
mesh finer around $z=0$ and $x=S_{0}$. In order to build our mesh,
we first define the grids in spot $\left(x_{i}\right)_{i\in\ensuremath{\left\llbracket 0,N_{S}\right\rrbracket }}$
and variance $\left(z_{j}\right)_{j\in\ensuremath{\left\llbracket 0,N_{V}\right\rrbracket }}$
separately. Additionally, to solve the PDE numerically, we need to
truncate at the boundary and use $\bar{\Omega}=\left\{ \left(x,z\right)\in\left[0,S_{\max}\right]\times\left[0,V_{\max}\right]\right\} $
on a time interval $\left[0,T\right]$ . We choose 
\[
S_{\max}=S_{0}e^{5\alpha\left(S_{0},T\right)\sqrt{\left(v_{0}e^{-\kappa T}+\theta\left(1-e^{-\kappa T}\right)\right)T}}
\]
and recall that the stationary distribution of the CIR process is
a gamma distribution of density $\phi_{v}^{\infty}$ as defined in
(\ref{eq:hestonDensityStationary}). We compute $V_{\max}$ with the
inverse cumulative density function such that 
\[
\mathbb{P}\left(z>V_{\max}\right)=0.01\%\,.
\]
We write 
\begin{eqnarray*}
x_{i} & = & f_{h}\left(g_{h}\left(\bar{x}_{i}\right)\right)\,,\\
z_{j} & = & f_{e}\left(g_{e}\left(\bar{z}_{j}\right)\right)\,,
\end{eqnarray*}
with

\noindent
\begin{minipage}[c]{1\columnwidth}%
\begin{flushleft}
\begin{minipage}[c]{0.49\columnwidth}%
\begin{eqnarray*}
f_{h}\left(x\right) & = & S_{0}+b\,\text{sinh}\left(\nu x+d\right),\\
b & = & \eta\left(S_{\max}-S_{\min}\right),\\
d & = & \text{arcsinh}\left(\frac{S_{\min}-S_{0}}{b}\right),\\
\nu & = & \text{arcsinh}\left(\frac{S_{\max}-S_{0}}{b}\right)-d,\\
\bar{x}_{i} & = & \frac{i}{\left(N_{S}+1\right)},\\
\eta & = & 0.02\,,
\end{eqnarray*}

where $\eta$ is defined according to our numerical experiments and
$g_{h}$ is the quadratic polynomial that passes through the points
$\left(0,0\right)$, $\left(1,1\right)$, 
\[
\left(\frac{\left\lfloor f_{h}^{-1}\left(S_{0}\right)\left(N_{S}+1\right)+0.5\right\rfloor }{N_{S}+1},f_{h}^{-1}\left(S_{0}\right)\right).
\]
\end{minipage}~~~%
\begin{minipage}[c]{0.48\columnwidth}%
\begin{eqnarray*}
f_{e}\left(z\right) & = & c+c\,\exp\left(\lambda z\right),\\
c & = & \frac{V_{\max}}{e^{\lambda}-1}\,,\\
\lambda & = & \max\left(1,4-\frac{3\kappa\theta}{\xi^{2}}\right),\\
\bar{z}_{j} & = & \frac{j}{\left(N_{V}+1\right)}\,,\\
 & \,\\
 & \,
\end{eqnarray*}
\vspace{-0.2 cm}

where $\lambda$ is defined according to our numerical experiments and $g_{h}$ is the quadratic polynomial that passes through the points
$\left(0,0\right)$, $\left(1,1\right)$, 
\[
\left(\frac{\left\lfloor f_{e}^{-1}\left(v_{0}\right)\left(N_{V}+1\right)+0.5\right\rfloor }{N_{V}+1},f_{e}^{-1}\left(v_{0}\right)\right).
\]
\end{minipage}
\par\end{flushleft}%
\end{minipage}
\vspace{0.2 cm}

The latter intermediate step makes sure that both $S_{0}$ and $v_{0}$
are vertices of their respective grids. The construction of the finite
element triangular mesh can be achieved by creating a vertex at each
point $\left(x_{i},z_{j}\right)$ and defining two triangular cells
(upper left and lower right) in each rectangle.

\ifarxivvar

\section{Shifted CIR model and calibration}

\label{subsec:rate}\label{sub:CIR++-Model-Calibration}

The domestic and foreign short interest rates are modeled by the shifted
CIR (CIR\scalebox{.9}{\raisebox{.5pt}{++}}) process \cite{Brigo2006}.
On the one hand, this model preserves the analytical tractability
of the CIR model for bonds, caps and other basic interest
rate products. On the other hand, it is flexible enough to fit the
initial term structure of interest rates exactly. For $i\in\{d,f\}$,
the short rate dynamics under their respective spot measures, i.e.,
$\mathbb{Q}^{d}$ -- domestic and $\mathbb{Q}^{f}$ -- foreign, are
given by 
\begin{align}
\begin{cases}
\hspace{5.5pt}r_{t}^{i}=g_{t}^{i}+h^{i}(t),\\
dg_{t}^{i}\hspace{-0.5pt}=\kappa_{i}(\theta_{i}-g_{t}^{i})dt+\xi_{i}\sqrt{g_{t}^{i}}\,dB_{t}^{\gi},\hspace{0.75em}g_{0}^{i}>0,
\end{cases}\label{eq5.1}
\end{align}
where $B^{\gd}$ and $B^{\gf}$ are Brownian motions under $\mathbb{Q}^{d}$
and $\mathbb{Q}^{f}$, respectively. The mean-reversion parameters
$\kappa_{i}$, the long-term mean parameters $\theta_{i}$ and the
volatility parameters $\xi_{i}$ are the same as in (\ref{eq:Model Definition}).
The calibration of the short rate model (\ref{eq5.1}) follows the
same approach for both the domestic and the foreign interest rate.
For simplicity, we drop the subscripts and superscripts ``$d$''
and ``$f$'' in the remainder of the subsection and define the vector
of parameters $\beta_{1}=(g_{0},\kappa,\theta,\xi)$. According to
Brigo and Mercurio \cite{Brigo2001}, an exact fit to the initial
term structure of interest rates is equivalent to $h(t)=\varphi^{\text{CIR}}(t;\beta_{1})$
for all $t\in[0,T]$, where 
\begin{align}
\varphi^{\text{CIR}}(t;\beta_{1}) & =\bar{r}(0,t)-\bar{r}^{\text{CIR}}(0,t;\beta_{1}),\label{eq5.2}\\[3pt]
\bar{r}^{\text{CIR}}(0,t;\beta_{1}) & =\frac{2\kappa\theta(\exp\{t\nu\}-1)}{2\nu+(\kappa+\nu)(\exp\{t\nu\}-1)}+g_{0}\hspace{1pt}\frac{4\nu^{2}\exp\{t\nu\}}{[2\nu+(\kappa+\nu)(\exp\{t\nu\}-1)]^{2}},\nonumber 
\end{align}
$\nu=\sqrt{\kappa^{2}+2\xi^{2}}$ and $\bar{r}(0,t)$ is the market
instantaneous forward rate at time $0$ for a maturity $t$, i.e.,
\begin{equation}
\bar{r}(0,t)=\bar{r}(t)=-\hspace{1pt}\frac{\partial\ln P(0,t)}{\partial t}\hspace{1pt},\label{eq5.3}
\end{equation}
where $P(0,t)$ is the market zero coupon bond price at time $0$
for a maturity $t$. The value of the zero coupon bond is given by
\begin{equation}
P(0,t)=\frac{1}{1+\Delta(0,t)R(0,t)},\label{eq5.4}
\end{equation}
where $\Delta(0,t)$ is the year fraction from $0$ to time $t$ and
$R(0,t)$ is the current (simply-compounded) deposit rate with maturity
date $t$ which is quoted in the market%
. As an aside, note that the standard day count convention for USD
and EUR is Actual $360$. %

The detailed calibration procedure for both domestic and foreign rate
processes can be found in Appendix \ref{appendix:Shifted-CIR-model}.
The calibration results are displayed in Table \ref{table4.3}. 
\begin{table}[htb]
\centering{}\caption{The calibrated CIR parameters}
\label{table4.3} \begin{tabularx}{\textwidth}{@{}YYYYY@{}} 	\addlinespace[-5pt]   \toprule[.1em] 	 \textbf{CCY} & $g_{0}$ & $\kappa$ & $\theta$ & $\xi$ \\   \midrule 	\textbf{USD} & $0.0001$ & $0.0837$ & $0.5469$ & $0.0274$ \\[2pt] 	\textbf{EUR} & $0.0001$ & $0.0110$ & $1.1656$ & $0.0370$ \\ 	\bottomrule[.1em] 	\addlinespace[3pt] \end{tabularx}
\end{table}

\subsection{Shifted CIR model calibration\label{appendix:Shifted-CIR-model}}

In order to estimate the zero coupon curve (also known as the term
structure of interest rates or the yield curve), we assume that the
instantaneous forward rate is piecewise-flat. Consider the time nodes
$t_{0}\!=\!0,t_{1},\hdots,t_{n}$ and the set of estimated instantaneous
forward rates $f_{1},f_{2},\hdots,f_{n}$ from which the curve is
constructed, and define 
\begin{equation}
\bar{r}(t)=f_{i}\hspace{1em}\text{ if }\hspace{1em}t_{i-1}\leq t<t_{i},\hspace{1em}\text{ for }\hspace{1em}i=1,2,\hdots,n.\label{eq5.5}
\end{equation}
Using (\ref{eq5.3}) -- (\ref{eq5.5}) and solving the resulting linear
system of equations, we get 
\begin{equation}
f_{i}=\frac{1}{\Delta(t_{i-1},t_{i})}\ln\left(\frac{1+\Delta(0,t_{i})R(0,t_{i})}{1+\Delta(0,t_{i-1})R(0,t_{i-1})}\right)\hspace{1em}\text{ for }\hspace{1em}i=1,2,\hdots,n.\label{eq5.6}
\end{equation}
The continuously-compounded spot rate, i.e., the constant rate at
which the value of a pure discount bond must grow to yield one unit
of currency at maturity, is defined as 
\begin{equation}
R_{0}(0,t)=\frac{1}{\Delta(0,t)}\int_{0}^{t}{\bar{r}(s)\,ds}.\label{eq5.7}
\end{equation}
Using (\ref{eq5.5}) -- (\ref{eq5.7}), we deduce that 
\begin{equation}
R_{0}(0,t)=\frac{\Delta(t,t_{i})}{\Delta(0,t)\Delta(t_{i-1},t_{i})}\ln\Big(1+\Delta(0,t_{i-1})R(0,t_{i-1})\Big)+\frac{\Delta(t_{i-1},t)}{\Delta(0,t)\Delta(t_{i-1},t_{i})}\ln\Big(1+\Delta(0,t_{i})R(0,t_{i})\Big)
\end{equation}
whenever $t_{i-1}\leq t<t_{i}$. In Figure \ref{fig:1}, we plot the
USD and EUR zero coupon curves $t\mapsto R_{0}(0,t)$, $t>0$, estimated
from the quoted deposit rates from March 18, 2016, together with the
flat-forward instantaneous forward rates. 
\begin{figure}[htb]
\centering \begin{subfigure}{.5\textwidth} \centering \captionsetup{justification=centering,margin=2cm}
\includegraphics[width=0.97\linewidth,height=1.7in]{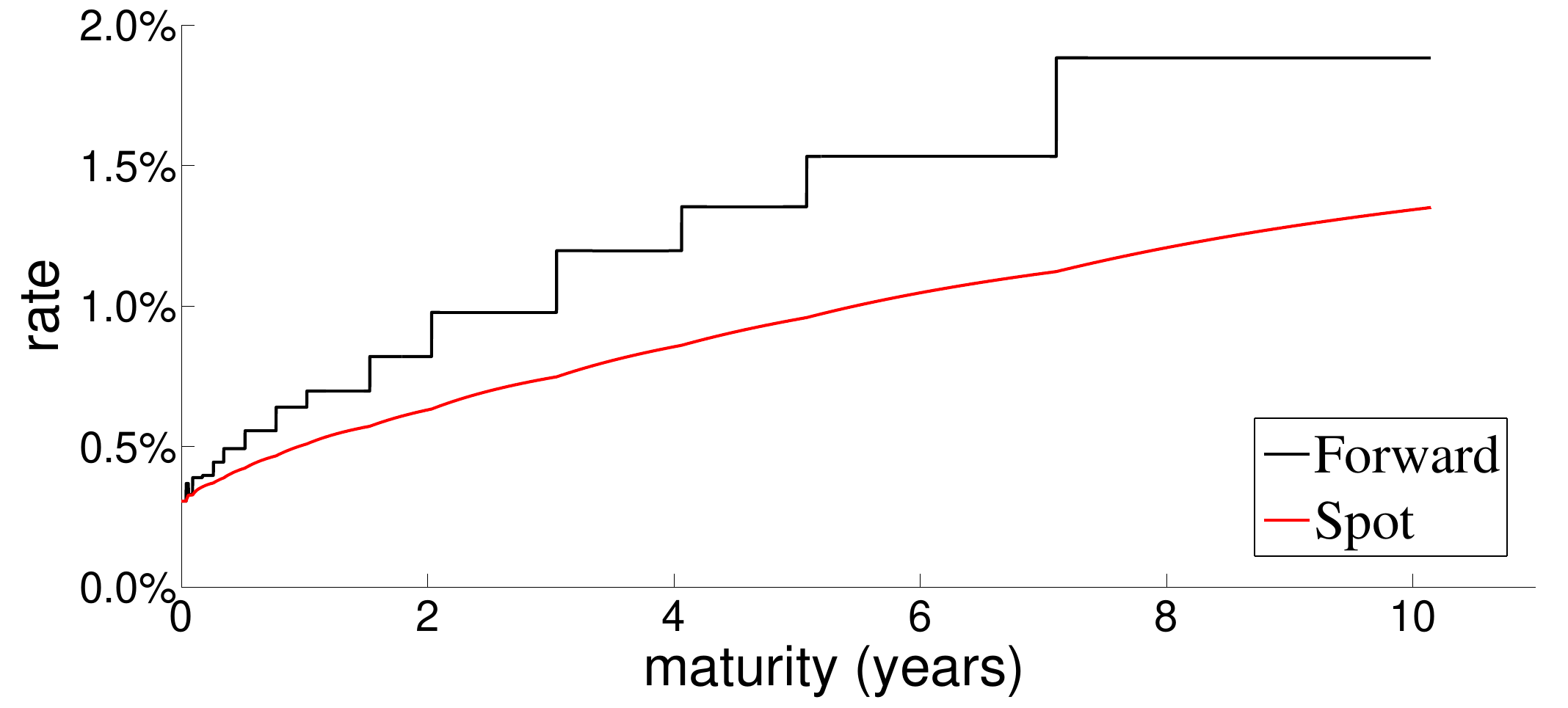}
\caption{USD Currency}
\label{fig:1a} \end{subfigure}\begin{subfigure}{.5\textwidth}
\centering \captionsetup{justification=centering,margin=2cm} \includegraphics[width=0.97\linewidth,height=1.7in]{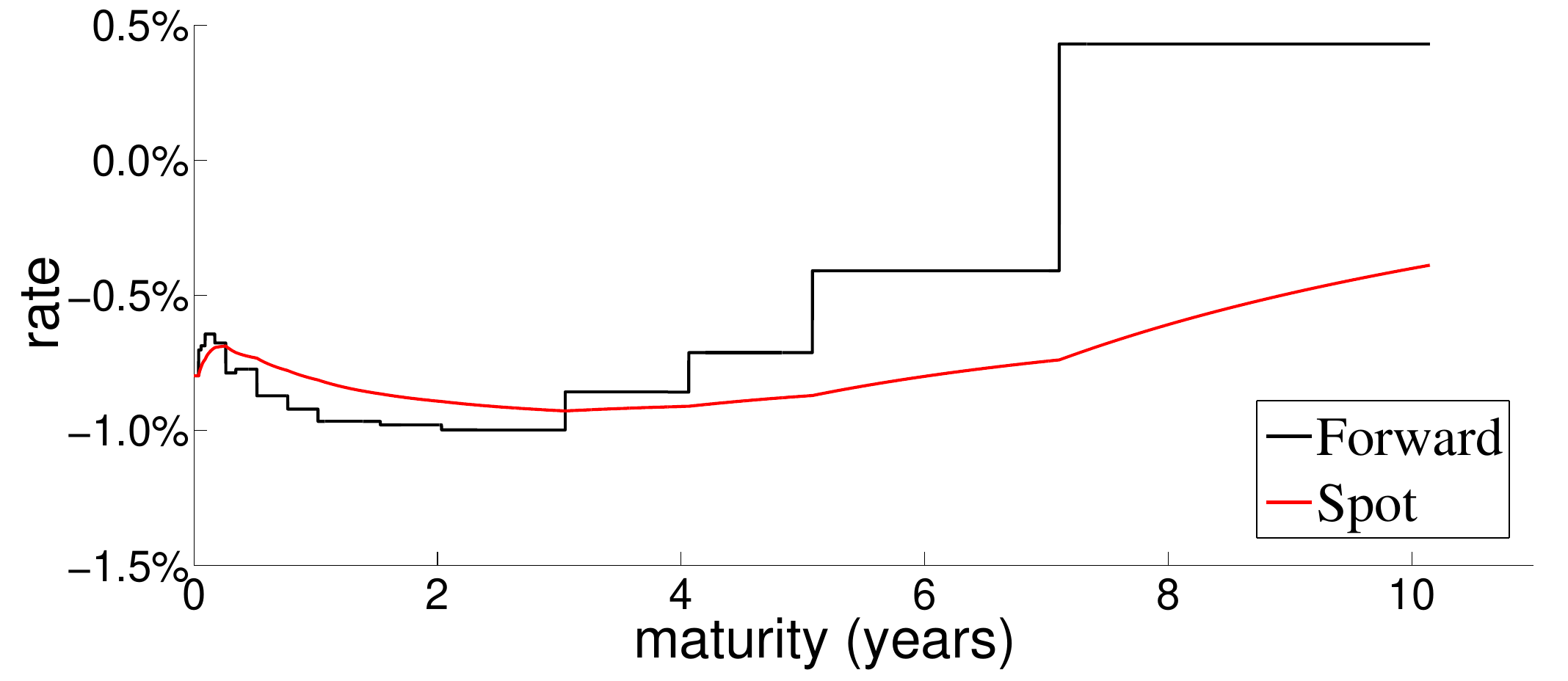}
\caption{EUR Currency}
\label{fig:1b} \end{subfigure} \\[1em] \caption{The instantaneous forward rates and the continuously-compounded spot
rates.}
\label{fig:1} 
\end{figure}

A choice of the shift function $h$ as in (\ref{eq5.2}) results in
an exact fit to the initial term structure of interest rates independent
of the value of the parameter vector $\beta_{1}$.

Next, we determine $\beta_{1}$ by calibrating the CIR\scalebox{.9}{\raisebox{.5pt}{++}}
model to the current term structure of volatilities, in particular,
by fitting at-the-money (ATM) cap volatilities.%
{} We consider caps with integer maturities ranging from $1$ to $10$
years for both currencies, with an additional $18$ month cap for
EUR. For USD, all caps have quarterly frequency, whereas for EUR the
$1$ year and $18$ month caps have quarterly frequency and the $2$
to $10$ year caps have semi-annual frequency. %
A cap is a set of spanning caplets with a common strike so the value
of the cap is simply the sum of the values of its caplets. It is market
standard to price caplets with the Black formula, in which case the
fair value of the cap at time $0$ with rate (strike) $K$, reset
times $T_{a},T_{a+1},\hdots,T_{b-1}$ and payment times $T_{a+1},\hdots,T_{b-1},T_{b}$
is: 
\begin{equation}
\text{Cap}_{\text{Black}}(K,\sigma_{a,b})=\sum_{i=a+1}^{b}P(0,T_{i})\Delta(T_{i-1},T_{i})\text{Black}\big(K,F(0,T_{i-1},T_{i}),\sigma_{a,b}\sqrt{T_{i-1}}\big),\label{eq5.9}
\end{equation}
where $F(0,T,S)$ is the simply-compounded forward rate at time $0$
for the expiry $T$ and maturity $S$ defined as 
\begin{equation}
F(0,T,S)=\frac{1}{\Delta(T,S)}\left(\frac{P(0,T)}{P(0,S)}-1\right)\label{eq5.10}
\end{equation}
and the Black volatility $\sigma_{a,b}$ corresponding to a strike
$K$ is retrieved from market quotes. Denoting by $\phi_{0}$ and
$\Phi_{0}$ the standard normal probability density function (PDF)
and cumulative distribution function (CDF), respectively, Black's
formula is: 
\begin{align}
\text{Black}(K,F,v) & =F\Phi_{0}(d_{1})-K\Phi_{0}(d_{2}),\label{eq5.11}\\[2pt]
d_{1,2} & =\frac{\ln(F/K)\pm v^{2}/2}{v}.\nonumber 
\end{align}
However, Black's formula cannot cope with negative forward rates $F$
or strikes $K$, in which case we switch to Bachelier's (normal) formula
in (\ref{eq5.9}): 
\begin{align}
\text{Normal}(K,F,v) & =(F-K)\Phi_{0}(d)+v\phi_{0}(d),\label{eq5.12}\\[2pt]
d & =\frac{F-K}{v}.\nonumber 
\end{align}
The data in Figure \ref{fig:1b} suggest that the instantaneous forward
rate for EUR takes negative values. Therefore, we use Black cap volatility
quotes for USD and Normal cap volatility quotes for EUR. The market
prices of at-the-money caps are computed by inserting the forward
swap rate 
\begin{equation}
S_{a,b}=\frac{P(0,T_{a})-P(0,T_{b})}{\sum_{i=a+1}^{b}\Delta(T_{i-1},T_{i})P(0,T_{i})}\label{eq5.13}
\end{equation}
as strike and the quoted cap volatility as $\sigma_{a,b}$ in (\ref{eq5.9}),
using either Black's or Bachelier's formula.

Fitting the CIR\scalebox{.9}{\raisebox{.5pt}{++}} model to cap
volatilities means finding the value of $\beta_{1}$ for which the
model cap prices, which are available in closed-form \cite{Brigo2006},
best match the market cap prices. The calibration is performed by
minimising the sum of the squared differences between model- and market-implied
cap volatilities: 
\begin{equation}
\min_{\beta_{1}\in\mathbb{R}_{\scalebox{.5}{\raisebox{.5pt}{+}}}^{4}}\hspace{1pt}\sum_{1\leq i\leq n}{\big[\sigma^{\text{CIR}}(T_{i};\beta_{1})-\sigma^{\text{M}}(T_{i})\big]^{2}},\label{eq5.14}
\end{equation}
where $\sigma^{\text{CIR}}$ and $\sigma^{\text{M}}$ stand for the
model- and the market-implied cap volatilities, respectively, and
$T_{1},\hdots,T_{n}$ are the cap maturities. Model-implied cap volatilities
are obtained by pricing market caps with the CIR\scalebox{.9}{\raisebox{.5pt}{++}}
model and then inverting the formula (\ref{eq5.9}) in order to retrieve
the implied volatility associated with each maturity. We choose to
calibrate the model to cap volatilities since they are of similar
magnitude, unlike cap prices which can differ by a few orders of magnitude.
The calibration results are displayed in Table \ref{table4.3}. 

On the one hand, (\ref{eq5.14}) is a highly nonlinear and non-convex
optimisation problem, and the objective function may have multiple
local minima. On the other hand, global optimisation algorithms require
a very high computation time and do not scale well with complexity,
as opposed to local optimisation methods. A fast calibration is important
in practice since option pricing models may need to be re-calibrated
several times within a short time span. Therefore, we used a nonlinear
least-squares solver, in particular the trust-region-reflective algorithm
\cite{Coleman1996}, for the calibration and a global optimisation
method, in particular a genetic algorithm \cite{Storn1997}, for verification
purposes only.

Figure \ref{fig:2} shows the fitting capability of the CIR\scalebox{.9}{\raisebox{.5pt}{++}}
model, and the implied cap volatility curve is compared to the market
curve for each currency. Taking into account that the model has only
$4$ parameters to fit between $10$ and $11$ data points, we conclude
that the CIR\scalebox{.9}{\raisebox{.5pt}{++}} model provides
a fairly reasonable fit to the term structure of cap volatilities
$T_{i}\mapsto\sigma^{\text{M}}(T_{i})$, $1\leq i\leq n$. 
\begin{figure}[htb]
\centering \begin{subfigure}{.5\textwidth} \centering \captionsetup{justification=centering,margin=2cm}
\includegraphics[width=0.97\linewidth,height=1.7in]{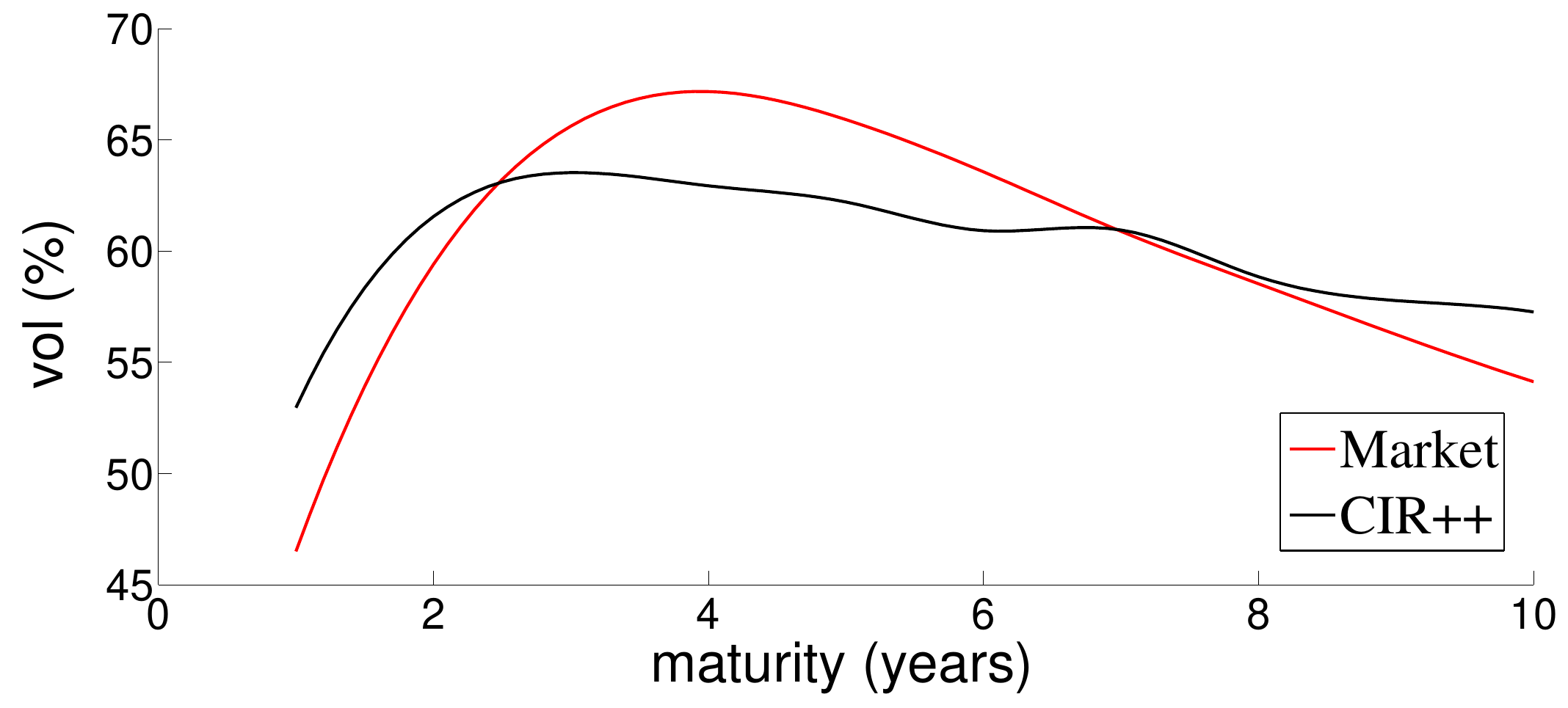}
\caption{USD Currency}
\label{fig:2a} \end{subfigure}\begin{subfigure}{.5\textwidth}
\centering \captionsetup{justification=centering,margin=2cm} \includegraphics[width=0.97\linewidth,height=1.7in]{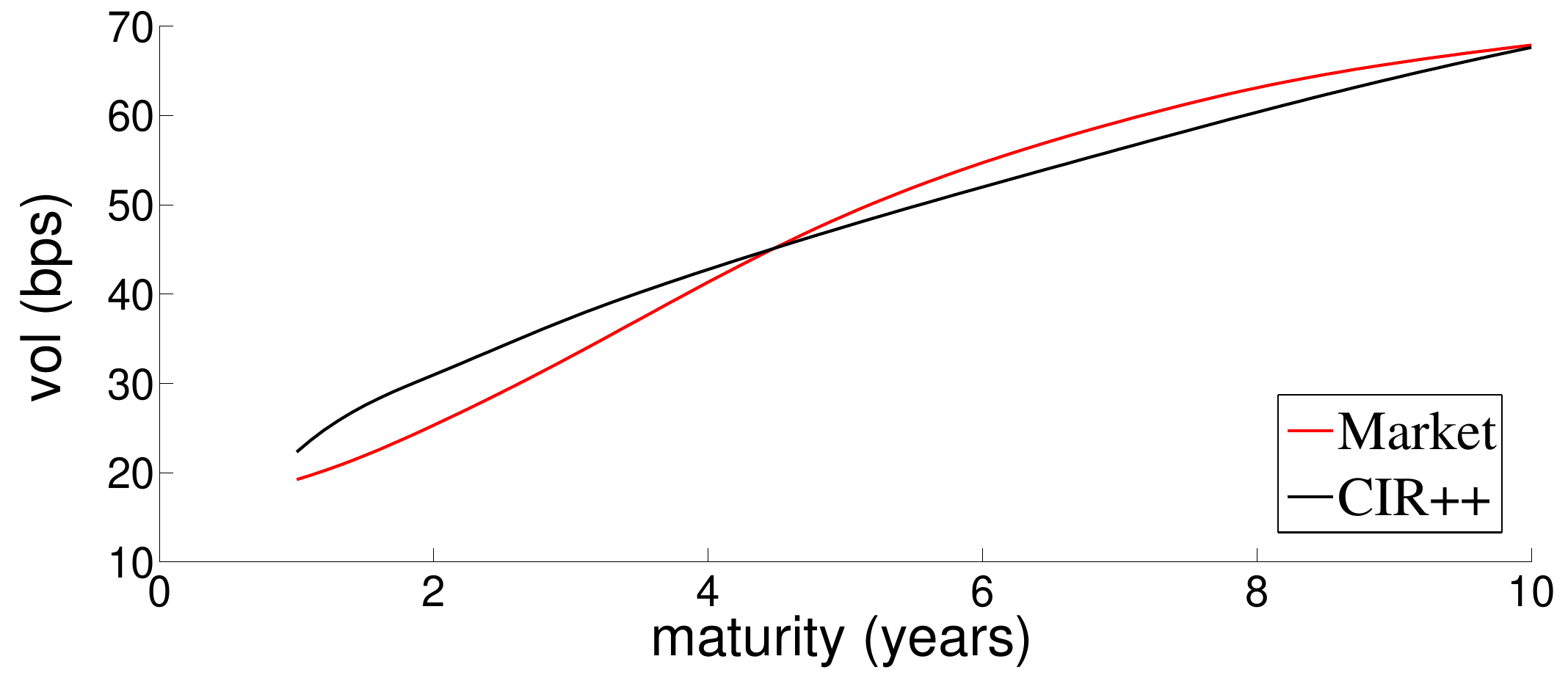}
\caption{EUR Currency}
\label{fig:2b} \end{subfigure} \\[1em] \caption{The market- and model-implied term structures of cap volatilities.}
\label{fig:2} 
\end{figure}


\section{Local volatility calibration algorithm\label{sub:Appendix-Local-volatility-calibration}}

\subsection{Calibration with Dupire PDE \label{sub:1-Factor-Local-Volatility-Calibration}}

The calibration routine for a pure local volatility model is run with
a standard algorithm forward in maturity. We recall that model (\ref{eq:modelDefinitionLocalVolatility})
is written as
\[
\frac{dS_{t}^{LV}}{S_{t}^{LV}}=\left(\bar{r}^{d}\left(t\right)-\bar{r}^{f}\left(t\right)\right)\,dt+\sigma_{LV}\left(S_{t}^{LV},t\right)\,dW_{t}\,,
\]
 and we want to find the function $\sigma_{LV}$ for which the call
prices under the local volatility model match the quoted market prices
exactly. This is crucial as both $\sigma_{LV}$ and
$\frac{\partial^{2}C_{LV}}{\partial K^{2}}$ appear in the leverage
function formula (\ref{eq:alphaFormula}).{
The forward Dupire PDE (\ref{eq:dupirePDE}),
\begin{equation}
\begin{cases}
\frac{\partial C_{LV}}{\partial T}+\left(\bar{r}^{d}\left(T\right)-\bar{r}^{f}\left(T\right)\right)K\frac{\partial C_{LV}}{\partial K}+\bar{r}^{f}\left(T\right)C_{LV}-\frac{1}{2}K^{2}\frac{\partial^{2}C_{LV}}{\partial K^{2}}\sigma_{LV}^{2}(K,T) =0\,,\\
C_{LV}\left(K,0\right)=\left(S_{0}-K\right)^{+},\quad C_{LV}\left(0,T\right)=S_{0},\quad C_{LV}\left(S_{\max},T\right)=0\, ,
\end{cases}\label{eq:Dupire-Forward-PDE}
\end{equation}
provides an efficient way to calibrate and, eventually, regularise the problem. Denote by $\Phi$ the map from the local volatility function to the model implied
volatility function $\Sigma_{Model}$. 
Furthermore, the PDE solution for a guess $\sigma_{S}$ of the local volatility gives call prices for the whole set of strikes and maturities.
Inverting the Black formula allows to retrieve the model implied volatilities
$\Sigma$. Hence, as proposed in \cite{Tur2014}, we can use the forward
Dupire PDE (\ref{eq:Dupire-Forward-PDE}) combined with an efficient
implied volatility inverter \cite{Jackel2015} as the mapping function
$\Phi$. A very useful property of this PDE is that it can be solved
forward in maturity. Let a set of maturities quoted on the market
be $\left(T_{1},...,T_{N_{Mat}}\right)$ and a set of $M_{i}$ strikes
for a given maturity $T_{i}$ be $\left(K_{T_{i},1},...,K_{T_{i},M_{i}}\right)$.
It is possible to solve the PDE on $\left[0,T_{1}\right]$, then on
$\left[T_{1},T_{2}\right]$ and so forth. The full calibration algorithm
is presented for completeness in Appendix \ref{sub:Appendix-Local-volatility-calibration}.}

\subsection{Computation of the target volatility surface}

For the calibration routine, we will compute the solution of the PDE
(\ref{eq:Dupire-Forward-PDE}) by a finite difference method. The spot grid is defined on $\left[0,S_{\max}\right]$, where $S_{\max}=S_{0}e^{\frac{6}{2}\sigma_{ATM}\left(\frac{T_{\max}}{2}\right)\sqrt{\frac{T_{\max}}{2}}}$. In order to speed up the calibration
routine, we prefer not to use too many spot steps and time steps (150
steps in space and 20 time steps per year). Hence, the scheme will
not have converged to the solution of the PDE at this point. In order
to tackle this problem and still benefit from a good speed-up, we
will compute a ``target volatility surface'': instead of calibrating
the market volatility surface, we will calibrate a volatility surface
that takes into account the discretisation error of the numerical PDE
solution. Industry practitioners like Murex use this approach \cite{MXLV2007}.
The algorithm to build the target surface is explained below.

\begin{algorithm}[H]
\begin{algor}
\item [{for}] ( $i=1\,;\,i\leq N_{Mat}\,;\,i++$) 
\item [{for}] ( $j=1\,;\,j\leq M_{i}\,;\,j++$)
\item [{{*}}] \textbf{define} $\sigma_{Market}=$ $\Sigma_{Market}\left(K_{i,j},T_{i}\right)$
from the market volatility surface
\item [{{*}}] \textbf{solve} the PDE (\ref{eq:Dupire-Forward-PDE}) with
constant local vol $\sigma_{LV}=\sigma_{Market}$
\item [{{*}}] \textbf{get} $C\left(K_{i,j},T_{i}\right)$ from the numerical
solution
\item [{{*}}] \textbf{get} {$\Sigma_{Target}\left(K_{i,j},T_{i}\right)$
by inverting the price with the Black-Scholes formula}
\item [{endfor}]~
\item [{endfor}]~
\end{algor}
\caption{Computation of the target volatility surface}
\end{algorithm}

\subsection{Calibration by fixed-point algorithm and forward induction}
\label{subsec:fixed-point}

The local volatility
function is defined on a grid of points interpolated with cubic splines
in spot and backward flat in time. In the FX case, where there are
5 quoted strikes per maturity (10 maturities), the local volatility
is defined on a grid of $50$ points. Each one of the points $\sigma_{LV}^{i,j}=\sigma_{LV}\left(K_{T_{i},j},T_{i}\right)$,
with $i\in\llbracket1,10\rrbracket$ and $j\in\llbracket1,5\rrbracket$,
can be seen as a parameter of the local volatility surface. For a
given maturity $T_{i}$, the local volatility is defined on the interval
$\left[K_{T_{i},1},K_{T_{i},5}\right]$ and is extrapolated flat outside
those bounds.

In order to define a first guess for the calibration routine, we use
a smoothed bi-variate cubic spline following the algorithm in \cite{Dierckx1981}
to interpolate in strike and maturity the call prices on the market.
This allows us to use the Dupire formula to define a first guess for
the first maturity $T=T_{1}$. After the calibration of the first
maturity pillar $T_{1}$, the first guess for the next pillar is the
current maturity local volatility. This approach has shown the best
stability and speed in our tests. 

As we now have a way to get the model implied volatility from the
local volatility (with $\Phi$), one can follow a Picard fixed-point
algorithm as proposed in \cite{Rehai2006,Tur2014} that
we describe below.

\begin{algorithm}[H]
\begin{algor}
\item [{for}] ( $i=1\,;\,i\leq N_{Mat}\,;\,i++$) 
\item [{while}] it < maxIter
\item [{{*}}] \textbf{solve} PDE (\ref{eq:Dupire-Forward-PDE}) on $\left[T_{i-1},T_{i}\right]$
\item [{{*}}] \textbf{compute} model implied vol $\Sigma_{Model}$ for
maturity $T_{i}$ from the computed call prices
\item [{{*}}] \textbf{compute} $error=\sum_{m=1}^{M_{i}}\left(\Sigma_{Model}\left(K_{T_{i},m},T_{i}\right)-\Sigma_{Target}\left(K_{T_{i},m},T_{i}\right)\right)^{2}$
\item [{if}] error < tol
\item [{{*}}] \textbf{endwhile}
\item [{else}]~
\item [{for}] ( $j=1\,;\,j\leq M_{i}\,;\,j++$)
\item [{{*}}] \textbf{update} local volatility guess 
\[
\sigma_{LV}\left(K_{T_{i},j},T_{i}\right)=\sigma_{LV}\left(K_{T_{i},j},T_{i}\right)\frac{\Sigma_{Target}\left(K_{T_{i},j},T_{i}\right)}{\Sigma_{Model}\left(K_{T_{i},j},T_{i}\right)}
\]

\item [{endfor}]~
\item [{endif}]~
\item [{{*}}] it++
\item [{endwhile}]~
\item [{endfor}]~
\end{algor}
\caption{%
Fixed-point forward induction%
}
\end{algorithm}

\begin{remark}
It is stated but not proved in \cite{Rehai2006} that the map $\Phi$ is contracting and
so is $f\left(\left\{ \sigma_{LV}\right\} \right)\rightarrow\left\{ \sigma_{LV}\right\} *\frac{\left\{ \Sigma_{Target}\right\} }{\Phi\left(\left\{ \sigma_{LV}\right\} \right)}.$
Assuming this to be true, $f$ admits a unique fixed point that is the limit of the sequence
of local volatility guesses $\left(\sigma_{LV}^{n}\right)_{n\in\mathbb{N}}$
defined as $\left\{ \sigma_{LV}^{n+1}\right\} =f\left(\left\{ \sigma_{LV}^{n}\right\} \right)$.
In practice, convergence is achieved particularly fast (between 10 and 20 iterations).

The calibrated local volatility is shown in Figure \ref{fig:Local-Volatility-function-1},
where we plot it on a time scale to $T_{\max}$ for a better illustration
of its shape.
\end{remark}

%
We perform the calibration with 800 space steps and 100 time steps per year for the forward Dupire PDE, where we use the finite element method with
quadratic basis functions. We then price quoted vanilla
contracts with the backward Feynman--Kac PDE under the calibrated local volatility model. We get a maximum error in implied volatility
smaller than 0.01\% (i.e., for a market volatility of 20\%, the calibrated
volatility could be $20.00\pm0.01$\% in the worst case scenario).

Additionally, we plot the discounted marginal density of the spot
extracted from the market. As mentioned before, this quantity is $\frac{\partial^{2}C_{LV}}{\partial K^{2}}$
and can be computed from the PDE solution immediately and accurately.
As we will use the density in the calibration formula in Theorem \ref{eq:alphaFormula},
we want it to be smooth and accurate. Figure {\ref{fig:spotMarginalDensityPDE-1}
}shows that quantity.

%

\noindent
\begin{minipage}[c]{1.\columnwidth}%
\begin{minipage}[c]{0.46\columnwidth}%
\begin{centering}
\includegraphics[scale=0.19]{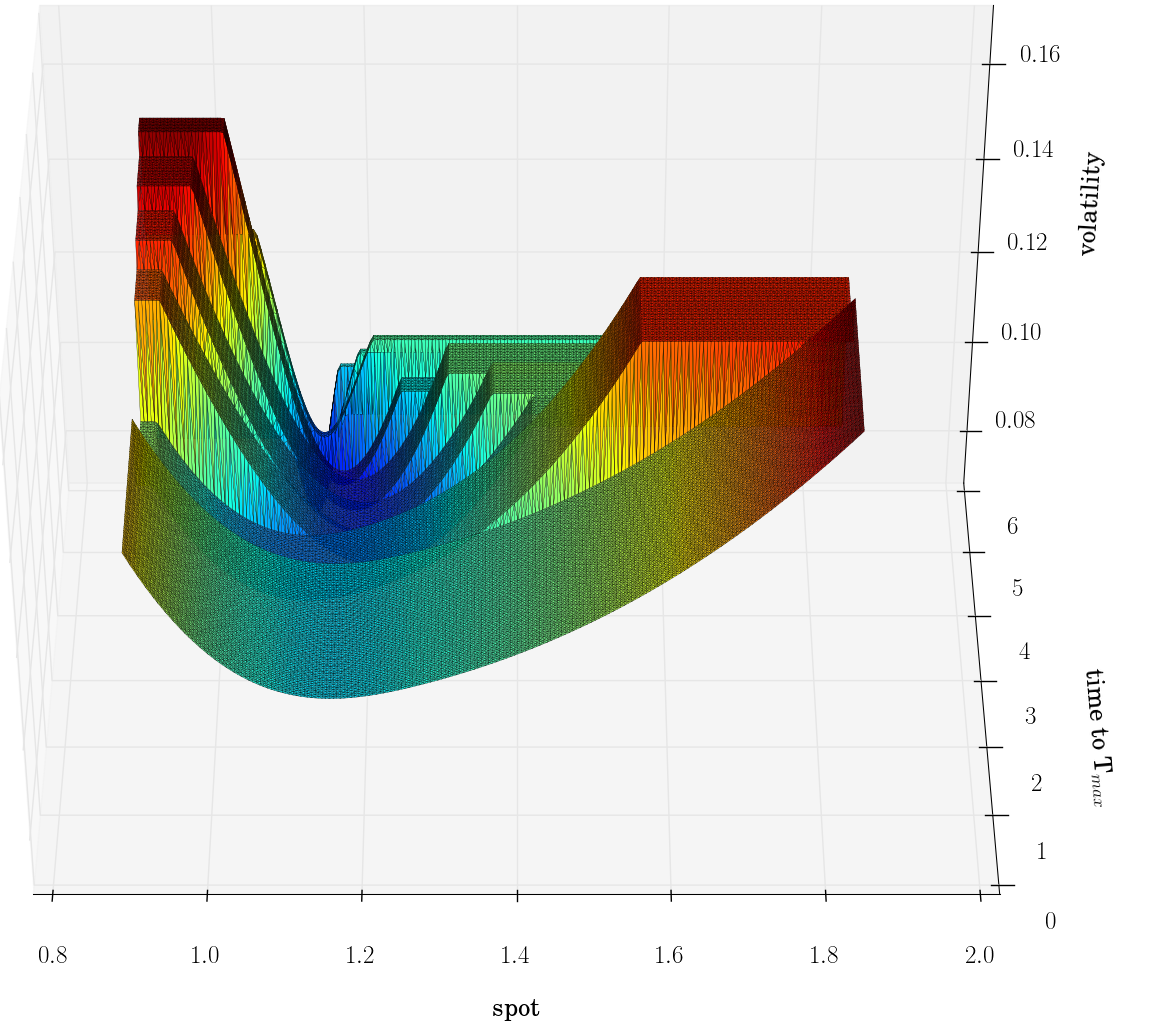}

\captionof{figure}{EURUSD Local volatility function calibrated by forward PDE and fixed-point algorithm}

\label{fig:Local-Volatility-function-1}%
\end{centering}
\end{minipage}~~~~%
\begin{minipage}[c]{0.46\columnwidth}%
\begin{centering}
\includegraphics[scale=0.17]{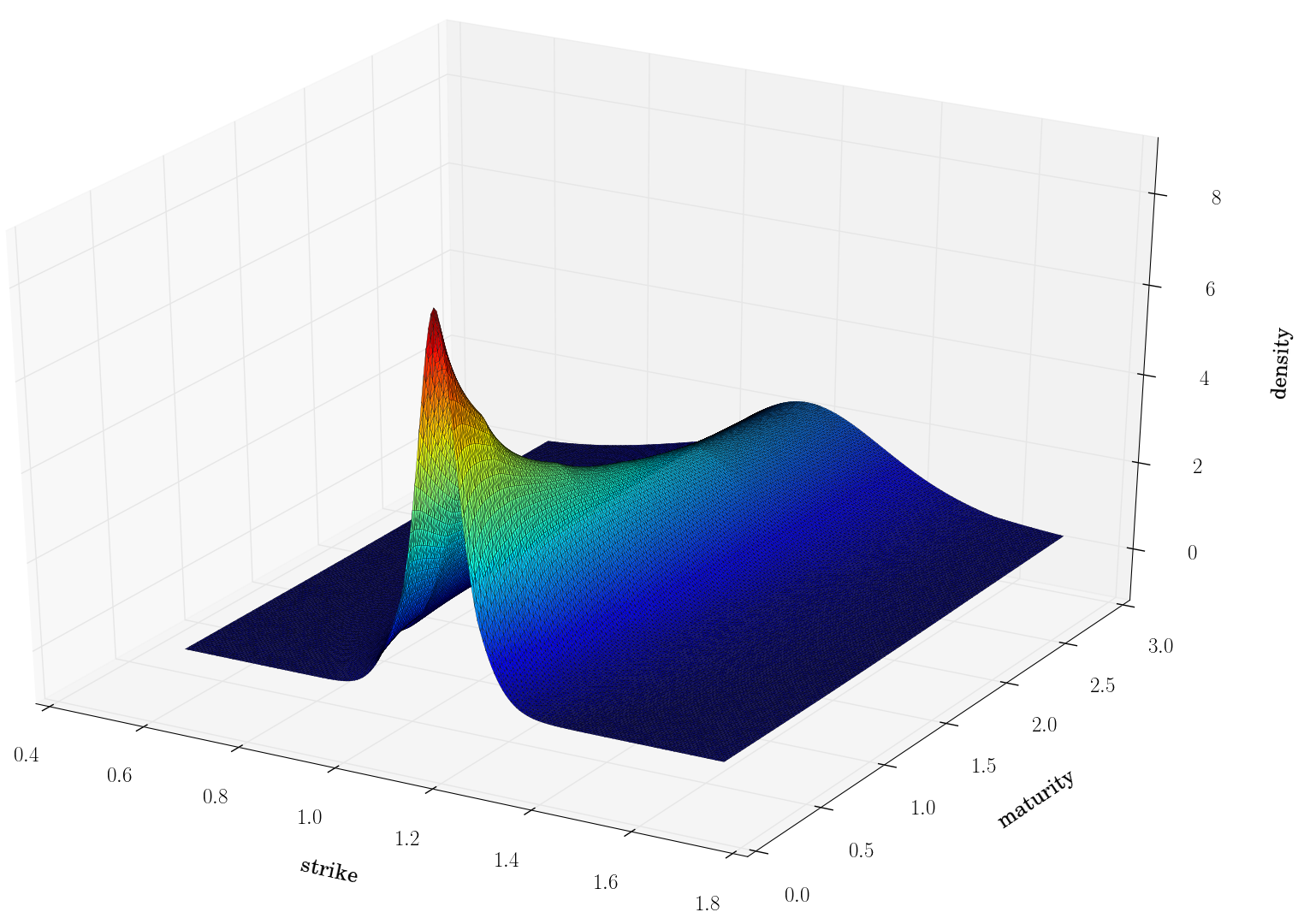}

\captionof{figure}{Market spot marginal density computed from the Dupire forward PDE with calibrated $\sigma_{LV}$}\label{fig:spotMarginalDensityPDE-1}%
\end{centering}
\end{minipage}%
\end{minipage}

%


\section{Four-factor hybrid stochastic volatility model calibration}
\label{subsec:hybridsv}
\label{sub:Heston-2CIR++-Calibration}

Consider a ``purely stochastic'' version of the model (\ref{eq:Model Definition})
-- the Heston-2CIR\scalebox{.9}{\raisebox{.5pt}{++}} model with
leverage function $\equiv1$ -- and additionally suppose that the
domestic and the foreign short interest rate dynamics are independent
of the dynamics of the spot FX rate. The model is governed by the
following system of SDEs under the domestic risk-neutral measure $\mathbb{Q}^{d}$:
\begin{align}
\begin{cases}
\cfrac{dS_{t}^{SV}}{S_{t}^{SV}}=\left(r_{t}^{d}-r_{t}^{f}\right)\,dt+\sqrt{V_{t}}\,dW_{t},\,\,S_{0}^{SV}=S_{0},\\
r_{t}^{d}=g_{t}^{d}+h^{d}\left(t\right)\\
r_{t}^{f}=g_{t}^{f}+h^{f}\left(t\right)\\
dg_{t}^{d}=\kappa_{d}\left(\theta_{d}-g_{t}^{d}\right)\,dt+\xi_{d}\sqrt{g_{t}^{d}}\,dW_{t}^{\gd}\\
dg_{t}^{f}=\kappa_{f}\left(\theta_{f}-g_{t}^{f}\right)\,dt+\xi_{f}\sqrt{g_{t}^{f}}\,dW_{t}^{\gf}\\
dV_{t}=\kappa\left(\theta-V_{t}\right)\,dt+\xi\sqrt{V_{t}}\,dW_{t}^{V}\,,
\end{cases}\label{eq5.15}
\end{align}
where $W$ and $W^{V}$ are correlated Brownian motions with correlation
coefficient $\rho$. Note that the quanto correction
term in the drift of the foreign short rate vanishes due to the postulated
independence assumption between the spot FX rate and foreign short
rate dynamics.

Define the vector of parameters $\beta_{2}=(v_{0},\kappa,\theta,\xi,\rho)$.
The next step in our calibration is to find the values of these $5$
model parameters for which European call option prices best match
the market call prices retrieved from volatility quotes for different
strikes and maturities%
. 
For a EURUSD transaction, the market standard is to choose USD as
the domestic currency and EUR as the foreign currency. The forward
FX rate for a payment date $T$ is defined as 
\begin{equation}
F_{T}=\frac{P^{f}(0,T)}{P^{d}(0,T)}\hspace{1pt}S_{0},\label{eq5.16}
\end{equation}
where $P^{d}(0,T)$ and $P^{f}(0,T)$ are the domestic and foreign
discount factors at time $0$ for a maturity $T$, respectively. 

Under the postulated simple correlation structure of the Brownian
drivers and when the short rates are driven by the CIR process, i.e.,
when $h^{d,f}=0$, Ahlip and Rutkowski \cite{Ahlip2013} derive an
efficient closed-form formula for the European call option price.
Hence, we denote by $C_{\text{A}}(K,T)$ the fair value under the
Heston--2CIR model of a European call option with strike $K$ and
maturity $T$ computed with the aforementioned formula, and by $C_{\text{H}}(K,T)$
the fair value of the same option but under the Heston--2CIR\scalebox{.9}{\raisebox{.5pt}{++}}
model. For $i\in\{d,f\}$, we define for brevity 
\[
H_{i}=\exp\left\{ \int_{0}^{T}{h^{i}(t)\,dt}\right\} ,
\]
where the shift functions $h^{d,f}$ were calibrated in 
Appendix \ref{subsec:rate}.
Then we can extend the pricing formula of Ahlip and Rutkowski
\cite{Ahlip2013} as follows. 
\begin{align*}
C_{\text{H}}(K,T) & =\E^{\mathbb{Q}^{d}}\left[\exp\left\{ -\int_{0}^{T}{r_{t}^{d}\,dt}\right\} \big(S_{T}^{SV}-K\big)^{+}\right]\\[2pt]
 & =H_{f}^{-1}\E^{\mathbb{Q}^{d}}\left[\exp\left\{ -\int_{0}^{T}{g_{t}^{d}\,dt}\right\} \big(H_{f}H_{d}^{-1}S_{T}^{SV}-H_{f}H_{d}^{-1}K\big)^{+}\right].
\end{align*}
Therefore, $C_{\text{H}}(K,T)=H_{f}^{-1}C_{\text{A}}(\tilde{K},T)$,
where $\tilde{K}=H_{f}H_{d}^{-1}K$. We now calibrate the Heston--2CIR\scalebox{.9}{\raisebox{.5pt}{++}}
model by minimising the sum of the squared differences between model
and market call prices: 
\begin{equation}
\min_{\beta_{2}\in\mathbb{R}_{\scalebox{.5}{\raisebox{.5pt}{+}}}^{4}\times[-1,\hspace{0.5pt}1]}\hspace{2pt}\sum_{\substack{1\leq i\leq n\\[1pt]
1\leq j\leq m
}
}{\big[C_{\text{H}}(K_{j},T_{i};\beta_{2})-C_{\text{BS}}(K_{j},T_{i},\sigma_{i,j})\big]^{2}},\label{eq5.18}
\end{equation}
where $\sigma_{i,j}$ is the quoted volatility corresponding to a
strike $K_{j}$ and a maturity $T_{i}$, for $j=1,\hdots,m$ and $i=1,\hdots,n$.
There are many ways to choose the objective function (error measure)
in (\ref{eq5.18}). For instance, we may consider either call prices
or Black--Scholes implied volatilities and minimise the sum of either
absolute or relative (squared) differences between model and market
values, using either uniform or non-uniform weights. We choose this
particular error measure, which assigns more weight to more expensive
options (in-the-money, long-term) and less weight to cheaper options
(out-of-the-money, short-term), for two reasons. First, the Heston
model, and hence the Heston--2CIR\scalebox{.9}{\raisebox{.5pt}{++}}
model by extension, cannot reproduce the smiles or skews typically
observed for short maturities that well and a more careful calibration
to these smiles would result in a larger overall model error due to
the inherent poor fit of the model to the short-term. Second, market
data becomes scarce as the maturity increases, and hence we already
assigned more weight to the short- and mid-term sections of the volatility
surface; for instance, we have more maturities up to $1$ year than
between $1$ and $5$ years.

As before, we employ a nonlinear least-squares solver (the trust-region-reflective
algorithm, see \cite{coleman1994convergence}) for the calibration and a global optimisation method (a
genetic algorithm) for verification purposes. Due to the non-linearity
and non-convexity of the problem, the calibrated model parameters
may end up in a local rather than a global minimum of the objective
function. Hence, a good initial parameter guess may significantly
improve the quality of the calibration. Practitioners usually use
variance swap prices to calibrate $v_{0}$, $\kappa$ and $\theta$.
In our case, we found the squared ATM $3$-week and $5$-year volatilities
to provide good initial guesses for $v_{0}$ and $\theta$, respectively.


\fi

\end{document}